\documentclass[submission,copyright,creativecommons]{eptcs}
\providecommand{\event}{QPL 2024}

\usepackage{graphicx}
\usepackage{stfloats}
\usepackage{placeins}
\usepackage{float}
\usepackage{amsthm}
\usepackage{amsmath}
\usepackage{booktabs}
\usepackage{cite}
\usepackage[bibliography=common]{apxproof}

\renewcommand{\appendixsectionformat}[2]{Proofs of Theorems\label{sec:proofs}}

\newtheoremrep{theorem}{Theorem}
\newtheoremrep{lemma}{Lemma}
\newtheorem{definition}{Definition}
\usepackage[utf8]{inputenc}
\usepackage[linesnumbered,ruled]{algorithm2e}

\usepackage{tikzit}
\input{quantum.tikzdefs}

\tikzstyle{gate}=[shape=rectangle, text height=1.5ex, text depth=0.25ex, yshift=0.5mm, fill=white, draw=black, minimum height=3mm, yshift=-0.5mm, minimum width=3mm, font={\small}, tikzit category=circuit, inner sep=2pt]
\tikzstyle{big gate}=[shape=rectangle, text height=1.5ex, text depth=0.25ex, yshift=0.5mm, fill=white, draw=black, minimum height=10mm, yshift=-0.5mm, minimum width=5mm, font={\small}, tikzit category=circuit]
\tikzstyle{Z dot}=[inner sep=0mm, minimum size=2mm, shape=circle, draw=black, fill={rgb,255: red,221; green,255; blue,221}, tikzit category=zx]
\tikzstyle{Z phase dot}=[minimum size=5mm, font={\footnotesize\boldmath}, shape=rectangle, rounded corners=2mm, inner sep=0.2mm, outer sep=-2mm, scale=0.8, tikzit shape=circle, draw=black, fill={rgb,255: red,221; green,255; blue,221}, tikzit draw=blue, tikzit category=zx]
\tikzstyle{X dot}=[Z dot, shape=circle, draw=black, fill={rgb,255: red,255; green,136; blue,136}, tikzit category=zx]
\tikzstyle{X phase dot}=[Z phase dot, tikzit shape=circle, tikzit draw=blue, fill={rgb,255: red,255; green,136; blue,136}, font={\footnotesize\boldmath}, tikzit category=zx]
\tikzstyle{hadamard}=[fill=yellow, draw=black, shape=rectangle, inner sep=0.6mm, minimum height=1.5mm, minimum width=1.5mm, tikzit category=zx]
\tikzstyle{paulibox}=[fill={rgb,255: red,221; green,221; blue,255}, draw=black, shape=rectangle, inner sep=0.6mm, minimum height=5mm, minimum width=5mm, font={\footnotesize}, text height=1.5ex, text depth=0.25ex, tikzit category=zx]
\tikzstyle{vertex}=[inner sep=0mm, minimum size=1mm, shape=circle, draw=black, fill=black, tikzit category=misc]
\tikzstyle{vertex set}=[inner sep=0mm, minimum size=1mm, shape=circle, draw=black, fill=white, font={\footnotesize\boldmath}, tikzit category=misc]
\tikzstyle{small black dot}=[fill=black, draw=black, shape=circle, inner sep=0pt, minimum width=1.2mm, tikzit category=circuit]
\tikzstyle{cnot ctrl}=[fill=black, draw=black, shape=circle, inner sep=0pt, minimum width=1.2mm, tikzit category=circuit]
\tikzstyle{cnot targ}=[fill=white, draw=white, shape=circle, tikzit category=circuit, label={center:$\oplus$}, inner sep=0pt, minimum width=2.1mm, tikzit fill={rgb,255: red,102; green,204; blue,255}, tikzit draw=black]
\tikzstyle{ket}=[fill=white, draw=black, shape=regular polygon, regular polygon sides=3, regular polygon rotate=-30, scale=0.7, inner sep=1pt, tikzit category=circuit, tikzit shape=rectangle, tikzit fill=green]
\tikzstyle{bra}=[fill=white, draw=black, shape=regular polygon, regular polygon sides=3, regular polygon rotate=30, scale=0.7, inner sep=1pt, tikzit category=circuit, tikzit shape=rectangle, tikzit fill=red]
\tikzstyle{scalar}=[shape=rectangle, text height=1.5ex, text depth=0.25ex, yshift=0.5mm, fill=white, draw=black, minimum height=5mm, yshift=-0.5mm, minimum width=5mm, font={\small}]
\tikzstyle{clabel}=[fill=white, draw=none, shape=rectangle, tikzit fill={rgb,255: red,56; green,255; blue,242}, font={\footnotesize}, inner sep=1pt, tikzit category=labels]
\tikzstyle{empty diagram}=[draw={gray!40!white}, dashed, shape=rectangle, minimum width=1cm, minimum height=1cm, tikzit category=misc]
\tikzstyle{amap}=[fill=white, draw=black, shape=NEbox, tikzit category=asymmetric, tikzit fill=yellow, tikzit shape=rectangle]
\tikzstyle{amap conj}=[fill=white, draw=black, shape=NWbox, tikzit category=asymmetric, tikzit fill=green, tikzit shape=rectangle]
\tikzstyle{amap adj}=[fill=white, draw=black, shape=SEbox, tikzit category=asymmetric, tikzit fill=red, tikzit shape=rectangle]
\tikzstyle{amap trans}=[fill=white, draw=black, shape=SWbox, tikzit category=asymmetric, tikzit fill=orange, tikzit shape=rectangle]
\tikzstyle{astate}=[fill=white, draw=black, shape=NEtriangle, tikzit category=asymmetric, tikzit shape=circle, tikzit fill=yellow]
\tikzstyle{astate conj}=[fill=white, draw=black, shape=NWtriangle, tikzit category=asymmetric, tikzit shape=circle, tikzit fill=green]
\tikzstyle{astate adj}=[fill=white, draw=black, shape=SEtriangle, tikzit category=asymmetric, tikzit shape=circle, tikzit fill=red]
\tikzstyle{astate trans}=[fill=white, draw=black, shape=SWtriangle, tikzit category=asymmetric, tikzit shape=circle, tikzit fill=orange]
\tikzstyle{white dot}=[inner sep=0mm, minimum size=2mm, shape=circle, draw=black, fill={rgb,255: red,250; green,250; blue,250}]
\tikzstyle{white phase dot}=[minimum size=5mm, font={\footnotesize\boldmath}, shape=rectangle, rounded corners=2mm, inner sep=0.2mm, outer sep=-2mm, scale=0.8, tikzit shape=circle, draw=black, fill={rgb,255: red,250; green,250; blue,250}, tikzit draw=blue]
\tikzstyle{hbox}=[shape=rectangle, text height=2mm, fill={rgb,255: red,255; green,235; blue,61}, draw=black, minimum height=2mm, minimum width=2mm, font={\small}, tikzit category=zh, inner sep=0pt, rounded corners=0.5mm]
\tikzstyle{Z dot (zh)}=[inner sep=0mm, minimum size=2mm, shape=circle, draw=black, fill={rgb,255: red,250; green,250; blue,250}, tikzit category=zh]
\tikzstyle{X dot (zh)}=[Z dot, shape=circle, draw=black, fill={rgb,255: red,193; green,193; blue,193}, tikzit category=zh]
\tikzstyle{triangle}=[fill={rgb,255: red,255; green,136; blue,136}, draw=black, shape=isosceles triangle, isosceles triangle apex angle=60, minimum size=2.5mm, inner sep=0mm]
\tikzstyle{labelled hbox}=[shape=rectangle, text height=1.75ex, text depth=0.5ex, fill={rgb,255: red,255; green,235; blue,61}, draw=black, minimum height=3mm, minimum width=4mm, font={\small}, tikzit category=zh, inner sep=1.3pt, rounded corners=0.5mm]
\tikzstyle{Z phase dot (zh)}=[Z phase dot, tikzit shape=circle, tikzit draw=blue, fill={rgb,255: red,250; green,250; blue,250}, font={\footnotesize\boldmath}, tikzit category=zh]
\tikzstyle{X phase dot (zh)}=[Z phase dot, tikzit shape=circle, tikzit draw=blue, fill={rgb,255: red,193; green,193; blue,193}, font={\footnotesize\boldmath}, tikzit category=zh]
\tikzstyle{W node}=[fill=black, draw=black, shape=regular polygon, regular polygon sides=3, minimum size=2mm]
\tikzstyle{Z dot (zw)}=[fill=white, draw=black, shape=circle, minimum width=1.2mm, inner sep=0pt]
\tikzstyle{Z phase dot XL}=[Z phase dot, fill={rgb,255: red,250; green,250; blue,250}, draw=black, shape=circle, tikzit draw={rgb,255: red,191; green,0; blue,64}, tikzit shape=circle, font={\large\boldmath}, inner sep=0.0mm]

\tikzstyle{hadamard edge}=[-, dashed, dash pattern=on 2pt off 0.5pt, thick, draw={rgb,255: red,68; green,136; blue,255}]
\tikzstyle{box edge}=[-, dashed, dash pattern=on 2pt off 0.5pt, thick, draw={rgb,255: red,203; green,192; blue,225}]
\tikzstyle{brace edge}=[-, tikzit draw=blue, decorate, decoration={brace,amplitude=1mm,raise=-1mm}]
\tikzstyle{diredge}=[->, thick]
\tikzstyle{double edge}=[-, double, shorten <=-1mm, shorten >=-1mm, double distance=2pt]
\tikzstyle{gray edge}=[-, {gray!60!white}]
\tikzstyle{pointer edge}=[->, very thick, gray]
\tikzstyle{boldedge}=[-, line width=1.0pt, shorten <=-0.17mm, shorten >=-0.17mm]
\tikzstyle{bidir edge}=[<->, very thick, draw={rgb,255: red,191; green,191; blue,191}]
\tikzstyle{purple edge}=[->, thick, draw={rgb,255: red,225; green,117; blue,216}]
\tikzstyle{green edge}=[->, thick, draw={rgb,255: red,167; green,231; blue,137}]
\tikzstyle{orange edge}=[->, thick, draw={rgb,255: red,245; green,170; blue,63}]
\tikzstyle{blue edge}=[->, thick, draw={rgb,255: red,68; green,136; blue,255}]
\tikzstyle{any edge}=[->, thick, draw=cyan]
\tikzstyle{red edge}=[->, thick, draw={rgb,255: red,255; green,136; blue,136}]
\tikzstyle{bidiredge}=[<->, thick]
\tikzstyle{dashed diredge}=[->, dashed, dash pattern=on 1pt off 0.5pt]
\tikzstyle{bidashed diredge}=[<->, dashed, dash pattern=on 1pt off 0.5pt]

\usepackage{color}

\newcommand{\Substack}[1]{\hbox to 0pt{\hss $\substack{ #1}$\hss}}

\begin{document}

\title{A Graphical \#SAT Algorithm for Formulae with Small Clause Density}

\author{Tuomas Laakkonen$^{1, a}$, Konstantinos Meichanetzidis$^{1, b}$, John van de Wetering $^{2, c}$
\institute{$^{1}$ Quantinuum, 17 Beaumont Street, Oxford OX1 2NA, United Kingdom\\
$^{2}$ Informatics Institute, University of Amsterdam, 1098 XH Amsterdam, The Netherlands}
\email{$^{a, b}$ \{tuomas.laakkonen, k.mei\}@quantinuum.com, $^{c}$ john@vdwetering.name}
}

\def\titlerunning{A Graphical \#SAT Algorithm for Formulae with Small Clause Density}
\def\authorrunning{T. Laakkonen, K. Meichanetzidis \& J. van de Wetering}

\begin{toappendix}
    \section{Rewriting Rules}
    \label{sec:zhrules}

    The ZH-calculus \cite{backens_zh_2019} has the following rewriting rules:
    \begin{equation*}
        \input{./figures/qpl-zhcalculusrules.tikz}
    \end{equation*}
    We will also used some derived rules. In particular, we have that
    \begin{equation}
        \label{satembed_zerostate}
        \begin{tikzpicture}
	\begin{pgfonlayer}{nodelayer}
		\node [style=labelled hbox] (0) at (-6.25, 7.75) {$0$};
		\node [style=none] (1) at (-6.25, 9.5) {};
		\node [style=none] (2) at (-4.75, 8) {$\stackrel{A}{=}$};
		\node [style=none] (3) at (-2.75, 9.5) {};
		\node [style={Z dot (zh)}] (4) at (-2.75, 8.5) {};
		\node [style=labelled hbox] (5) at (-3.25, 7.5) {$1$};
		\node [style=hbox] (6) at (-2.25, 7.5) {};
		\node [style={X dot (zh)}] (7) at (-2.75, 6.5) {$\pi$};
		\node [style=none] (8) at (-0.75, 8) {$\stackrel{SF_H}{=}$};
		\node [style=none] (9) at (2.5, 9.5) {};
		\node [style={Z dot (zh)}] (10) at (2.5, 8.5) {};
		\node [style=hbox] (12) at (3, 7.5) {};
		\node [style={X dot (zh)}] (13) at (2.5, 6.5) {$\pi$};
		\node [style=hbox] (14) at (2, 7.5) {};
		\node [style=hbox] (15) at (1.25, 7.5) {};
		\node [style=labelled hbox] (16) at (0.5, 7.5) {$1$};
		\node [style=none] (17) at (4.25, 8) {$\stackrel{U}{=}$};
		\node [style=none] (18) at (7.25, 9.5) {};
		\node [style={Z dot (zh)}] (19) at (7.25, 8.5) {};
		\node [style=hbox] (20) at (7.75, 7.5) {};
		\node [style={X dot (zh)}] (21) at (7.25, 6.5) {$\pi$};
		\node [style=hbox] (22) at (6.75, 7.5) {};
		\node [style=hbox] (23) at (6, 7.5) {};
		\node [style={Z dot (zh)}] (24) at (5.25, 7.5) {};
		\node [style=none] (25) at (9, 7.75) {$=$};
		\node [style=none] (26) at (11.5, 9.5) {};
		\node [style={Z dot (zh)}] (27) at (11.5, 8.5) {};
		\node [style=hbox] (28) at (12, 7.5) {};
		\node [style={X dot (zh)}] (29) at (11.5, 6.5) {$\pi$};
		\node [style=hbox] (30) at (11, 7.5) {};
		\node [style={X dot (zh)}] (32) at (10.25, 7.5) {};
		\node [style=none] (33) at (13.25, 8) {$\stackrel{BA_H}{=}$};
		\node [style=none] (34) at (15.5, 9.5) {};
		\node [style={Z dot (zh)}] (35) at (15.5, 8.5) {};
		\node [style=hbox] (36) at (16, 7.5) {};
		\node [style={X dot (zh)}] (37) at (15.5, 6.5) {$\pi$};
		\node [style={Z dot (zh)}] (38) at (14.75, 8.25) {};
		\node [style={Z dot (zh)}] (39) at (14.75, 6.75) {};
		\node [style=none] (40) at (17.75, 8) {$\stackrel{SF_Z}{=}$};
		\node [style=none] (41) at (19.75, 9.5) {};
		\node [style={Z dot (zh)}] (42) at (19.75, 8.5) {};
		\node [style=hbox] (43) at (20.25, 7.5) {};
		\node [style={X dot (zh)}] (44) at (19.75, 6.5) {};
		\node [style={Z dot (zh)}] (46) at (19, 6.75) {};
		\node [style={X dot (zh)}] (47) at (20.5, 6.5) {$\pi$};
		\node [style=none] (48) at (-4.75, 3.5) {$\stackrel{BA_Z}{=}$};
		\node [style=none] (49) at (-2.5, 5) {};
		\node [style={Z dot (zh)}] (50) at (-2.75, 4) {};
		\node [style=hbox] (51) at (-2.25, 3) {};
		\node [style={Z dot (zh)}] (53) at (-3.5, 2.25) {};
		\node [style={X dot (zh)}] (54) at (-2, 2) {$\pi$};
		\node [style={Z dot (zh)}] (55) at (-2.75, 2) {};
		\node [style=none] (56) at (-1, 3.5) {$\stackrel{I_Z}{=}$};
		\node [style=none] (57) at (1.25, 5) {};
		\node [style=hbox] (59) at (1.5, 3) {};
		\node [style={Z dot (zh)}] (60) at (0.25, 2.25) {};
		\node [style={X dot (zh)}] (61) at (1.75, 2) {$\pi$};
		\node [style={Z dot (zh)}] (62) at (1, 2) {};
		\node [style=none] (63) at (4.25, 5) {};
		\node [style={X dot (zh)}] (64) at (4.25, 3.5) {};
		\node [style=none] (68) at (2.75, 3.5) {$=$};
		\node [style=labelled hbox] (69) at (1.25, 8.75) {$\frac{1}{4}$};
		\node [style=labelled hbox] (70) at (6, 8.75) {$\frac{1}{4}$};
		\node [style=labelled hbox] (71) at (-3.75, 8.75) {$\frac{1}{2}$};
		\node [style=labelled hbox] (72) at (10.25, 8.75) {$\frac{1}{2}$};
		\node [style=labelled hbox] (73) at (14, 9) {$\frac{1}{2}$};
		\node [style=labelled hbox] (74) at (19.25, 7.75) {$\frac{1}{2}$};
		\node [style=labelled hbox] (75) at (-3.25, 3.25) {$\frac{1}{2}$};
		\node [style=labelled hbox] (76) at (0.5, 3.75) {$\frac{1}{2}$};
	\end{pgfonlayer}
	\begin{pgfonlayer}{edgelayer}
		\draw (1.center) to (0);
		\draw [bend right=15] (4) to (5);
		\draw [bend right=15] (5) to (7);
		\draw [bend right=15] (7) to (6);
		\draw [bend right=15] (6) to (4);
		\draw (4) to (3.center);
		\draw [bend right=15] (13) to (12);
		\draw [bend right=15] (12) to (10);
		\draw (10) to (9.center);
		\draw [bend left=15] (14) to (10);
		\draw [bend right=15] (14) to (13);
		\draw (14) to (15);
		\draw (16) to (15);
		\draw [bend right=15] (21) to (20);
		\draw [bend right=15] (20) to (19);
		\draw (19) to (18.center);
		\draw [bend left=15] (22) to (19);
		\draw [bend right=15] (22) to (21);
		\draw (22) to (23);
		\draw (24) to (23);
		\draw [bend right=15] (29) to (28);
		\draw [bend right=15] (28) to (27);
		\draw (27) to (26.center);
		\draw [bend left=15] (30) to (27);
		\draw [bend right=15] (30) to (29);
		\draw (32) to (30);
		\draw [bend right=15] (37) to (36);
		\draw [bend right=15] (36) to (35);
		\draw (35) to (34.center);
		\draw (39) to (37);
		\draw (38) to (35);
		\draw [bend right=15] (44) to (43);
		\draw [bend right=15] (43) to (42);
		\draw (42) to (41.center);
		\draw (46) to (44);
		\draw (47) to (44);
		\draw [bend right=15] (51) to (50);
		\draw (50) to (49.center);
		\draw [bend right=15] (53) to (51);
		\draw (55) to (54);
		\draw [bend right=15] (60) to (59);
		\draw (62) to (61);
		\draw [bend right=15] (59) to (57.center);
		\draw (64) to (63.center);
	\end{pgfonlayer}
\end{tikzpicture}
    \end{equation}
    and also that:
    \begin{equation}
        \label{satembed_copy}
        \begin{tikzpicture}
	\begin{pgfonlayer}{nodelayer}
		\node [style={X dot (zh)}] (0) at (-2, 3.25) {};
		\node [style={Z dot (zh)}] (1) at (-1, 3.25) {};
		\node [style=none] (2) at (0, 4) {};
		\node [style=none] (3) at (0, 2.5) {};
		\node [style=none] (4) at (0, 3.5) {$\vdots$};
		\node [style=none] (5) at (1, 3.5) {$\stackrel{BA_Z}{=}$};
		\node [style={X dot (zh)}] (6) at (2.5, 4) {};
		\node [style=none] (8) at (3.5, 4) {};
		\node [style=none] (9) at (3.5, 2.5) {};
		\node [style=none] (10) at (3.5, 3.5) {$\vdots$};
		\node [style={X dot (zh)}] (12) at (2.5, 2.5) {};
		\node [style={X dot (zh)}] (13) at (-2, -0.5) {$\pi$};
		\node [style={Z dot (zh)}] (14) at (-1, -0.5) {};
		\node [style=none] (15) at (0, 0.25) {};
		\node [style=none] (16) at (0, -1.25) {};
		\node [style=none] (17) at (0, -0.25) {$\vdots$};
		\node [style=none] (18) at (1, -0.5) {$=$};
		\node [style={Z dot (zh)}] (20) at (3.5, -0.5) {};
		\node [style=none] (21) at (4.5, 0.25) {};
		\node [style=none] (22) at (4.5, -1.25) {};
		\node [style=none] (23) at (4.5, -0.25) {$\vdots$};
		\node [style=hbox] (24) at (3, -0.5) {};
		\node [style=hbox] (25) at (2, -0.5) {};
		\node [style={Z dot (zh)}] (26) at (2.5, -0.5) {};
		\node [style=labelled hbox] (27) at (2.5, 0.5) {$\frac{1}{2}$};
		\node [style={Z dot (zh)}] (28) at (8.25, -0.5) {};
		\node [style=none] (29) at (9.25, 0.25) {};
		\node [style=none] (30) at (9.25, -1.25) {};
		\node [style=none] (31) at (9.25, -0.25) {$\vdots$};
		\node [style=hbox] (32) at (7.75, -0.5) {};
		\node [style=hbox] (33) at (7.25, -0.5) {};
		\node [style=labelled hbox] (35) at (7.75, 0.5) {$\frac{1}{2}$};
		\node [style=none] (36) at (5.75, -0.25) {$\stackrel{I_Z}{=}$};
		\node [style={Z dot (zh)}] (37) at (13, -0.5) {};
		\node [style=none] (38) at (14.75, 0.25) {};
		\node [style=none] (39) at (14.75, -1.25) {};
		\node [style=none] (40) at (14.75, -0.25) {$\vdots$};
		\node [style=hbox] (41) at (12.25, -0.5) {};
		\node [style=hbox] (42) at (11.75, -0.5) {};
		\node [style=labelled hbox] (43) at (12.25, 0.5) {$2^{-n-1}$};
		\node [style=none] (44) at (10.5, -0.25) {$\stackrel{I_H}{=}$};
		\node [style=hbox] (46) at (13.75, 0.25) {};
		\node [style=hbox] (47) at (14.25, 0.25) {};
		\node [style=hbox] (48) at (13.75, -1.25) {};
		\node [style=hbox] (49) at (14.25, -1.25) {};
		\node [style={X dot (zh)}] (50) at (17.75, -0.5) {};
		\node [style=none] (51) at (19.25, 0.25) {};
		\node [style=none] (52) at (19.25, -1.25) {};
		\node [style=none] (53) at (19.25, -0.25) {$\vdots$};
		\node [style=hbox] (55) at (16.75, -0.5) {};
		\node [style=labelled hbox] (56) at (17.25, 0.5) {$2^{-n}$};
		\node [style=hbox] (58) at (18.75, 0.25) {};
		\node [style=hbox] (60) at (18.75, -1.25) {};
		\node [style=none] (61) at (15.75, -0.5) {$=$};
		\node [style=none] (63) at (4.5, -3) {};
		\node [style=none] (64) at (4.5, -4.5) {};
		\node [style=none] (65) at (4.5, -3.5) {$\vdots$};
		\node [style=hbox] (66) at (3.5, -3) {};
		\node [style=labelled hbox] (67) at (2.5, -3.75) {$2^{-n}$};
		\node [style=hbox] (68) at (4, -3) {};
		\node [style=hbox] (69) at (4, -4.5) {};
		\node [style=hbox] (70) at (3.5, -4.5) {};
		\node [style=none] (71) at (1, -3.5) {$\stackrel{BA_H}{=}$};
		\node [style=none] (72) at (5.75, -3.5) {$\stackrel{I_Z}{=}$};
		\node [style=none] (73) at (9.75, -3) {};
		\node [style=none] (74) at (9.75, -4.5) {};
		\node [style=none] (75) at (9.75, -3.5) {$\vdots$};
		\node [style=hbox] (76) at (8.25, -3) {};
		\node [style=labelled hbox] (77) at (7.25, -3.75) {$2^{-n}$};
		\node [style=hbox] (78) at (9.25, -3) {};
		\node [style=hbox] (79) at (9.25, -4.5) {};
		\node [style=hbox] (80) at (8.25, -4.5) {};
		\node [style={Z dot (zh)}] (81) at (8.75, -3) {};
		\node [style={Z dot (zh)}] (82) at (8.75, -4.5) {};
		\node [style=none] (83) at (10.75, -3.75) {$=$};
		\node [style=none] (84) at (13, -3) {};
		\node [style=none] (85) at (13, -4.5) {};
		\node [style=none] (86) at (13, -3.5) {$\vdots$};
		\node [style={X dot (zh)}] (88) at (12, -3) {$\pi$};
		\node [style={X dot (zh)}] (89) at (12, -4.5) {$\pi$};
	\end{pgfonlayer}
	\begin{pgfonlayer}{edgelayer}
		\draw (0) to (1);
		\draw [bend left=15] (1) to (2.center);
		\draw [bend right=15] (1) to (3.center);
		\draw (6) to (8.center);
		\draw (9.center) to (12);
		\draw (13) to (14);
		\draw [bend left=15] (14) to (15.center);
		\draw [bend right=15] (14) to (16.center);
		\draw [bend left=15] (20) to (21.center);
		\draw [bend right=15] (20) to (22.center);
		\draw (26) to (24);
		\draw (24) to (20);
		\draw (26) to (25);
		\draw [bend left=15] (28) to (29.center);
		\draw [bend right=15] (28) to (30.center);
		\draw (32) to (28);
		\draw (33) to (32);
		\draw (42) to (41);
		\draw (46) to (47);
		\draw (48) to (49);
		\draw [bend left=15] (37) to (46);
		\draw [bend right=15] (37) to (48);
		\draw (37) to (41);
		\draw (47) to (38.center);
		\draw (39.center) to (49);
		\draw (58) to (51.center);
		\draw (52.center) to (60);
		\draw [bend left=15] (60) to (50);
		\draw [bend left=15] (50) to (58);
		\draw (50) to (55);
		\draw (68) to (63.center);
		\draw (64.center) to (69);
		\draw (70) to (69);
		\draw (68) to (66);
		\draw (78) to (73.center);
		\draw (74.center) to (79);
		\draw (79) to (82);
		\draw (80) to (82);
		\draw (76) to (81);
		\draw (81) to (78);
		\draw (88) to (84.center);
		\draw (85.center) to (89);
	\end{pgfonlayer}
\end{tikzpicture}
    \end{equation}
    Finally, we recall the presentation of the Z-spider as a sum over basis states:
    \begin{equation}
        \label{zspidersum}
        \begin{tikzpicture}
	\begin{pgfonlayer}{nodelayer}
		\node [style={Z phase dot (zh)}] (0) at (0, 0) {$\alpha$};
		\node [style=none] (1) at (-0.75, 1.25) {};
		\node [style=none] (2) at (0.75, 1.25) {};
		\node [style=none] (3) at (-0.75, -1.25) {};
		\node [style=none] (4) at (0.75, -1.25) {};
		\node [style=none] (5) at (0, 1.25) {\dots};
		\node [style=none] (6) at (0, -1.25) {\dots};
		\node [style=none] (8) at (2.75, 1.25) {};
		\node [style=none] (9) at (4.25, 1.25) {};
		\node [style=none] (10) at (2.75, -1.25) {};
		\node [style=none] (11) at (4.25, -1.25) {};
		\node [style=none] (12) at (3.5, 1.25) {\dots};
		\node [style=none] (13) at (3.5, -1.25) {\dots};
		\node [style=none] (15) at (8.25, 1.25) {};
		\node [style=none] (16) at (9.75, 1.25) {};
		\node [style=none] (17) at (8.25, -1.25) {};
		\node [style=none] (18) at (9.75, -1.25) {};
		\node [style=none] (19) at (9, 1.25) {\dots};
		\node [style=none] (20) at (9, -1.25) {\dots};
		\node [style=none] (21) at (5.5, 0) {$+$};
		\node [style=none] (22) at (1.5, 0) {$=$};
		\node [style={X dot (zh)}] (23) at (2.75, 0.5) {};
		\node [style={X dot (zh)}] (24) at (2.75, -0.5) {};
		\node [style={X dot (zh)}] (25) at (4.25, -0.5) {};
		\node [style={X dot (zh)}] (26) at (8.25, -0.5) {$\pi$};
		\node [style={X dot (zh)}] (27) at (9.75, -0.5) {$\pi$};
		\node [style={X dot (zh)}] (28) at (9.75, 0.5) {$\pi$};
		\node [style={X dot (zh)}] (29) at (8.25, 0.5) {$\pi$};
		\node [style={X dot (zh)}] (30) at (4.25, 0.5) {};
		\node [style=none] (31) at (7, 0.25) {$e^{i\alpha}$};
	\end{pgfonlayer}
	\begin{pgfonlayer}{edgelayer}
		\draw [bend left=15] (2.center) to (0);
		\draw [bend left=15] (0) to (1.center);
		\draw [bend right=15] (0) to (3.center);
		\draw [bend left=15] (0) to (4.center);
		\draw (26) to (17.center);
		\draw (27) to (18.center);
		\draw (28) to (16.center);
		\draw (29) to (15.center);
		\draw (25) to (11.center);
		\draw (24) to (10.center);
		\draw (23) to (8.center);
		\draw (30) to (9.center);
	\end{pgfonlayer}
\end{tikzpicture}
    \end{equation}

\section{Upper Bounds on Quantum Circuit Simulation}
\label{sec:circuit-sim}

Extending the ideas from Section \ref{sec:SAT-phases}, we can produce an upper bound on the time required to exactly compute the amplitude of a quantum computation. This task is naturally represented in terms of tensor networks \cite{orus_practical_2014}, which can be presented in terms of a graphical calculus, see \cite[page 6]{vandewetering2020zxcalculus} for an example in the ZX-calculus. Note that while this task is easily shown to be \sP-hard, it is not obviously in \sP and so it is hard to relate directly to \sSAT. By relaxing our definition of \sSAT to $\sSAT_\mathbb{C}$, we can do this translation more straightforwardly, but still apply existing knowledge about \sSAT.

Specifically, given a quantum circuit $C$ with $n$ qubits and $G$ gates representing a unitary $U$ over the CZ, Hadamard, and Z-phase gate set, the quantity $\bra{+}U\ket{+}$ can be represented by the following ZH-diagram:
\begin{equation}
    \begin{tikzpicture}
	\begin{pgfonlayer}{nodelayer}
		\node [style=none] (0) at (-2, 0) {$\langle+|U|+\rangle$};
		\node [style=none] (1) at (0, 0) {$=$};
		\node [style=labelled hbox] (2) at (1.5, 0) {$2^{-n}$};
		\node [style={Z dot (zh)}] (3) at (3, 1) {};
		\node [style={Z dot (zh)}] (4) at (3, -1) {};
		\node [style=none] (5) at (3.5, 1.25) {};
		\node [style=none] (6) at (5, 1.25) {};
		\node [style=none] (7) at (5, -1.25) {};
		\node [style=none] (8) at (3.5, -1.25) {};
		\node [style=none] (9) at (3.5, 1) {};
		\node [style=none] (10) at (3.5, -1) {};
		\node [style=none] (11) at (3, 0.25) {$\vdots$};
		\node [style={Z dot (zh)}] (12) at (5.5, 1) {};
		\node [style={Z dot (zh)}] (13) at (5.5, -1) {};
		\node [style=none] (14) at (5, 1) {};
		\node [style=none] (15) at (5, -1) {};
		\node [style=none] (16) at (5.5, 0.25) {$\vdots$};
		\node [style=none] (17) at (4.25, 0) {$C$};
	\end{pgfonlayer}
	\begin{pgfonlayer}{edgelayer}
		\draw (8.center) to (7.center);
		\draw (7.center) to (6.center);
		\draw (6.center) to (5.center);
		\draw (5.center) to (8.center);
		\draw (10.center) to (4);
		\draw (3) to (9.center);
		\draw (15.center) to (13);
		\draw (12) to (14.center);
	\end{pgfonlayer}
\end{tikzpicture}
\end{equation}
Where the portion marked $C$ is assembled from the following components
\begin{equation}
    \begin{tikzpicture}
	\begin{pgfonlayer}{nodelayer}
		\node [style=gate] (0) at (-2, 1.75) {$H$};
		\node [style=none] (1) at (-1, 1.75) {};
		\node [style=none] (2) at (-3, 1.75) {};
		\node [style=none] (3) at (0, 1.75) {$\Leftrightarrow$};
		\node [style=labelled hbox] (4) at (2, 1.75) {$2^{-1/2}$};
		\node [style=none] (5) at (3.5, 1.75) {};
		\node [style=none] (6) at (6, 1.75) {};
		\node [style=hbox] (7) at (4.75, 1.75) {};
		\node [style=cnot ctrl] (8) at (-2, 0.5) {};
		\node [style=cnot ctrl] (9) at (-2, -1) {};
		\node [style=none] (10) at (-1, 0.5) {};
		\node [style=none] (11) at (-1, -1) {};
		\node [style=none] (12) at (-3, -1) {};
		\node [style=none] (13) at (-3, 0.5) {};
		\node [style=none] (14) at (0, -0.25) {$\Leftrightarrow$};
		\node [style=none] (16) at (1, 0.5) {};
		\node [style=none] (17) at (3, 0.5) {};
		\node [style={Z dot (zh)}] (18) at (2, 0.5) {};
		\node [style=none] (19) at (1, -1) {};
		\node [style=none] (20) at (3, -1) {};
		\node [style={Z dot (zh)}] (21) at (2, -1) {};
		\node [style=hbox] (22) at (2, -0.25) {};
		\node [style=gate] (23) at (-2, -2.25) {$Z_\alpha$};
		\node [style=none] (24) at (-3, -2.25) {};
		\node [style=none] (25) at (-1, -2.25) {};
		\node [style=none] (26) at (0, -2.25) {$\Leftrightarrow$};
		\node [style=none] (27) at (1, -2.25) {};
		\node [style=none] (28) at (3, -2.25) {};
		\node [style={Z phase dot (zh)}] (29) at (2, -2.25) {$\alpha$};
		\node [style={Z dot (zh)}] (30) at (4, 1.75) {};
		\node [style={Z dot (zh)}] (31) at (5.5, 1.75) {};
	\end{pgfonlayer}
	\begin{pgfonlayer}{edgelayer}
		\draw (0) to (2.center);
		\draw (0) to (1.center);
		\draw (13.center) to (8);
		\draw (8) to (10.center);
		\draw (8) to (9);
		\draw (9) to (12.center);
		\draw (9) to (11.center);
		\draw (18) to (16.center);
		\draw (18) to (17.center);
		\draw (21) to (19.center);
		\draw (21) to (20.center);
		\draw (22) to (18);
		\draw (22) to (21);
		\draw (25.center) to (23);
		\draw (23) to (24.center);
		\draw (29) to (28.center);
		\draw (29) to (27.center);
		\draw (31) to (7);
		\draw (7) to (30);
		\draw (30) to (5.center);
		\draw (6.center) to (31);
	\end{pgfonlayer}
\end{tikzpicture}
\end{equation}
by placing them according to their connections in the quantum circuit. However, we can translate these to $\stwoSAT_\mathbb{C}$ diagrams, up to scalar factors, by applying the following rewrite:
\begin{equation}
    \begin{tikzpicture}
	\begin{pgfonlayer}{nodelayer}
		\node [style=hbox] (0) at (0, 0) {};
		\node [style=none] (1) at (1, 0) {};
		\node [style=none] (2) at (-1, 0) {};
		\node [style=none] (3) at (2, 0.25) {$\stackrel{SF_H}{=}$};
		\node [style=labelled hbox] (4) at (3.5, 0) {$\frac{1}{2}$};
		\node [style=hbox] (5) at (5.5, 0) {};
		\node [style=none] (6) at (6.5, 0) {};
		\node [style=none] (7) at (4.5, 0) {};
		\node [style=hbox] (8) at (5.5, -0.75) {};
		\node [style=hbox] (9) at (5.5, -1.5) {};
		\node [style=none] (10) at (7.5, 0.25) {$\stackrel{\ref{completeness_zh0tofgen}}{=}$};
		\node [style=none] (13) at (8.5, 0) {};
		\node [style=labelled hbox] (14) at (9.25, 0) {$0$};
		\node [style={X dot (zh)}] (15) at (10.25, 0) {$\pi$};
		\node [style={Z dot (zh)}] (16) at (11, 0) {$\pi$};
		\node [style={X dot (zh)}] (17) at (11.75, 0) {$\pi$};
		\node [style=labelled hbox] (18) at (12.75, 0) {$0$};
		\node [style=none] (19) at (13.5, 0) {};
		\node [style={X dot (zh)}] (20) at (11, -0.75) {$\pi$};
		\node [style=labelled hbox] (21) at (11, -1.75) {$0$};
		\node [style={Z dot (zh)}] (22) at (11, -2.75) {$\pi$};
		\node [style=hbox] (23) at (11, -3.5) {};
		\node [style=none] (24) at (1, -5) {$=$};
		\node [style=none] (25) at (2, -5) {};
		\node [style=labelled hbox] (26) at (2.75, -5) {$0$};
		\node [style={X dot (zh)}] (27) at (3.75, -5) {$\pi$};
		\node [style={Z dot (zh)}] (28) at (4.5, -5) {$\pi$};
		\node [style={X dot (zh)}] (29) at (5.25, -5) {$\pi$};
		\node [style=labelled hbox] (30) at (6.25, -5) {$0$};
		\node [style=none] (31) at (7, -5) {};
		\node [style={X dot (zh)}] (32) at (4.5, -5.75) {$\pi$};
		\node [style=labelled hbox] (33) at (4.5, -6.75) {$0$};
		\node [style={Z dot (zh)}] (34) at (4.5, -7.75) {};
		\node [style=none] (35) at (8, -4.75) {$\stackrel{\eqref{extension_absorb}}{=}$};
		\node [style=none] (36) at (9, -5) {};
		\node [style=labelled hbox] (37) at (9.75, -5) {$0$};
		\node [style={X dot (zh)}] (38) at (10.75, -5) {$\pi$};
		\node [style={Z phase dot (zh)}] (39) at (12.5, -5) {$\pi - i\ln(2)$};
		\node [style={X dot (zh)}] (40) at (14.25, -5) {$\pi$};
		\node [style=labelled hbox] (41) at (15.25, -5) {$0$};
		\node [style=none] (42) at (16, -5) {};
	\end{pgfonlayer}
	\begin{pgfonlayer}{edgelayer}
		\draw (1.center) to (0);
		\draw (0) to (2.center);
		\draw (6.center) to (5);
		\draw (5) to (7.center);
		\draw (9) to (8);
		\draw (8) to (5);
		\draw (18) to (17);
		\draw (17) to (16);
		\draw (16) to (15);
		\draw (15) to (14);
		\draw (14) to (13.center);
		\draw (19.center) to (18);
		\draw (20) to (16);
		\draw (21) to (20);
		\draw (23) to (22);
		\draw (22) to (21);
		\draw (30) to (29);
		\draw (29) to (28);
		\draw (28) to (27);
		\draw (27) to (26);
		\draw (26) to (25.center);
		\draw (31.center) to (30);
		\draw (32) to (28);
		\draw (33) to (32);
		\draw (34) to (33);
		\draw (41) to (40);
		\draw (40) to (39);
		\draw (39) to (38);
		\draw (38) to (37);
		\draw (37) to (36.center);
		\draw (42.center) to (41);
	\end{pgfonlayer}
\end{tikzpicture}
\end{equation}
To map $\bra{+}U\ket{+}$ into this form, we apply this translation and then spider fusion wherever possible. This leaves a $\stwoSAT_\mathbb{C}$ diagram and a scalar factor. Therefore, to compute the value of $\bra{+}U\ket{+}$, we can compute the model count of the corresponding $\stwoSAT_\mathbb{C}$ instance by $\mathrm{CDP}_\mathbb{C}$ and multiply it by the scalar factor. Since Hadamard gates and CZ gates each add two clauses, and Z-phase gates add no clauses, we have $m \leq 2G$. Hence, by applying the $O^*(1.1740^m)$ bound for \stwoSAT by Wang and Gu \cite{wang_worst_2013}, we can calculate $\bra{+}U\ket{+}$ in time $O^*(1.3783^G)$. This is better than statevector simulation whenever $G \leq 2.16n$.

It is also possible to work with circuits that contain $\textbf{C}^k\textbf{Z}$ gates directly, since we have the following translation, as before:
\begin{equation}
    \begin{tikzpicture}
	\begin{pgfonlayer}{nodelayer}
		\node [style=cnot ctrl] (0) at (-2, 2) {};
		\node [style=cnot ctrl] (1) at (-2, -2) {};
		\node [style=none] (2) at (-1, 2) {};
		\node [style=none] (3) at (-1, -2) {};
		\node [style=none] (4) at (-3, -2) {};
		\node [style=none] (5) at (-3, 2) {};
		\node [style=none] (6) at (0, 0.5) {$\Leftrightarrow$};
		\node [style=none] (7) at (1, 2) {};
		\node [style=none] (8) at (3, 2) {};
		\node [style={Z dot (zh)}] (9) at (2, 2) {};
		\node [style=none] (10) at (1, -2) {};
		\node [style=none] (11) at (3, -2) {};
		\node [style={Z dot (zh)}] (12) at (2, -2) {};
		\node [style=hbox] (13) at (2.5, 1.25) {};
		\node [style=cnot ctrl] (14) at (-2, 0.5) {};
		\node [style=none] (15) at (-1, 0.5) {};
		\node [style=none] (16) at (-3, 0.5) {};
		\node [style=none] (17) at (1, 0.5) {};
		\node [style=none] (18) at (3, 0.5) {};
		\node [style={Z dot (zh)}] (19) at (2, 0.5) {};
		\node [style=none] (20) at (-2.75, -0.5) {$\vdots$};
		\node [style=none] (21) at (1.25, -0.5) {$\vdots$};
		\node [style=none] (22) at (-1.25, -0.5) {$\vdots$};
		\node [style=none] (23) at (2.75, -0.5) {$\vdots$};
		\node [style=none] (24) at (4, 0.5) {$=$};
		\node [style=none] (25) at (5, 2) {};
		\node [style=none] (26) at (7, 2) {};
		\node [style={Z dot (zh)}] (27) at (6, 2) {};
		\node [style=none] (28) at (5, -2) {};
		\node [style=none] (29) at (7, -2) {};
		\node [style={Z dot (zh)}] (30) at (6, -2) {};
		\node [style=none] (32) at (5, 0.5) {};
		\node [style=none] (33) at (7, 0.5) {};
		\node [style={Z dot (zh)}] (34) at (6, 0.5) {};
		\node [style=none] (35) at (5.25, -0.5) {$\vdots$};
		\node [style=labelled hbox] (37) at (6.75, 1.25) {$0$};
		\node [style=labelled hbox] (38) at (6.75, -0.25) {$0$};
		\node [style=labelled hbox] (39) at (6.75, -1.25) {$0$};
		\node [style={X dot (zh)}] (40) at (7.75, 1.25) {$\pi$};
		\node [style={X dot (zh)}] (41) at (7.75, -0.25) {$\pi$};
		\node [style={X dot (zh)}] (42) at (7.75, -1.25) {$\pi$};
		\node [style={Z phase dot (zh)}] (43) at (9.5, 0) {$\pi - i\ln(2)$};
	\end{pgfonlayer}
	\begin{pgfonlayer}{edgelayer}
		\draw (5.center) to (0);
		\draw (0) to (2.center);
		\draw (1) to (4.center);
		\draw (1) to (3.center);
		\draw (9) to (7.center);
		\draw (9) to (8.center);
		\draw (12) to (10.center);
		\draw (12) to (11.center);
		\draw (13) to (9);
		\draw (14) to (16.center);
		\draw (14) to (15.center);
		\draw (14) to (1);
		\draw (14) to (0);
		\draw (19) to (17.center);
		\draw (19) to (18.center);
		\draw (13) to (19);
		\draw (13) to (12);
		\draw (27) to (25.center);
		\draw (27) to (26.center);
		\draw (30) to (28.center);
		\draw (30) to (29.center);
		\draw (34) to (32.center);
		\draw (34) to (33.center);
		\draw [bend right] (39) to (30);
		\draw [bend right] (34) to (38);
		\draw [bend left] (37) to (27);
		\draw (42) to (39);
		\draw (38) to (41);
		\draw (40) to (37);
		\draw [bend left=15] (40) to (43);
		\draw [in=0, out=-165] (43) to (41);
		\draw [bend right=15] (42) to (43);
	\end{pgfonlayer}
\end{tikzpicture}
\end{equation}
Then each gate adds at most $k + 1$ clauses to the \stwoSAT instance, leading to an overall runtime of $O^*(1.1740^{(k + 1)G})$. For instance, for $k = 2$ this is $O^*(1.6181^G)$, which is substantially better than the corresponding runtime given by decomposing CCZ gates into single- and two-qubit gates (which would be $O^*(6.8552^G)$, as each CCZ would require twelve clauses). It may be possible to give better bounds by analyzing the spider fusion that occurs in this translation and applying bounds for \stwoSAT in terms of variables rather than clauses, but we postpone this to future work. 

\section{A Non-diagrammatic Argument for $\mathrm{CDP}_\pm^{\to 2}$}
\label{app:nodiagrams}

Suppose we have a function $\phi : \mathbb{B}^n \to \mathbb{B}$ given by
$$\phi(x_1, \dots, x_n) = \bigwedge_{i = 1}^{m} (c_{i1} \lor c_{i2} \lor \cdots \lor c_{ik_i})$$
then we can define $\phi' : \mathbb{B}^{n+m} \to \{-1, 0, 1\}$ given by:
$$\phi'(x_1, \dots, x_n, y_1 \dots, y_m) = (-1)^{\sum_i y_i}\left[\bigwedge_{i = 1}^{m} \bigwedge_{j = 1}^{k_i} (\lnot c_{ij} \lor \lnot y_{i})\right] = (-1)^{\sum_i y_i}\left[\bigwedge_{i = 1}^{m} \lnot \left(y_{i} \land \bigvee_{j = 1}^{k_i} c_{ij}\right)\right]$$
Lifting Boolean logic to integer arithmetic, we have that
\begin{align*}
    \phi'(x_1, \dots, x_n, y_1 \dots, y_m) &= (-1)^{\sum_i y_i} \prod_{i = 1}^{m} \left(1 - y_i \left[\bigvee_{j=1}^{k_i}c_{ij}\right]\right) = (-1)^{\sum_i y_i} \sum_{S \subseteq [1, m]} \prod_{i \in S} (-y_i)\left[\bigvee_{j=1}^{k_i}c_{ij}\right] \\
    &= \sum_{S \subseteq [1, m]} (-1)^{|S| + \sum_{i} y_i} \left[\bigwedge_{i \in S} y_i\right] \prod_{i \in S}\left[\bigvee_{j=1}^{k_i}c_{ij}\right] \\
    &=  \sum_{S \subseteq [1, m]} (-1)^{\sum_{i \notin S} y_i} \left[\bigwedge_{i \in S} y_i\right] \prod_{i \in S}\left[\bigvee_{j=1}^{k_i}c_{ij}\right]
\end{align*}
and summing over all possibilities for $y_i$,
\begin{align*}
    \sum_{y_1, \dots, y_m \in \mathbb{B}} \phi'(x_1, \dots, x_n, y_1 \dots, y_m) &= \sum_{S \subseteq [1, m]} \prod_{i \in S}\left[\bigvee_{j=1}^{k_i}c_{ij}\right] \sum_{y_1, \dots, y_m \in \mathbb{B}} (-1)^{\sum_{i \notin S} y_i} \left[\bigwedge_{i \in S} y_i\right] \\ 
    &= \sum_{S \subseteq [1, m]} \prod_{i \in S}\left[\bigvee_{j=1}^{k_i}c_{ij}\right] \cdot \begin{cases} 1 & |S| = m \\ 0 & |S| < m\end{cases} = \left[\bigwedge_{i=1}^m \bigvee_{j=1}^{k_i} c_{ij}\right] = \phi(x_1, \dots, x_n)
\end{align*}
therefore, to compute $\#(\phi)$, it is sufficient to sum over all values of $\phi'$:
$$\#(\phi) = \sum_{x_1, \dots, x_n \in \mathbb{B}} \phi(x_1, \dots, x_n) = \sum_{x_1, \dots, x_n, y_1, \dots, y_m \in \mathbb{B}} \phi'(x_1, \dots, x_n, y_1, \dots, y_m)$$
Finally, we can see that
\begin{align*}
    \#(\phi) &= \sum_{x_1, \dots, x_n, y_1, \dots, y_m \in \mathbb{B}}  (-1)^{\sum_i y_i}\left[\bigwedge_{i = 1}^{m} \bigwedge_{j = 1}^{k_i} (\lnot c_{ij} \lor \lnot y_{i})\right] \\
    &= \sum_{\substack{x_1, \dots, x_n, y_1, \dots, y_m \in \mathbb{B} \\ \sum_i y_i \text{ even}}}  \left[\bigwedge_{i = 1}^{m} \bigwedge_{j = 1}^{k_i} (\lnot c_{ij} \lor \lnot y_{i})\right] - \sum_{\substack{x_1, \dots, x_n, y_1, \dots, y_m \in \mathbb{B} \\ \sum_i y_i \text{ odd}}} \left[\bigwedge_{i = 1}^{m} \bigwedge_{j = 1}^{k_i} (\lnot c_{ij} \lor \lnot y_{i})\right] \\
    &= \sSAT_\pm\left(\bigwedge_{i = 1}^{m} \bigwedge_{j = 1}^{k_i} (\lnot c_{ij} \lor \lnot y_{i}), \{y_1, \dots, y_m\}\right)
\end{align*}
by definition, and thus this defines a reduction from \sSAT to $\stwoSAT_\pm$ which can be computed using Algorithm \ref{dpllpm}. 

The argument that Wahlstr\"om's bound \cite{wahlstrom_tighter_2008} can be applied to Algorithm \ref{dpllpm} is the same as in Section \ref{subsec:reduction}: since the structure of the algorithm remains the same (regardless of scalar factors, unit propagation and branching still generate the same structures), and because unrelated instances combine multiplicatively in that
$$\sSAT_\pm(f_1(x_1, \dots, x_n) \land f_2(y_1, \dots, y_m), N_1 \cup N_2) = \sSAT_\pm(f_1(x_1, \dots, x_n), N_1)\sSAT_\pm(f_2(y_1, \dots, y_m), N_2)$$
then the bound doesn't distinguish between Algorithm \ref{cdp_algo} and \ref{dpllpm}, so it applies directly.

\section{Additional Lemmas}

\begin{lemma}
    \label{satembed_postsel}
    The following diagrammatic equivalence holds:
    \begin{equation}
        \begin{tikzpicture}
	\begin{pgfonlayer}{nodelayer}
		\node [style=none] (0) at (-4.75, 0.75) {};
		\node [style=none] (1) at (-4.75, -0.75) {};
		\node [style={X dot (zh)}] (2) at (-4, 0.75) {$\pi$};
		\node [style={X dot (zh)}] (3) at (-4, -0.75) {$\pi$};
		\node [style={X dot (zh)}] (4) at (-1.5, 0) {$\pi$};
		\node [style=hbox] (5) at (-3, 0) {};
		\node [style=hbox] (6) at (-2.25, 0) {};
		\node [style=none] (7) at (-4.5, 0.25) {$\vdots$};
		\node [style=none] (8) at (0, 0) {$=$};
		\node [style=none] (9) at (1.25, 0.75) {};
		\node [style=none] (10) at (1.25, -0.75) {};
		\node [style=labelled hbox] (11) at (2.25, 0.75) {$0$};
		\node [style=labelled hbox] (12) at (2.25, -0.75) {$0$};
		\node [style=none] (13) at (1.5, 0.25) {$\vdots$};
	\end{pgfonlayer}
	\begin{pgfonlayer}{edgelayer}
		\draw (4) to (6);
		\draw (6) to (5);
		\draw [bend right] (5) to (2);
		\draw [bend left] (5) to (3);
		\draw (3) to (1.center);
		\draw (0.center) to (2);
		\draw (12) to (10.center);
		\draw (9.center) to (11);
	\end{pgfonlayer}
\end{tikzpicture}
    \end{equation}
\end{lemma}

\begin{proof}
    This follows from Lemma 2.28 in \cite{backens_completeness_2021} and Equation~\eqref{satembed_zerostate}.
\end{proof}

\begin{lemma}
    \label{completeness_zh0tof}
    The following diagram equivalence holds:
    \begin{equation}
        \begin{tikzpicture}
	\begin{pgfonlayer}{nodelayer}
		\node [style={X dot (zh)}] (0) at (-2.75, -0.75) {$\pi$};
		\node [style={X dot (zh)}] (1) at (-2.75, 0.75) {$\pi$};
		\node [style=labelled hbox] (2) at (-1.75, -0.75) {$0$};
		\node [style=labelled hbox] (3) at (-1.75, 0.75) {$0$};
		\node [style={Z dot (zh)}] (4) at (-0.75, 0) {$\pi$};
		\node [style=labelled hbox] (6) at (1.25, 0) {$0$};
		\node [style={Z dot (zh)}] (7) at (2.25, 0) {$\pi$};
		\node [style=none] (8) at (3, 0) {};
		\node [style=none] (9) at (-3.5, 0.75) {};
		\node [style=none] (10) at (-3.5, -0.75) {};
		\node [style=none] (13) at (4, 0) {$=$};
		\node [style=hbox] (14) at (6, 0) {};
		\node [style=hbox] (15) at (7, 0) {};
		\node [style=none] (16) at (8, 0) {};
		\node [style=none] (17) at (5, -0.75) {};
		\node [style=none] (18) at (5, 0.75) {};
		\node [style={X dot (zh)}] (19) at (0.25, 0) {$\pi$};
		\node [style=labelled hbox] (20) at (0.25, 1.25) {2};
	\end{pgfonlayer}
	\begin{pgfonlayer}{edgelayer}
		\draw (10.center) to (0);
		\draw (0) to (2);
		\draw (3) to (1);
		\draw (1) to (9.center);
		\draw (6) to (7);
		\draw (7) to (8.center);
		\draw [bend right] (14) to (18.center);
		\draw [bend left] (14) to (17.center);
		\draw (14) to (15);
		\draw (15) to (16.center);
		\draw [bend left] (4) to (2);
		\draw [bend right] (4) to (3);
		\draw (19) to (4);
		\draw (6) to (19);
	\end{pgfonlayer}
\end{tikzpicture}
    \end{equation}
    This is translated from a quantum circuit identity of Ng and Wang \cite{ng_completeness_2018} and is a generalization of the rules $\mathrm{HT}$ and $\mathrm{BW}$ from the $\Delta\mathrm{ZX}_\pi$-calculus \cite{vilmart_zx-calculus_2019}, a system for describing tensor networks that is closely related to ZH-calculus.
\end{lemma}

\begin{proof}
    This can be verified by concrete calculation of the matrices.
\end{proof}

\begin{lemma}
    \label{completeness_zh0tofident}
    The following diagram equivalence holds:
    \begin{equation}
        \begin{tikzpicture}
	\begin{pgfonlayer}{nodelayer}
		\node [style=none] (0) at (-8, 0) {};
		\node [style={X dot (zh)}] (1) at (-7.25, 0) {$\pi$};
		\node [style={Z dot (zh)}] (2) at (-5.25, 0) {$\pi$};
		\node [style=labelled hbox] (3) at (-6.25, 0) {$0$};
		\node [style={X dot (zh)}] (4) at (-4.25, 0) {$\pi$};
		\node [style={Z dot (zh)}] (5) at (-2.25, 0) {$\pi$};
		\node [style=labelled hbox] (6) at (-3.25, 0) {$0$};
		\node [style=none] (7) at (-1.5, 0) {};
		\node [style=none] (8) at (0, 0) {$=$};
		\node [style=none] (9) at (3.5, 0) {};
		\node [style=none] (10) at (1.5, 0) {};
	\end{pgfonlayer}
	\begin{pgfonlayer}{edgelayer}
		\draw (0.center) to (1);
		\draw (1) to (3);
		\draw (3) to (2);
		\draw (2) to (4);
		\draw (4) to (6);
		\draw (6) to (5);
		\draw (5) to (7.center);
		\draw (10.center) to (9.center);
	\end{pgfonlayer}
\end{tikzpicture}
    \end{equation}
\end{lemma}

\begin{proof}
    This can be verified by concrete calculation of the matrices.
\end{proof}

\begin{lemma}
    \label{completeness_zh0tofgen}
    The following diagram equivalence holds for all $n \geq 0$:
    \begin{equation}
        \begin{tikzpicture}
	\begin{pgfonlayer}{nodelayer}
		\node [style={X dot (zh)}] (0) at (-2.75, -0.75) {$\pi$};
		\node [style={X dot (zh)}] (1) at (-2.75, 0.75) {$\pi$};
		\node [style=labelled hbox] (2) at (-1.75, -0.75) {$0$};
		\node [style=labelled hbox] (3) at (-1.75, 0.75) {$0$};
		\node [style={Z dot (zh)}] (4) at (-0.75, 0) {$\pi$};
		\node [style=labelled hbox] (6) at (1.25, 0) {$0$};
		\node [style={Z dot (zh)}] (7) at (2.25, 0) {$\pi$};
		\node [style=none] (8) at (3, 0) {};
		\node [style=none] (9) at (-3.5, 0.75) {};
		\node [style=none] (10) at (-3.5, -0.75) {};
		\node [style=none] (13) at (4, 0) {$=$};
		\node [style=hbox] (14) at (8, 0) {};
		\node [style=hbox] (15) at (9, 0) {};
		\node [style=none] (16) at (10, 0) {};
		\node [style=none] (17) at (7, -0.75) {};
		\node [style=none] (18) at (7, 0.75) {};
		\node [style={X dot (zh)}] (19) at (0.25, 0) {$\pi$};
		\node [style=labelled hbox] (20) at (0.25, 1.25) {2};
		\node [style=none] (21) at (-3.5, 0.25) {$\vdots$};
		\node [style=none] (22) at (-4, -0.75) {};
		\node [style=none] (23) at (-4, 0.75) {};
		\node [style=none] (24) at (-5, 0) {$n$};
		\node [style=none] (25) at (7, 0.25) {$\vdots$};
		\node [style=none] (26) at (6.5, -0.75) {};
		\node [style=none] (27) at (6.5, 0.75) {};
		\node [style=none] (28) at (5.5, 0) {$n$};
	\end{pgfonlayer}
	\begin{pgfonlayer}{edgelayer}
		\draw (10.center) to (0);
		\draw (0) to (2);
		\draw (3) to (1);
		\draw (1) to (9.center);
		\draw (6) to (7);
		\draw (7) to (8.center);
		\draw [bend right] (14) to (18.center);
		\draw [bend left] (14) to (17.center);
		\draw (14) to (15);
		\draw (15) to (16.center);
		\draw [bend left] (4) to (2);
		\draw [bend right] (4) to (3);
		\draw (19) to (4);
		\draw (6) to (19);
		\draw [style=brace edge] (22.center) to (23.center);
		\draw [style=brace edge] (26.center) to (27.center);
	\end{pgfonlayer}
\end{tikzpicture}
    \end{equation}
\end{lemma}

\begin{proof}
    We proceed by induction on $n$. For the case of $n = 0$:
    \begin{equation}
        \begin{tikzpicture}
	\begin{pgfonlayer}{nodelayer}
		\node [style={Z dot (zh)}] (0) at (-4.5, 6) {$\pi$};
		\node [style={X dot (zh)}] (1) at (-3.5, 6) {$\pi$};
		\node [style=labelled hbox] (2) at (-2.5, 6) {$0$};
		\node [style=white dot] (3) at (-1.5, 6) {$\pi$};
		\node [style=none] (4) at (-0.75, 6) {};
		\node [style=none] (5) at (0.5, 6) {$=$};
		\node [style=labelled hbox] (11) at (-3.5, 7) {$2$};
		\node [style={Z dot (zh)}] (29) at (2.75, 6) {$\pi$};
		\node [style={X dot (zh)}] (30) at (3.75, 6) {$\pi$};
		\node [style=labelled hbox] (31) at (4.75, 6) {$0$};
		\node [style=white dot] (32) at (5.75, 6) {$\pi$};
		\node [style=none] (33) at (6.5, 6) {};
		\node [style=labelled hbox] (34) at (1.75, 7) {$2$};
		\node [style={Z dot (zh)}] (36) at (1.75, 6) {};
		\node [style=none] (37) at (-4, 3.5) {$=$};
		\node [style={Z dot (zh)}] (38) at (-1.75, 3.5) {$\pi$};
		\node [style={X dot (zh)}] (39) at (-0.75, 3.5) {$\pi$};
		\node [style=labelled hbox] (40) at (0.25, 3.5) {$0$};
		\node [style=white dot] (41) at (1.25, 3.5) {$\pi$};
		\node [style=none] (42) at (2, 3.5) {};
		\node [style={Z dot (zh)}] (44) at (-2.75, 3.5) {};
		\node [style={X dot (zh)}] (45) at (-2.75, 4.5) {};
		\node [style=labelled hbox] (46) at (-1.75, 4.5) {$0$};
		\node [style={Z dot (zh)}] (48) at (5.25, 3.5) {$\pi$};
		\node [style={X dot (zh)}] (49) at (6.25, 3.5) {$\pi$};
		\node [style=labelled hbox] (50) at (7.25, 3.5) {$0$};
		\node [style=white dot] (51) at (8.25, 3.5) {$\pi$};
		\node [style=none] (52) at (9, 3.5) {};
		\node [style={Z dot (zh)}] (53) at (4.25, 3.5) {};
		\node [style={Z dot (zh)}] (54) at (4.25, 4.5) {};
		\node [style=labelled hbox] (55) at (6.25, 4.5) {$0$};
		\node [style=none] (56) at (3, 3.5) {$=$};
		\node [style=hbox] (58) at (5.25, 4.5) {};
		\node [style={Z dot (zh)}] (59) at (-0.75, 1) {$\pi$};
		\node [style={X dot (zh)}] (60) at (0.25, 1) {$\pi$};
		\node [style=labelled hbox] (61) at (1.25, 1) {$0$};
		\node [style=white dot] (62) at (2.25, 1) {$\pi$};
		\node [style=none] (63) at (3, 1) {};
		\node [style=labelled hbox] (66) at (0.25, 2) {$0$};
		\node [style=none] (67) at (-4, 1.25) {$\stackrel{BA_H}{=}$};
		\node [style=hbox] (69) at (-0.75, 2) {};
		\node [style=hbox] (70) at (-1.75, 1) {};
		\node [style={X dot (zh)}] (71) at (-2.75, 1) {};
		\node [style={Z dot (zh)}] (72) at (-0.75, -1.25) {$\pi$};
		\node [style={X dot (zh)}] (73) at (0.25, -1.25) {$\pi$};
		\node [style=labelled hbox] (74) at (1.25, -1.25) {$0$};
		\node [style=white dot] (75) at (2.25, -1.25) {$\pi$};
		\node [style=none] (76) at (3, -1.25) {};
		\node [style=none] (78) at (-4, -1) {$\stackrel{SF_H}{=}$};
		\node [style={X dot (zh)}] (82) at (-2.75, -1.25) {};
		\node [style=labelled hbox] (83) at (-1.75, -1.25) {$0$};
		\node [style=none] (89) at (-4, -3.5) {$\stackrel{SF}{=}$};
		\node [style={Z dot (zh)}] (90) at (0.25, -3.75) {$\pi$};
		\node [style={X dot (zh)}] (91) at (1.25, -3.75) {$\pi$};
		\node [style=labelled hbox] (92) at (2.25, -3.75) {$0$};
		\node [style=white dot] (93) at (3.25, -3.75) {$\pi$};
		\node [style=none] (94) at (4, -3.75) {};
		\node [style={X dot (zh)}] (95) at (-1.75, -3.75) {$\pi$};
		\node [style=labelled hbox] (96) at (-0.75, -3.75) {$0$};
		\node [style={X dot (zh)}] (97) at (-2.75, -3.75) {$\pi$};
		\node [style=none] (98) at (-4, -6) {$\stackrel{\ref{completeness_zh0tofident}}{=}$};
		\node [style=none] (103) at (-1.25, -6.25) {};
		\node [style={X dot (zh)}] (106) at (-2.5, -6.25) {$\pi$};
		\node [style=labelled hbox] (107) at (-0.75, 4.5) {$2$};
		\node [style=labelled hbox] (108) at (-0.75, -0.25) {$2$};
		\node [style=labelled hbox] (109) at (0.25, -2.75) {$2$};
		\node [style=labelled hbox] (110) at (-2.5, -5.25) {$2$};
		\node [style=none] (111) at (-0.25, -6.25) {$=$};
		\node [style=none] (112) at (3.25, -6.25) {};
		\node [style=hbox] (115) at (0.75, -6.25) {};
		\node [style=hbox] (116) at (2.25, -6.25) {};
		\node [style={Z dot (zh)}] (117) at (1.5, -6.25) {};
		\node [style=none] (118) at (4.25, -6) {$\stackrel{I_Z}{=}$};
		\node [style=none] (119) at (7.25, -6.25) {};
		\node [style=hbox] (120) at (5.25, -6.25) {};
		\node [style=hbox] (121) at (6.25, -6.25) {};
	\end{pgfonlayer}
	\begin{pgfonlayer}{edgelayer}
		\draw (4.center) to (3);
		\draw (3) to (2);
		\draw (2) to (1);
		\draw (1) to (0);
		\draw (33.center) to (32);
		\draw (32) to (31);
		\draw (31) to (30);
		\draw (30) to (29);
		\draw (36) to (29);
		\draw (42.center) to (41);
		\draw (41) to (40);
		\draw (40) to (39);
		\draw (39) to (38);
		\draw (44) to (38);
		\draw (46) to (45);
		\draw (52.center) to (51);
		\draw (51) to (50);
		\draw (50) to (49);
		\draw (49) to (48);
		\draw (53) to (48);
		\draw (58) to (55);
		\draw (58) to (54);
		\draw (63.center) to (62);
		\draw (62) to (61);
		\draw (61) to (60);
		\draw (60) to (59);
		\draw (69) to (66);
		\draw (70) to (59);
		\draw [bend left=15] (70) to (69);
		\draw (71) to (70);
		\draw (76.center) to (75);
		\draw (75) to (74);
		\draw (74) to (73);
		\draw (73) to (72);
		\draw (83) to (72);
		\draw (83) to (82);
		\draw (94.center) to (93);
		\draw (93) to (92);
		\draw (92) to (91);
		\draw (91) to (90);
		\draw (96) to (90);
		\draw (96) to (95);
		\draw (95) to (97);
		\draw (103.center) to (106);
		\draw (112.center) to (116);
		\draw (116) to (117);
		\draw (117) to (115);
		\draw (119.center) to (121);
		\draw (121) to (120);
	\end{pgfonlayer}
\end{tikzpicture}
    \end{equation}
    For the case of $n = k + 1$, assuming the result holds for $k$:
    \begin{equation}
        \begin{tikzpicture}
	\begin{pgfonlayer}{nodelayer}
		\node [style=hbox] (0) at (0, 4.25) {};
		\node [style=hbox] (1) at (1, 4.25) {};
		\node [style=none] (2) at (2, 4.25) {};
		\node [style=none] (3) at (-1, 3.5) {};
		\node [style=none] (4) at (-1, 5) {};
		\node [style=none] (5) at (-1, 4.5) {$\vdots$};
		\node [style=hbox] (9) at (5, 4.25) {};
		\node [style=hbox] (10) at (6, 4.25) {};
		\node [style=none] (11) at (9, 4.25) {};
		\node [style=none] (12) at (4, 3.5) {};
		\node [style=none] (13) at (4, 5) {};
		\node [style=none] (14) at (4, 4.5) {$\vdots$};
		\node [style=none] (18) at (3, 4.5) {$\stackrel{SF_H}{=}$};
		\node [style=labelled hbox] (19) at (8, 5.25) {$\frac{1}{2}$};
		\node [style=hbox] (20) at (7, 4.25) {};
		\node [style=hbox] (21) at (8, 4.25) {};
		\node [style=none] (22) at (5.75, 5.75) {};
		\node [style=none] (23) at (-1, 5.75) {};
		\node [style=none] (24) at (4, 5.75) {};
		\node [style={X dot (zh)}] (37) at (2.25, 0) {$\pi$};
		\node [style={X dot (zh)}] (38) at (2.25, 1.5) {$\pi$};
		\node [style=labelled hbox] (39) at (3.25, 0) {$0$};
		\node [style=labelled hbox] (40) at (3.25, 1.5) {$0$};
		\node [style={Z dot (zh)}] (41) at (4.25, 0.75) {$\pi$};
		\node [style=labelled hbox] (42) at (6.25, 0.75) {$0$};
		\node [style={Z dot (zh)}] (43) at (7.25, 0.75) {$\pi$};
		\node [style=none] (45) at (1.5, 1.5) {};
		\node [style=none] (46) at (1.5, 0) {};
		\node [style={X dot (zh)}] (47) at (5.25, 0.75) {$\pi$};
		\node [style=none] (49) at (1.5, 1) {$\vdots$};
		\node [style={X dot (zh)}] (50) at (8.25, 0.75) {$\pi$};
		\node [style={X dot (zh)}] (51) at (7.75, 2.25) {$\pi$};
		\node [style=labelled hbox] (52) at (9.25, 0.75) {$0$};
		\node [style=labelled hbox] (53) at (8.75, 2.25) {$0$};
		\node [style={Z dot (zh)}] (54) at (10.25, 0.75) {$\pi$};
		\node [style=labelled hbox] (55) at (12.25, 0.75) {$0$};
		\node [style={Z dot (zh)}] (56) at (13.25, 0.75) {$\pi$};
		\node [style=none] (57) at (14, 0.75) {};
		\node [style=none] (58) at (1.5, 2.25) {};
		\node [style={X dot (zh)}] (60) at (11.25, 0.75) {$\pi$};
		\node [style=labelled hbox] (61) at (12.25, 2) {2};
		\node [style={X dot (zh)}] (85) at (2.25, -8.75) {$\pi$};
		\node [style={X dot (zh)}] (86) at (2.25, -7.25) {$\pi$};
		\node [style=labelled hbox] (87) at (3.25, -8.75) {$0$};
		\node [style=labelled hbox] (88) at (3.25, -7.25) {$0$};
		\node [style={Z dot (zh)}] (89) at (4.25, -8) {$\pi$};
		\node [style=none] (92) at (1.5, -7.25) {};
		\node [style=none] (93) at (1.5, -8.75) {};
		\node [style=none] (95) at (1.5, -7.75) {$\vdots$};
		\node [style={X dot (zh)}] (97) at (2.25, -6.25) {$\pi$};
		\node [style=labelled hbox] (99) at (3.25, -6.25) {$0$};
		\node [style=labelled hbox] (101) at (6.75, -8) {$0$};
		\node [style={Z dot (zh)}] (102) at (7.75, -8) {$\pi$};
		\node [style=none] (103) at (8.5, -8) {};
		\node [style=none] (104) at (1.5, -6.25) {};
		\node [style={X dot (zh)}] (105) at (5.75, -8) {$\pi$};
		\node [style=labelled hbox] (106) at (5.75, -6.25) {2};
		\node [style={Z dot (zh)}] (107) at (5, -8) {};
		\node [style={X dot (zh)}] (108) at (2.25, -12.75) {$\pi$};
		\node [style={X dot (zh)}] (109) at (2.25, -11.25) {$\pi$};
		\node [style=labelled hbox] (110) at (3.25, -12.75) {$0$};
		\node [style=labelled hbox] (111) at (3.25, -11.25) {$0$};
		\node [style={Z dot (zh)}] (112) at (4.25, -12) {$\pi$};
		\node [style=none] (113) at (1.5, -11.25) {};
		\node [style=none] (114) at (1.5, -12.75) {};
		\node [style=none] (115) at (1.5, -11.75) {$\vdots$};
		\node [style={X dot (zh)}] (116) at (2.25, -10.25) {$\pi$};
		\node [style=labelled hbox] (117) at (3.25, -10.25) {$0$};
		\node [style=labelled hbox] (118) at (6.25, -12) {$0$};
		\node [style={Z dot (zh)}] (119) at (7.25, -12) {$\pi$};
		\node [style=none] (120) at (8, -12) {};
		\node [style=none] (121) at (1.5, -10.25) {};
		\node [style={X dot (zh)}] (122) at (5.25, -12) {$\pi$};
		\node [style=labelled hbox] (123) at (5.25, -10.75) {2};
		\node [style=none] (125) at (0.5, -11.75) {$\stackrel{SF_Z}{=}$};
		\node [style=none] (127) at (0.5, -8) {$\stackrel{\ref{completeness_zh0tofident}}{=}$};
		\node [style=none] (128) at (0.5, 0.75) {$\stackrel{\ref{completeness_zh0tof}}{=}$};
		\node [style={X dot (zh)}] (129) at (2.25, -4.75) {$\pi$};
		\node [style={X dot (zh)}] (130) at (2.25, -3.25) {$\pi$};
		\node [style=labelled hbox] (131) at (3.25, -4.75) {$0$};
		\node [style=labelled hbox] (132) at (3.25, -3.25) {$0$};
		\node [style={Z dot (zh)}] (133) at (4.25, -4) {$\pi$};
		\node [style=labelled hbox] (134) at (6.25, -4) {$0$};
		\node [style={Z dot (zh)}] (135) at (7.25, -4) {$\pi$};
		\node [style=none] (136) at (1.5, -3.25) {};
		\node [style=none] (137) at (1.5, -4.75) {};
		\node [style={X dot (zh)}] (138) at (5.25, -4) {$\pi$};
		\node [style=none] (139) at (1.5, -3.75) {$\vdots$};
		\node [style={X dot (zh)}] (140) at (8.25, -4) {$\pi$};
		\node [style={X dot (zh)}] (141) at (7.75, -2.5) {$\pi$};
		\node [style=labelled hbox] (142) at (9.25, -4) {$0$};
		\node [style=labelled hbox] (143) at (8.75, -2.5) {$0$};
		\node [style={Z dot (zh)}] (144) at (11, -4) {};
		\node [style=labelled hbox] (145) at (12.75, -4) {$0$};
		\node [style={Z dot (zh)}] (146) at (13.75, -4) {$\pi$};
		\node [style=none] (147) at (14.5, -4) {};
		\node [style=none] (148) at (1.5, -2.5) {};
		\node [style={X dot (zh)}] (149) at (11.75, -4) {$\pi$};
		\node [style=labelled hbox] (150) at (12.75, -2.75) {2};
		\node [style=none] (151) at (0.5, -4) {$\stackrel{SF_Z}{=}$};
		\node [style={Z dot (zh)}] (152) at (10.25, -4) {$\pi$};
	\end{pgfonlayer}
	\begin{pgfonlayer}{edgelayer}
		\draw [bend right] (0) to (4.center);
		\draw [bend left] (0) to (3.center);
		\draw (0) to (1);
		\draw (1) to (2.center);
		\draw [bend right] (9) to (13.center);
		\draw [bend left] (9) to (12.center);
		\draw (9) to (10);
		\draw (11.center) to (21);
		\draw (21) to (20);
		\draw (20) to (10);
		\draw [bend left] (23.center) to (0);
		\draw (24.center) to (22.center);
		\draw [in=90, out=0] (22.center) to (20);
		\draw (46.center) to (37);
		\draw (37) to (39);
		\draw (40) to (38);
		\draw (38) to (45.center);
		\draw (42) to (43);
		\draw [bend left] (41) to (39);
		\draw [bend right] (41) to (40);
		\draw (47) to (41);
		\draw (42) to (47);
		\draw (50) to (52);
		\draw (53) to (51);
		\draw (51) to (58.center);
		\draw (55) to (56);
		\draw (56) to (57.center);
		\draw [bend left=360] (54) to (52);
		\draw [bend right] (54) to (53);
		\draw (60) to (54);
		\draw (55) to (60);
		\draw (50) to (43);
		\draw (93.center) to (85);
		\draw (85) to (87);
		\draw (88) to (86);
		\draw (86) to (92.center);
		\draw [bend left] (89) to (87);
		\draw [bend right] (89) to (88);
		\draw (99) to (97);
		\draw (97) to (104.center);
		\draw (101) to (102);
		\draw (102) to (103.center);
		\draw (101) to (105);
		\draw [in=0, out=90] (107) to (99);
		\draw (107) to (105);
		\draw (89) to (107);
		\draw (114.center) to (108);
		\draw (108) to (110);
		\draw (111) to (109);
		\draw (109) to (113.center);
		\draw [bend left] (112) to (110);
		\draw [bend right] (112) to (111);
		\draw (117) to (116);
		\draw (116) to (121.center);
		\draw (118) to (119);
		\draw (119) to (120.center);
		\draw (118) to (122);
		\draw [bend right] (112) to (117);
		\draw (122) to (112);
		\draw (137.center) to (129);
		\draw (129) to (131);
		\draw (132) to (130);
		\draw (130) to (136.center);
		\draw (134) to (135);
		\draw [bend left] (133) to (131);
		\draw [bend right] (133) to (132);
		\draw (138) to (133);
		\draw (134) to (138);
		\draw (140) to (142);
		\draw (143) to (141);
		\draw (141) to (148.center);
		\draw (145) to (146);
		\draw (146) to (147.center);
		\draw [bend right] (144) to (143);
		\draw (149) to (144);
		\draw (145) to (149);
		\draw (140) to (135);
		\draw (152) to (144);
		\draw (152) to (142);
	\end{pgfonlayer}
\end{tikzpicture}
    \end{equation}
\end{proof}

\end{toappendix}

\maketitle

\begin{abstract}
We study the counting version of the Boolean satisfiability problem \sSAT using the ZH-calculus, a graphical language originally introduced to reason about quantum circuits. Using this, we generalize \sSAT to a weighted variant we call $\sSAT_\pm$, which is complete for the class GapP. We show there is an efficient linear-time reduction from \sSAT to $\stwoSAT_\pm$, unlike previous reductions from \sSAT to \stwoSAT which blow up the size of the formula by a polynomial factor. Our main conceptual contribution is that introducing weights to \sSAT allows for more efficient translations, and we use this to remove the dependence on clause width $k$ in this case.
We observe that DPLL-style algorithms for \stwoSAT can be adapted to $\stwoSAT_\pm$ directly and hence the best-known upper bounds for \stwoSAT apply. Applying an upper bound for \stwoSAT in terms of variables gives us upper bounds for \sSAT in terms of clauses and variables that are better than $O^*(2^n)$ for small clause densities of $\frac{m}{n} < 2.25$, and improve on previous average-case and worst-case bounds for $k \geq 6$ and $k \geq 4$, respectively. Applying a similar bound in terms of clauses produces a bound of $O^*(1.1740^L)$ in terms of the length of the formula. These are, to our knowledge, the first non-trivial upper bounds for \sSAT that is independent of clause size, and in terms of formula length, respectively. Based on a result of Kutzkov, we find an improved bound on \sthreeSAT for $1.2577 < \frac{m}{n} \leq \frac{7}{3}$. Finally, we use this technique to find an upper bound on the complexity of calculating amplitudes of quantum circuits in terms of the total number of gates. Our results demonstrate that graphical reasoning can lead to new algorithmic insights, even outside the domain of quantum computing that the calculus was intended for.
\end{abstract}



\section{Introduction}

A graphical calculus is a language consisting of diagrams that can be transformed according to specific graphical rewrite rules. Usually these diagrams correspond to some underlying mathematical object that would be hard to reason about directly---like a matrix, tensor, relation or some combinatorial object---and the rewrite rules preserve the semantics of this interpretation.
There are for instance graphical calculi for linear algebra~\cite{BONCHI2017144,zanasi2015thesis,Bonchi2019Graphical,Boisseau2022Graphical}, for studying concurrency~\cite{Bonchi2014Categorical,Bonchi2019Diagrammatic}, and for finite-state automata~\cite{Piedeleu2021stringdiagrammatic}.

Most relevant for this paper are the graphical calculi developed for studying quantum computing. The \emph{ZX-calculus}~\cite{CD1,coecke_interacting_2011} can represent arbitrary linear maps between any number of qubits, and has different versions of rewrite rules that are \emph{complete} (meaning the rules can prove any true equality) for various relevant fragments of quantum computing~\cite{Backens1,SimonCompleteness,HarnyAmarCompleteness,vilmarteulercompleteness}.
It has seen use in a variety of areas like optimizing quantum computations~\cite{cliffsimp,kissinger2019tcount,deBeaudrap2020Techniques,hanks2019effective,Backens2020extraction}, more effectively classically simulating quantum computations~\cite{kissinger2021simulating,kissinger2022classical}, and several others like~\cite{magicFactories,cerveromartin2022barren,townsend-teague2021classifying}; see~\cite{vandewetering2020zxcalculus} for a review.

There are a number of variations on the ZX-calculus that include different or additional generators~\cite{Wang2019Algebraic,vilmart_zx-calculus_2019,hadzihasanovic2015diagrammatic}. The one we will use is the \emph{ZH-calculus}~\cite{backens_zh_2019,backens_completeness_2021}.
The ZH-calculus has turned out to be useful in a variety of areas~\cite{pathsRenaud,d.p.east2021spinnetworks,vilmart2021quantum}, but in particular it has been shown to naturally encode Boolean satisfiability and counting problems~\cite{de_beaudrap_tensor_2021}. We will build on this representation to show that this perspective leads to better algorithms for formulae that have a low number of clauses.

The Boolean satisfiability problem (\SAT) is to determine whether a given Boolean formula has a satisfying assignment of variables, and is a canonical example of an \NP-complete problem. Worst-case upper bounds for solving \SAT instances can be phrased in terms of some relevant parameters for a Boolean formula $f:\B^n\to \B$ in conjunctive normal form: its number of variables $n$, the maximum clause size $k$, the number of clauses $m$, the clause density $\delta:=\frac{m}{n}$, the number of literals $L$, and the maximal number of clauses a variable participates in.
For instance, for fixed $k$, there are bounds $O^*(c_k^n)$ where $c_k<2$~\cite{dubois_counting_1991}. Yamamoto~\cite{yamamoto_improved_2005-1} showed that \SAT can be solved in $O^*(2^{0.3033m})$ time (independent of $k$) . This hence implies a better than $O^*(2^n)$ runtime for clause densities $\delta < 3.297$, regardless of $k$. While for large $k$ and large $\delta$ all known bounds converge to $O^*(2^n)$, for small $k$ and arbitrary $\delta$ or vice-versa, better bounds are possible.

In this paper we will study the problem \sSAT, which asks \emph{how many} solutions a Boolean formula has. Hence, this is not a decision problem, but a counting problem. It is complete for the complexity class \sP, which is the `counting version' of \NP. \sSAT is believed to be a significantly harder problem than \SAT. For instance, the entire polynomial hierarchy is contained in $\textbf{P}^{\sSAT}$ \cite{toda_pp_1991}.
As it is \sP-complete, it has applications in a variety of areas. For instance, tensor-network contraction (when suitably formalized) is in \sP~\cite{damm_complexity_2002}. \sSAT also has applications in the field of artificial intelligence, where it is usually referred to as \emph{model counting} \cite[Chapter 20]{biere_handbook_2009}.
Note that while for \SAT the problem only becomes hard for $k\geq 3$, for \sSAT, the problem is already hard for $k=2$, as $\stwoSAT$ is \sP-complete \cite{valiant_complexity_1979}. In fact, previously in \cite{laakkonen_2023}, we used the ZH-calculus to provide a proof of this and other reductions between related counting problems. In this work, we make use of similar techniques in a different direction to give improved upper bounds for \sSAT.

We can phrase the known upper bounds to \sSAT in terms of $n$, $k$, $m$ and $\delta$, although for \sSAT much less is known. Similarly to \SAT, good bounds are known for small $k$ --- $O^*(1.2377^n)$ for $k = 2$ \cite{wahlstrom_tighter_2008}, and $O^*(1.6423^n)$ for $k = 3$ \cite{kutzkov_new_2007}. There are also bounds known for arbitrary (but fixed) $k$ in the worst-case setting~\cite{dubois_counting_1991} and the average-case setting~\cite{williams_computing_2004}. Unlike \SAT, bounds in terms of $m$ are not known for \sSAT independent of $k$, but only for $k = 2$, $O^*(1.1740^{m})$~\cite{wang_worst_2013}, and $k = 3$, $O^*(1.4142^{m})$~\cite{zhou_new_2010}. 


In this paper we establish, to the best of our knowledge, the first algorithm for $\sSAT$ that is better than brute-force for low clause density, independent of the clause size $k$. Specifically, we prove the following theorem:
\begin{theorem}
    \label{ssatmbound}
    Given a CNF formula $\phi:\B^n\to \B$, we can count the number of satisfying assignments $$\#(\phi) = \#\{\vec x\in \B^n\,|\, \phi(\vec x) = 1\}$$ in time $O^*(1.2377^{n + m_{\geq 3}})$, where $m_{\geq 3}$ is the number of clauses of width at least three.
    In particular, for clause density $\delta < 2.2503$, this gives a better than $O^*(2^n)$ bound, independent of maximal clause size $k$.
\end{theorem}

The worst-case bound for our algorithm improves the best-known bound for a variety of different parameters. We summarize this in Tables~\ref{boundtable} and~\ref{density_comparison}. Note that these tables also contain our results based on a more fine-grained analysis of \#\textbf{3SAT} that is presented in Section~\ref{low3sat}, as well as a bound on $\sSAT$ in terms of literals $L$ that is presented in Section~\ref{sec:literal-bound}. Furthermore, assuming the strong exponential time hypothesis, our results indicate that the `hardest' density of \sSAT must be some $\delta > 2.2503$. As far as we aware this is the first known bound on where the hardest clause density of \sSAT lies.

\begin{table}
    \centering

    \begin{tabular}{lllll}
        \toprule
        Problem & Previous Best & New Bound & Relevant Region & Algorithm\\\midrule
        $\skSAT$ for $k > 3$ & $O^*(c_k^n)$, $c_k \to 2$  \cite{dubois_counting_1991} & $O^*(1.2377^{n + m})$ & $\delta < 2.2503$  (as $k \to \infty$) & \cite{wahlstrom_tighter_2008} \\\midrule
        $\skSAT$  for $k > 3$ & $O^*(c_k^n)$, $c_k \to 2$  \cite{dubois_counting_1991} & $O^*(1.1740^L)$ & $\frac{L}{n} < 4.3209$  (as $k \to \infty$) & \cite{wang_worst_2013} \\\midrule
        $\#\textbf{3SAT}$ & $O^*(1.6423^n)$ \cite{kutzkov_new_2007} & $O^*(1.6350^n)$ & $1.2577 < \delta \leq \frac{7}{3}$ & \cite{wahlstrom_tighter_2008} \cite{kutzkov_new_2007} \\ \bottomrule
    \end{tabular}

    \caption{The different bounds obtained in this paper, along with the corresponding best previous bounds, and the underlying algorithm on which they are based. The given relevant regions are the intervals of formula parameters where our bounds are valid and improve on the previous bound. \label{boundtable}}
\end{table}

\setlength{\tabcolsep}{1.5mm}
\begin{table}[tp]
    \centering
    \begin{tabular}{r@{\hskip 3mm}|c|cccccccc}
        \toprule
        Improvement on & $k$ & 2 & 3 & 4 & 5 & 6 & 7 & 8 & 9 \\
        \midrule
        Average-Case & $\delta$ & -- & -- & -- & -- & -- & $<0.968$ & $<1.106$ & $<1.207$ \\
        \midrule
        Worst-Case & $\delta$ & -- & $<2.333^*$ & $<2.077$ & $<2.170$ & $<2.212$ & $<2.231$ & $<2.241$ & $<2.246$ \\
        \bottomrule
    \end{tabular}
    \caption{The maximum densities $\delta$ at which our upper bound improves on other bounds, as dependent on $k$. We include average \cite{williams_computing_2004} and worst-case bounds \cite{dubois_counting_1991}, both of which converge to $\delta = 2.2503$ as $k$ increases. Entries marked with `--' are where our bound is always worse. The entry marked `$*$' uses the alternate bound presented in Section \ref{low3sat} and only improves on the previous best when also $\delta > 1.2577$. \label{density_comparison}}
\end{table}

While the bound of Theorem~\ref{ssatmbound} is only effective at low densities, this is sufficient for many real-world use cases -- in particular, for the unweighted \sSAT instances from the Model Counting Competition 2020 \cite{fichte_model_2020}, we found that we improve on the previous best worst-case bound for 88\% of instances, and the average-case bound for 45\% of instances, assuming the constant factors of both algorithms are equal (we take $k$ to be the 90th percentile of clause widths to avoid few large clauses biasing the result in our favor). However, it is important to note that practical \sSAT solvers perform much better than predicted by any of these upper bounds.


Our results are based on two observations.
The first observation is that the standard CDP (Counting Davis-Putnam) algorithm \cite{birnbaum_good_1999} for solving \sSAT can actually solve a more general problem that we dub $\sSAT_{\pm}$. In this problem, variables $x_i$ are labeled by a $\phi_{i} = \pm 1$ phase that determines whether a solution to $f$ should be added or subtracted to the total, as determined by $\prod_{\vec x}\phi_{i}^{x_i}$.
The second observation is that we can translate an arbitrary \sSAT instance into a $\stwoSAT_{\pm}$ instance. This removes the dependence on maximal clause size $k$ from our problem, and means we can use the known upper bounds to \stwoSAT that apply directly to the problem $\stwoSAT_{\pm}$ as well.

We found this last observation by writing a \sSAT instance as a ZH-diagram as 
described by de Beaudrap \emph{et al.}~\cite{de_beaudrap_tensor_2021}. 
We extend their methods by relaxing the conditions on the types of diagrams we consider, which shows that ZH-diagrams also naturally represent $\sSAT_\pm$ instances. The translation from a \sSAT instance into a $\stwoSAT_\pm$ instance then follows from a known rewrite rule of the related ZX$\Delta$-calculus~\cite{vilmart_zx-calculus_2019}.
We also find that $\sSAT_\pm$ is in fact complete for the complexity class \textbf{GapP}~\cite{fenner_gap-definable_1994}, which is \sP closed under negation. 

In Section~\ref{sec:prelim} we recall the definition of the ZH-calculus, how to encode \sSAT instances as ZH-diagrams, and how the CDP algorithm for solving \sSAT works.
Then in Section~\ref{sec:results} we present our main results: we show how to interpret CDP graphically inside the ZH-calculus, and we find a graphical reduction from  \sSAT to $\stwoSAT_{\pm}$. We end with a complexity analysis of combining this reduction with the CDP algorithm, modified to work for $\stwoSAT_{\pm}$.
In Section~\ref{sec:variations} we study some variations on our algorithm: Section~\ref{sec:literal-bound} presents a new bound of $O^*(1.1740^{L})$ for \sSAT that is in terms of number of literals; Section~\ref{low3sat} gives a modified algorithm for \textbf{\#3SAT} that is better for certain densities; and Section~\ref{sec:SAT-phases} presents the problem of \sSAT where variables are labeled by arbitrary complex numbers, which we conjecture might be helpful for future improvements. In particular, in Appendix~\ref{sec:circuit-sim} we apply this technique to calculating amplitudes of quantum circuits, and show an upper bound in terms of the number of gates.
We end with some concluding remarks in Section~\ref{sec:conclusion}.

\section{Preliminaries}\label{sec:prelim}

We say that a Boolean formula $\phi : \mathbb{B}^n \to \mathbb{B}$ is in conjunctive normal form if we have 
\begin{equation}
    \phi(x_1, \dots, x_n) = \bigwedge_{i = 1}^{m} (c_{i1} \lor c_{i2} \lor \cdots \lor c_{ik_i})
\end{equation}
where each $c_{ij}$ is $x_l$ or $\lnot x_l$ for some $l$. We say that $\phi$ has $n$ variables, $m$ clauses, maximum clause width $k = \max_i \{k_i\}$, density $\delta = \frac{m}{n}$ and number of literals $L = \sum_i k_i$. 

\subsection{The ZH-calculus}\label{sec:ZH}

We use the ZH-calculus, a graphical language designed for reasoning about quantum computations~\cite{backens_zh_2019}. Here we only recall the definition of the generators; see \cite[Section 8]{vandewetering2020zxcalculus} for a more detailed overview. The calculus is defined by rewrites on ZH-diagram, which are composed of two generators, called the Z-spider and the H-box. These are given by
\begin{equation}
    \begin{tikzpicture}
	\begin{pgfonlayer}{nodelayer}
		\node [style={Z dot (zh)}] (0) at (0, 2) {};
		\node [style=none] (1) at (1, 3.25) {};
		\node [style=none] (2) at (1, 0.75) {};
		\node [style=none] (3) at (1, 1.75) {$\vdots$};
		\node [style=none] (4) at (1.5, 3.25) {$\sigma_1$};
		\node [style=none] (5) at (1, 2.5) {};
		\node [style=none] (6) at (1.5, 2.5) {$\sigma_2$};
		\node [style=none] (7) at (1.5, 0.75) {$\sigma_n$};
		\node [style=none] (8) at (3, 2) {$=$};
		\node [style=none] (9) at (5.5, 2) {$\delta_{\sigma_1\sigma_2\dots \sigma_n}$};
		\node [style=none] (10) at (8, 2) {$=$};
		\node [style=none] (11) at (13, 2) {$\begin{cases} 1 & \sigma_1 = \sigma_2 = \dots = \sigma_n \\ 0 & \text{otherwise}\end{cases}$};
		\node [style=labelled hbox] (12) at (0, -2) {$a$};
		\node [style=none] (13) at (1, -0.75) {};
		\node [style=none] (14) at (1, -3.25) {};
		\node [style=none] (15) at (1, -2.25) {$\vdots$};
		\node [style=none] (16) at (1.5, -0.75) {$\sigma_1$};
		\node [style=none] (17) at (1, -1.5) {};
		\node [style=none] (18) at (1.5, -1.5) {$\sigma_2$};
		\node [style=none] (19) at (1.5, -3.25) {$\sigma_n$};
		\node [style=none] (20) at (3, -2) {$=$};
		\node [style=none] (21) at (7, -2) {$1 + (a - 1)\delta_{1\sigma_1\sigma_2\dots \sigma_n}$};
		\node [style=none] (22) at (11, -2) {$=$};
		\node [style=none] (23) at (16.5, -2) {$\begin{cases} a & 1 = \sigma_1 = \sigma_2 = \dots = \sigma_n\\ 1 & \text{otherwise}\end{cases}$};
	\end{pgfonlayer}
	\begin{pgfonlayer}{edgelayer}
		\draw [bend left] (2.center) to (0);
		\draw [bend left] (0) to (1.center);
		\draw [bend right=15] (5.center) to (0);
		\draw [bend left] (12) to (13.center);
		\draw [bend right=15] (17.center) to (12);
		\draw [bend right] (12) to (14.center);
	\end{pgfonlayer}
\end{tikzpicture}
\end{equation}
along with their interpretation as tensors. Each H-box is labeled by a constant $a \in \mathbb{C}$, with unlabeled H-boxes corresponding to $a = -1$. Diagrams composed from these can be interpreted as tensor networks: see \cite[Section 4.1]{orus_practical_2014} for an introduction. In particular, parallel composition of generators corresponds to composing the tensors via tensor product, and connecting wires between generators corresponds to contracting the shared index. The tensors are symmetric over their indices, which implies that \emph{only connectivity matters} -- diagrams with the same topology represent the same tensor network. 

We also have the following derived generators, defined as
\begin{equation}
    \label{eq:zspiderphase}
    \begin{tikzpicture}
	\begin{pgfonlayer}{nodelayer}
		\node [style={X phase dot (zh)}] (0) at (16.75, 0) {$\alpha$};
		\node [style=none] (1) at (17.75, 1.25) {};
		\node [style=none] (2) at (17.75, -1.25) {};
		\node [style=none] (3) at (17.75, -0.25) {$\vdots$};
		\node [style=none] (4) at (17.75, 0.5) {};
		\node [style={Z phase dot (zh)}] (5) at (22.25, 0) {$\alpha$};
		\node [style=hbox] (6) at (23, 1.25) {};
		\node [style=hbox] (7) at (23, -1.25) {};
		\node [style=hbox] (8) at (23, 0.5) {};
		\node [style=none] (9) at (19.25, 0) {$=$};
		\node [style=none] (10) at (23.5, 1.25) {};
		\node [style=none] (11) at (23.5, -1.25) {};
		\node [style=none] (12) at (23.5, -0.25) {$\vdots$};
		\node [style=none] (13) at (23.5, 0.5) {};
		\node [style=labelled hbox] (14) at (21, 0) {$\frac{1}{2}$};
		\node [style={Z phase dot (zh)}] (26) at (5.75, 0) {$\alpha$};
		\node [style=none] (27) at (6.75, 1.25) {};
		\node [style=none] (28) at (6.75, -1.25) {};
		\node [style=none] (29) at (6.75, -0.25) {$\vdots$};
		\node [style=none] (30) at (6.75, 0.5) {};
		\node [style={Z dot (zh)}] (31) at (11, 0) {};
		\node [style=none] (32) at (12, 1.25) {};
		\node [style=none] (33) at (12, -1.25) {};
		\node [style=none] (34) at (12, -0.25) {$\vdots$};
		\node [style=none] (35) at (12, 0.5) {};
		\node [style=none] (36) at (8.25, 0) {$=$};
		\node [style=labelled hbox] (37) at (10, 0) {$e^{i\alpha}$};
	\end{pgfonlayer}
	\begin{pgfonlayer}{edgelayer}
		\draw [bend left] (2.center) to (0);
		\draw [bend left] (0) to (1.center);
		\draw [bend right=15] (4.center) to (0);
		\draw [bend left] (7) to (5);
		\draw [bend left] (5) to (6);
		\draw [bend right=15] (8) to (5);
		\draw (13.center) to (8);
		\draw (6) to (10.center);
		\draw (11.center) to (7);
		\draw [bend left] (28.center) to (26);
		\draw [bend left] (26) to (27.center);
		\draw [bend right=15] (30.center) to (26);
		\draw [bend left] (33.center) to (31);
		\draw [bend left] (31) to (32.center);
		\draw [bend right=15] (35.center) to (31);
		\draw (37) to (31);
	\end{pgfonlayer}
\end{tikzpicture}
\end{equation}
where the first is an extension of the Z-spider with $\alpha \in \mathbb{C}$ a phase, and the second is the X-spider. In this work, we consider any ZH-diagram to represent its underlying tensor network - hence, if there are $n$ open wires, this represents a tensor with $n$ indices. In particular, if there are no open wires we call this a scalar diagram, since this represents a scalar. 

The ZH-calculus is equipped with a set of rewrite rules, which are shown in Appendix \ref{sec:zhrules}. They are sound with respect to the tensor representation of ZH-diagrams, and also complete: for any pair of diagrams with identical tensor representations, there is a proof that they are equal using these rules.

\subsection{\sSAT instances as ZH-diagrams}
\label{sec:embedding}

    In previous work, de Beaudrap \emph{et al.} \cite{de_beaudrap_tensor_2021} gave a translation from \sSAT instances into ZH-diagrams, which we adopt here. In particular, a CNF Boolean formula $\phi$ is mapped into a ZH-diagram by translating clauses to zero-labeled H-boxes, variables to Z-spiders and negation to X-spiders:
    \begin{equation}
        \begin{tikzpicture}
	\begin{pgfonlayer}{nodelayer}
		\node [style=none] (0) at (2.75, 0.5) {};
		\node [style=none] (1) at (4.25, 0.5) {};
		\node [style={Z dot (zh)}] (2) at (3.5, -0.5) {};
		\node [style=none] (3) at (3.5, 0.5) {$\dots$};
		\node [style=none] (4) at (0, 0) {Variables $\iff$};
		\node [style=none] (5) at (10.75, -1) {};
		\node [style=none] (6) at (12.25, -1) {};
		\node [style=labelled hbox] (7) at (11.5, 0.5) {$0$};
		\node [style=none] (8) at (11.5, -1) {$\dots$};
		\node [style=none] (9) at (8, 0) {Clauses $\iff$};
		\node [style={X dot (zh)}] (10) at (10.75, -0.5) {$\pi$};
		\node [style={X dot (zh)}] (11) at (12.25, -0.5) {$\pi$};
		\node [style={X dot (zh)}] (12) at (19.5, 0) {$\pi$};
		\node [style=none] (13) at (19.5, -0.75) {};
		\node [style=none] (14) at (19.5, 0.75) {};
		\node [style=none] (15) at (16.5, 0) {Negation $\iff$};
	\end{pgfonlayer}
	\begin{pgfonlayer}{edgelayer}
		\draw [bend right=15] (2) to (1.center);
		\draw [bend left=15] (2) to (0.center);
		\draw [bend right=15] (11) to (7);
		\draw [bend right=15] (7) to (10);
		\draw (11) to (6.center);
		\draw (10) to (5.center);
		\draw (14.center) to (12);
		\draw (12) to (13.center);
	\end{pgfonlayer}
\end{tikzpicture}
    \end{equation}
    These are combined by connecting the variables to the clauses that they occur in. For unnegated literals, they are connected with a wire. For negated literals, they are connected via an X-spider. This yields the following
    \begin{equation}
        \begin{tikzpicture}
	\begin{pgfonlayer}{nodelayer}
		\node [style=none] (0) at (0, 0) {$\scount{\phi} =$};
		\node [style={Z dot (zh)}] (1) at (3, -1.75) {};
		\node [style=labelled hbox] (5) at (3, 2) {$0$};
		\node [style=labelled hbox] (6) at (6, 2) {$0$};
		\node [style=none] (7) at (1.5, 0.75) {};
		\node [style=none] (8) at (7.5, 0.75) {};
		\node [style=none] (9) at (1.5, -0.75) {};
		\node [style=none] (10) at (7.5, -0.75) {};
		\node [style=none] (11) at (4.5, 0) {$G$};
		\node [style=none] (14) at (7, -0.75) {};
		\node [style=none] (15) at (2, -0.75) {};
		\node [style=none] (16) at (4, -0.75) {};
		\node [style=none] (17) at (5, -0.75) {};
		\node [style={Z dot (zh)}] (18) at (6, -1.75) {};
		\node [style={X dot (zh)}] (21) at (5, 1.25) {$\pi$};
		\node [style={X dot (zh)}] (22) at (2, 1.25) {$\pi$};
		\node [style={X dot (zh)}] (23) at (4, 1.25) {$\pi$};
		\node [style={X dot (zh)}] (24) at (7, 1.25) {$\pi$};
		\node [style=none] (25) at (2, 0.75) {};
		\node [style=none] (26) at (4, 0.75) {};
		\node [style=none] (27) at (5, 0.75) {};
		\node [style=none] (28) at (7, 0.75) {};
		\node [style=none] (29) at (6, 1.25) {$\dots$};
		\node [style=none] (30) at (3, 1.25) {$\dots$};
		\node [style=none] (31) at (4.5, 2) {$\dots$};
		\node [style=none] (32) at (4.5, -1.5) {$\dots$};
		\node [style=none] (33) at (6, -1) {$\dots$};
		\node [style=none] (34) at (3, -1) {$\dots$};
	\end{pgfonlayer}
	\begin{pgfonlayer}{edgelayer}
		\draw (8.center) to (7.center);
		\draw (10.center) to (8.center);
		\draw (10.center) to (9.center);
		\draw (9.center) to (7.center);
		\draw [bend right] (15.center) to (1);
		\draw [bend right] (1) to (16.center);
		\draw [bend left] (18) to (17.center);
		\draw [bend right] (18) to (14.center);
		\draw [bend right] (5) to (22);
		\draw [bend left] (6) to (24);
		\draw [bend left] (5) to (23);
		\draw [bend right] (6) to (21);
		\draw (21) to (27.center);
		\draw (24) to (28.center);
		\draw (23) to (26.center);
		\draw (22) to (25.center);
	\end{pgfonlayer}
\end{tikzpicture}
    \end{equation}
    where $G$ defines the connections. The scalar value of this diagram is exactly equal to the number of satisfying assignments to $\phi$. This can be simplified by canceling all adjacent X-spiders, leaving an X-spider connected between a variable and a clause if and only if that variable is contained in that clause \emph{unnegated}. For example, for $\phi(x_1, x_2, x_3, x_4) = (\lnot x_1 \lor \lnot x_2 \lor \lnot x_3) \land (x_2 \lor x_3) \land (\lnot x_1 \lor x_3) \land (x_3 \lor x_4)$, we have:
    \begin{equation}
        \begin{tikzpicture}
	\begin{pgfonlayer}{nodelayer}
		\node [style={Z dot (zh)}] (0) at (1, -1.25) {};
		\node [style={Z dot (zh)}] (1) at (2.75, -1.25) {};
		\node [style={Z dot (zh)}] (2) at (4.5, -1.25) {};
		\node [style=labelled hbox] (3) at (1, 1.25) {$0$};
		\node [style=labelled hbox] (4) at (2.75, 1.25) {$0$};
		\node [style=labelled hbox] (5) at (4.5, 1.25) {$0$};
		\node [style={X dot (zh)}] (6) at (4.5, 0) {$\pi$};
		\node [style={X dot (zh)}] (7) at (2.75, -0.5) {$\pi$};
		\node [style={X dot (zh)}] (8) at (3.75, 0) {$\pi$};
		\node [style=none] (9) at (-1, 0) {$\scount{\phi} = $};
		\node [style=none] (10) at (1, -2) {$x_1$};
		\node [style=none] (11) at (2.75, -2) {$x_2$};
		\node [style=none] (12) at (4.5, -2) {$x_3$};
		\node [style={Z dot (zh)}] (13) at (6.25, -1.25) {};
		\node [style=none] (14) at (6.25, -2) {$x_4$};
		\node [style=labelled hbox] (15) at (6.25, 1.25) {$0$};
		\node [style={X dot (zh)}] (16) at (5.5, 0) {$\pi$};
		\node [style={X dot (zh)}] (17) at (6.25, 0) {$\pi$};
	\end{pgfonlayer}
	\begin{pgfonlayer}{edgelayer}
		\draw (1) to (7);
		\draw (7) to (4);
		\draw (4) to (8);
		\draw (8) to (2);
		\draw (3) to (1);
		\draw (2) to (3);
		\draw (5) to (0);
		\draw (6) to (5);
		\draw (6) to (2);
		\draw (0) to (3);
		\draw (16) to (2);
		\draw (16) to (15);
		\draw (15) to (17);
		\draw (17) to (13);
	\end{pgfonlayer}
\end{tikzpicture}
    \end{equation}
    Therefore, a diagram derived this way has $m$ H-boxes, $n$ Z-spiders, at most $L$ X-spiders, and the maximum clause size $k$ corresponds to the maximum degree of any H-box. Note that a complete graphical calculus for \SAT was introduced recently in~\cite{gu2022complete}. However, the semantics of their diagrams directly correspond to a matrix of True or False values, and hence cannot represent \sSAT instances directly. In this sense it is similar to the modified ZH-diagrams of~\cite{de_beaudrap_tensor_2021} where they set $2=1$ and consider diagrams over the Boolean semi-ring.

\subsection{The CDP algorithm for \#SAT}

The Counting David-Putnam (CDP) algorithm is an algorithm for solving $\sSAT$ that was first introduced in 1999 \cite{birnbaum_good_1999}, as an extension of the DPLL algorithm \cite{davis_machine_1962} for $\SAT$ solving. It is effectively an optimized depth-first search over all possible assignments of variables in a Boolean formula. The algorithm is based on the following two rules:
\begin{enumerate}
    \item \textit{Unit Propagation:} The following rewrite holds for any clauses $A_i$ and $B_i$ not containing some literal $x$:
    \begin{equation}
        \begin{aligned}
            x \land (A_1 \lor x) \land \cdots \land (A_n \lor x) \land (B_1 \lor \lnot x) \land \cdots \land (B_m \lor \lnot x) = B_1 \land \cdots \land B_m
        \end{aligned}
    \end{equation}
    \item \textit{Variable Branching:} For any variable $x$ appearing in a formula $f$, the number of satisfying assignments of $f$ is the sum of the numbers of satisfying assignments for
    \begin{equation}
        f_1 = f \land x \qquad \text{and} \qquad f_2 = f \land \lnot x
    \end{equation}
    since in each such assignment, either $x$ or $\lnot x$.
\end{enumerate}

The CDP algorithm, given in Algorithm~\ref{cdp_algo}, applies these two rules recursively until either a contradiction occurs and so the formula is unsatisfiable, or the formula has no clauses remaining, in which case the number of satisfying assignments is $2^n$. Clearly, at each recursive step, the formulas to be considered have fewer variables than at the previous steps, so this procedure terminates. 

This works better in practice than naively checking every assignment because unit propagation can eliminate many variables, thus removing whole branches from the computation tree. Additionally, the substitution of the assignment into the formula is done incrementally via unit propagation, rather than repeated for every assignment.

\begin{algorithm}[t]
    \DontPrintSemicolon
    \caption{The CDP \cite{birnbaum_good_1999} algorithm for solving $\sSAT$. \label{cdp_algo}}
    \KwIn{A CNF formula $f$ with $n$ variables and $m$ clauses.}
    \KwOut{The value of $\#\{ \vec{x} \in \{0, 1\}^n \mid f(\vec{x}) = 1 \}$.}
    \If{$f$ contains the clauses $x$ and $\lnot x$ for some variable $x$}{
        \Return{$0$}\;
    }
    \If{$f$ has no clauses remaining}{
        \Return{$2^n$}\;
    }
    Apply unit propagation to $f$ until it is no longer possible.\;
    Pick a variable $x$ that occurs in $f$ and generate $f_1 = f \land x$ and $f_2 = f \land \lnot x$.\;
    \Return{$\mathrm{CDP}(f_1) + \mathrm{CDP}(f_2)$}\;
\end{algorithm}

\section{Results}\label{sec:results}


\subsection{Interpreting CDP diagrammatically}

To interpret CDP diagrammatically we will first see how the two rules apply to the ZH-diagrams for $\sSAT$ instances detailed in Section \ref{sec:embedding}. The proofs of these lemmas and all the following results are postponed to Appendix \ref{sec:proofs}.

\begin{lemmarep}
    \label{dpll_unitprop_overview}
    The following diagrammatic equivalent to the unit propagation rule holds (without loss of generality, we assume the literal to be propagated is not negated):
    
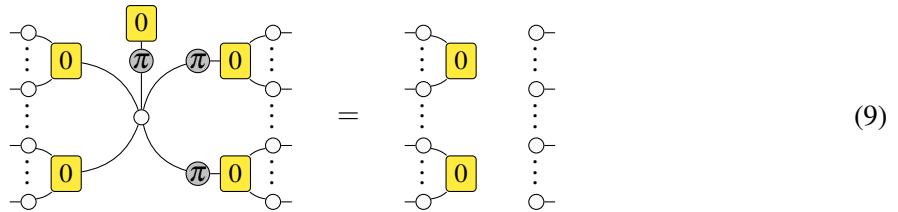
\begin{equation}
        \begin{tikzpicture}
	\begin{pgfonlayer}{nodelayer}
		\node [style={Z dot (zh)}] (0) at (-0.5, 0) {};
		\node [style=labelled hbox] (1) at (-0.5, 2.5) {$0$};
		\node [style={X dot (zh)}] (2) at (-0.5, 1.5) {$\pi$};
		\node [style={X dot (zh)}] (3) at (1, 1.5) {$\pi$};
		\node [style={X dot (zh)}] (4) at (1, -1.5) {$\pi$};
		\node [style=labelled hbox] (5) at (2, 1.5) {$0$};
		\node [style=labelled hbox] (6) at (2, -1.5) {$0$};
		\node [style={Z dot (zh)}] (8) at (3, 0.75) {};
		\node [style={Z dot (zh)}] (9) at (3, -0.75) {};
		\node [style={Z dot (zh)}] (10) at (3, -2.25) {};
		\node [style={Z dot (zh)}] (11) at (3, 2.25) {};
		\node [style={Z dot (zh)}] (16) at (-3.5, 0.75) {};
		\node [style={Z dot (zh)}] (17) at (-3.5, -0.75) {};
		\node [style={Z dot (zh)}] (18) at (-3.5, -2.25) {};
		\node [style={Z dot (zh)}] (19) at (-3.5, 2.25) {};
		\node [style=labelled hbox] (20) at (-2.5, 1.5) {$0$};
		\node [style=labelled hbox] (21) at (-2.5, -1.5) {$0$};
		\node [style=none] (22) at (3, 1.75) {$\vdots$};
		\node [style=none] (23) at (3, -1.25) {$\vdots$};
		\node [style=none] (24) at (-3.5, -1.25) {$\vdots$};
		\node [style=none] (25) at (-3.5, 1.75) {$\vdots$};
		\node [style=none] (26) at (3, 0.25) {$\vdots$};
		\node [style=none] (27) at (-3.5, 0.25) {$\vdots$};
		\node [style=none] (51) at (5, 0) {$=$};
		\node [style=none] (52) at (-4, 0.75) {};
		\node [style=none] (53) at (-4, -0.75) {};
		\node [style=none] (54) at (-4, -2.25) {};
		\node [style=none] (55) at (-4, 2.25) {};
		\node [style=none] (56) at (3.5, 0.75) {};
		\node [style=none] (57) at (3.5, -0.75) {};
		\node [style=none] (58) at (3.5, -2.25) {};
		\node [style=none] (59) at (3.5, 2.25) {};
		\node [style={Z dot (zh)}] (67) at (10, 0.75) {};
		\node [style={Z dot (zh)}] (68) at (10, -0.75) {};
		\node [style={Z dot (zh)}] (69) at (10, -2.25) {};
		\node [style={Z dot (zh)}] (70) at (10, 2.25) {};
		\node [style={Z dot (zh)}] (71) at (7, 0.75) {};
		\node [style={Z dot (zh)}] (72) at (7, -0.75) {};
		\node [style={Z dot (zh)}] (73) at (7, -2.25) {};
		\node [style={Z dot (zh)}] (74) at (7, 2.25) {};
		\node [style=labelled hbox] (75) at (8, 1.5) {$0$};
		\node [style=labelled hbox] (76) at (8, -1.5) {$0$};
		\node [style=none] (77) at (10, 1.75) {$\vdots$};
		\node [style=none] (78) at (10, -1.25) {$\vdots$};
		\node [style=none] (79) at (7, -1.25) {$\vdots$};
		\node [style=none] (80) at (7, 1.75) {$\vdots$};
		\node [style=none] (81) at (10, 0.25) {$\vdots$};
		\node [style=none] (82) at (7, 0.25) {$\vdots$};
		\node [style=none] (83) at (6.5, 0.75) {};
		\node [style=none] (84) at (6.5, -0.75) {};
		\node [style=none] (85) at (6.5, -2.25) {};
		\node [style=none] (86) at (6.5, 2.25) {};
		\node [style=none] (87) at (10.5, 0.75) {};
		\node [style=none] (88) at (10.5, -0.75) {};
		\node [style=none] (89) at (10.5, -2.25) {};
		\node [style=none] (90) at (10.5, 2.25) {};
	\end{pgfonlayer}
	\begin{pgfonlayer}{edgelayer}
		\draw (1) to (2);
		\draw (2) to (0);
		\draw [bend left] (0) to (3);
		\draw [bend right] (0) to (4);
		\draw [bend left=15] (5) to (11);
		\draw [bend right=15] (5) to (8);
		\draw (3) to (5);
		\draw (4) to (6);
		\draw [bend left=15] (6) to (9);
		\draw [bend right=15] (6) to (10);
		\draw [bend right] (0) to (20);
		\draw [bend left] (0) to (21);
		\draw [bend right=15] (21) to (17);
		\draw [bend left=15] (21) to (18);
		\draw [bend right=15] (16) to (20);
		\draw [bend right=15] (20) to (19);
		\draw (58.center) to (10);
		\draw (9) to (57.center);
		\draw (56.center) to (8);
		\draw (11) to (59.center);
		\draw (18) to (54.center);
		\draw (17) to (53.center);
		\draw (16) to (52.center);
		\draw (19) to (55.center);
		\draw [bend right=15] (76) to (72);
		\draw [bend left=15] (76) to (73);
		\draw [bend right=15] (71) to (75);
		\draw [bend right=15] (75) to (74);
		\draw (89.center) to (69);
		\draw (68) to (88.center);
		\draw (87.center) to (67);
		\draw (70) to (90.center);
		\draw (73) to (85.center);
		\draw (72) to (84.center);
		\draw (71) to (83.center);
		\draw (74) to (86.center);
	\end{pgfonlayer}
\end{tikzpicture}
    \end{equation}
    It can be read as follows - on the left-hand side, the zero H-box with one leg is the clause with a single non-negated literal $x$, the H-boxes on the left represent the clauses $B_i \lor \lnot x$ while the H-boxes on the right represent the clauses $A_i \lor x$. On the right-hand side, we see that the clauses $B_i$ remain, while the clauses $A_i \lor x$ have been removed entirely. Because the variable $x$ is now no longer mentioned, the Z-spider representing it is also removed.
\end{lemmarep}
\begin{proof}
    First, note that
    \begin{equation}
        \label{dpll_statepush}
        \begin{tikzpicture}
	\begin{pgfonlayer}{nodelayer}
		\node [style={X dot (zh)}] (0) at (-9.25, 2.5) {};
		\node [style=labelled hbox] (1) at (-8.25, 2.5) {$0$};
		\node [style={Z dot (zh)}] (2) at (-7.25, 1.75) {};
		\node [style={Z dot (zh)}] (3) at (-7.25, 3.25) {};
		\node [style=none] (4) at (-7.25, 2.75) {$\vdots$};
		\node [style=none] (5) at (-6.75, 1.75) {};
		\node [style=none] (6) at (-6.75, 3.25) {};
		\node [style={X dot (zh)}] (7) at (-5, 2.5) {};
		\node [style=hbox] (8) at (-4, 2.5) {};
		\node [style={Z dot (zh)}] (9) at (-3, 1.75) {};
		\node [style={Z dot (zh)}] (10) at (-3, 3.25) {};
		\node [style=none] (11) at (-3, 2.75) {$\vdots$};
		\node [style=none] (12) at (-2.5, 1.75) {};
		\node [style=none] (13) at (-2.5, 3.25) {};
		\node [style=none] (14) at (-6, 2.75) {$\stackrel{SF_H}{=}$};
		\node [style=hbox] (15) at (-4, 3.25) {};
		\node [style=labelled hbox] (16) at (-4, 4.25) {$0$};
		\node [style={Z dot (zh)}] (17) at (1.25, 1.75) {};
		\node [style={Z dot (zh)}] (18) at (1.25, 3.25) {};
		\node [style=none] (19) at (1.25, 2.75) {$\vdots$};
		\node [style=none] (20) at (1.75, 1.75) {};
		\node [style=none] (21) at (1.75, 3.25) {};
		\node [style=none] (22) at (-1.5, 2.75) {$\stackrel{BA_H}{=}$};
		\node [style=hbox] (23) at (-0.25, 3.25) {};
		\node [style=labelled hbox] (24) at (-0.25, 4.25) {$0$};
		\node [style={Z dot (zh)}] (25) at (0.5, 3.25) {};
		\node [style={Z dot (zh)}] (26) at (0.5, 1.75) {};
		\node [style={Z dot (zh)}] (27) at (-0.25, 2.5) {};
		\node [style={Z dot (zh)}] (28) at (4.75, 1.75) {};
		\node [style={Z dot (zh)}] (29) at (4.75, 3.25) {};
		\node [style=none] (30) at (4.75, 2.75) {$\vdots$};
		\node [style=none] (31) at (5.25, 1.75) {};
		\node [style=none] (32) at (5.25, 3.25) {};
		\node [style=none] (33) at (2.5, 2.75) {$\stackrel{SF_Z}{=}$};
		\node [style=labelled hbox] (35) at (3.75, 4.25) {$0$};
		\node [style={X dot (zh)}] (36) at (3.75, 3) {};
		\node [style=labelled hbox] (37) at (-3, 4.25) {$\frac{1}{2}$};
		\node [style=labelled hbox] (38) at (0.75, 4.25) {$\frac{1}{2}$};
		\node [style={Z dot (zh)}] (47) at (7, 1.75) {};
		\node [style={Z dot (zh)}] (48) at (7, 3.25) {};
		\node [style=none] (49) at (7, 2.75) {$\vdots$};
		\node [style=none] (50) at (7.5, 1.75) {};
		\node [style=none] (51) at (7.5, 3.25) {};
		\node [style=none] (52) at (6, 2.5) {$=$};
		\node [style={Z dot (zh)}] (53) at (-7.25, 0) {};
		\node [style={Z dot (zh)}] (54) at (-7.25, -1.5) {};
		\node [style=labelled hbox] (55) at (-8.25, -0.75) {$0$};
		\node [style=none] (56) at (-7.25, -0.5) {$\vdots$};
		\node [style=none] (57) at (-6.75, 0) {};
		\node [style=none] (58) at (-6.75, -1.5) {};
		\node [style={X dot (zh)}] (59) at (-9.25, -0.75) {$\pi$};
		\node [style={Z dot (zh)}] (60) at (-2, -0.25) {};
		\node [style={Z dot (zh)}] (61) at (-2, -1.75) {};
		\node [style=labelled hbox] (62) at (-3, -1) {$0$};
		\node [style=none] (63) at (-2, -0.75) {$\vdots$};
		\node [style=none] (64) at (-1.5, -0.25) {};
		\node [style=none] (65) at (-1.5, -1.75) {};
		\node [style=none] (66) at (-6, -0.75) {$=$};
		\node [style=hbox] (67) at (-4, -1) {};
		\node [style=hbox] (68) at (-5, -1) {};
		\node [style={Z dot (zh)}] (69) at (-4.5, -1) {};
		\node [style={Z dot (zh)}] (70) at (3.25, -0.25) {};
		\node [style={Z dot (zh)}] (71) at (3.25, -1.75) {};
		\node [style=labelled hbox] (72) at (2.25, -1) {$0$};
		\node [style=none] (73) at (3.25, -0.75) {$\vdots$};
		\node [style=none] (74) at (3.75, -0.25) {};
		\node [style=none] (75) at (3.75, -1.75) {};
		\node [style=none] (76) at (-0.5, -0.75) {$\stackrel{I_Z}{=}$};
		\node [style=hbox] (77) at (1.25, -1) {};
		\node [style=hbox] (78) at (0.75, -1) {};
		\node [style={Z dot (zh)}] (79) at (7.25, 0) {};
		\node [style={Z dot (zh)}] (80) at (7.25, -1.5) {};
		\node [style=labelled hbox] (81) at (6.25, -0.75) {$0$};
		\node [style=none] (82) at (7.25, -0.5) {$\vdots$};
		\node [style=none] (83) at (7.75, 0) {};
		\node [style=none] (84) at (7.75, -1.5) {};
		\node [style=none] (85) at (4.75, -0.75) {$\stackrel{SF_H}{=}$};
		\node [style=labelled hbox] (86) at (-4.5, 0) {$\frac{1}{2}$};
		\node [style=labelled hbox] (87) at (1, 0) {$\frac{1}{2}$};
	\end{pgfonlayer}
	\begin{pgfonlayer}{edgelayer}
		\draw [bend left=15] (1) to (3);
		\draw [bend right=15] (1) to (2);
		\draw (0) to (1);
		\draw (5.center) to (2);
		\draw (3) to (6.center);
		\draw [bend left=15] (8) to (10);
		\draw [bend right=15] (8) to (9);
		\draw (7) to (8);
		\draw (12.center) to (9);
		\draw (10) to (13.center);
		\draw (16) to (15);
		\draw (15) to (8);
		\draw (20.center) to (17);
		\draw (18) to (21.center);
		\draw (24) to (23);
		\draw (27) to (23);
		\draw (25) to (18);
		\draw (17) to (26);
		\draw (31.center) to (28);
		\draw (29) to (32.center);
		\draw (50.center) to (47);
		\draw (48) to (51.center);
		\draw [bend left=15] (55) to (53);
		\draw [bend right=15] (55) to (54);
		\draw (54) to (58.center);
		\draw (53) to (57.center);
		\draw (59) to (55);
		\draw [bend left=15] (62) to (60);
		\draw [bend right=15] (62) to (61);
		\draw (61) to (65.center);
		\draw (60) to (64.center);
		\draw (68) to (69);
		\draw (69) to (67);
		\draw (67) to (62);
		\draw [bend left=15] (72) to (70);
		\draw [bend right=15] (72) to (71);
		\draw (71) to (75.center);
		\draw (70) to (74.center);
		\draw (77) to (72);
		\draw (78) to (77);
		\draw [bend left=15] (81) to (79);
		\draw [bend right=15] (81) to (80);
		\draw (80) to (84.center);
		\draw (79) to (83.center);
		\draw (35) to (36);
	\end{pgfonlayer}
\end{tikzpicture}
    \end{equation}
    and so we have that
    \begin{equation}
        \begin{tikzpicture}
	\begin{pgfonlayer}{nodelayer}
		\node [style={Z dot (zh)}] (0) at (0.5, 3) {};
		\node [style=labelled hbox] (1) at (0.5, 5.5) {$0$};
		\node [style={X dot (zh)}] (2) at (0.5, 4.5) {$\pi$};
		\node [style={X dot (zh)}] (3) at (2, 4.5) {$\pi$};
		\node [style={X dot (zh)}] (4) at (2, 1.5) {$\pi$};
		\node [style=labelled hbox] (5) at (3, 4.5) {$0$};
		\node [style=labelled hbox] (6) at (3, 1.5) {$0$};
		\node [style={Z dot (zh)}] (7) at (4, 3.75) {};
		\node [style={Z dot (zh)}] (8) at (4, 2.25) {};
		\node [style={Z dot (zh)}] (9) at (4, 0.75) {};
		\node [style={Z dot (zh)}] (10) at (4, 5.25) {};
		\node [style={Z dot (zh)}] (11) at (-2, 3.75) {};
		\node [style={Z dot (zh)}] (12) at (-2, 2.25) {};
		\node [style={Z dot (zh)}] (13) at (-2, 0.75) {};
		\node [style={Z dot (zh)}] (14) at (-2, 5.25) {};
		\node [style=labelled hbox] (15) at (-1, 4.5) {$0$};
		\node [style=labelled hbox] (16) at (-1, 1.5) {$0$};
		\node [style=none] (17) at (4, 4.75) {$\vdots$};
		\node [style=none] (18) at (4, 1.75) {$\vdots$};
		\node [style=none] (19) at (-2, 1.75) {$\vdots$};
		\node [style=none] (20) at (-2, 4.75) {$\vdots$};
		\node [style=none] (21) at (4, 3.25) {$\vdots$};
		\node [style=none] (22) at (-2, 3.25) {$\vdots$};
		\node [style=none] (23) at (-2.5, 3.75) {};
		\node [style=none] (24) at (-2.5, 2.25) {};
		\node [style=none] (25) at (-2.5, 0.75) {};
		\node [style=none] (26) at (-2.5, 5.25) {};
		\node [style=none] (27) at (4.5, 3.75) {};
		\node [style=none] (28) at (4.5, 2.25) {};
		\node [style=none] (29) at (4.5, 0.75) {};
		\node [style=none] (30) at (4.5, 5.25) {};
		\node [style={Z dot (zh)}] (31) at (9.5, 3) {};
		\node [style={X dot (zh)}] (33) at (9.5, 4.5) {$\pi$};
		\node [style={X dot (zh)}] (34) at (11, 4.5) {$\pi$};
		\node [style={X dot (zh)}] (35) at (11, 1.5) {$\pi$};
		\node [style=labelled hbox] (36) at (12, 4.5) {$0$};
		\node [style=labelled hbox] (37) at (12, 1.5) {$0$};
		\node [style={Z dot (zh)}] (38) at (13, 3.75) {};
		\node [style={Z dot (zh)}] (39) at (13, 2.25) {};
		\node [style={Z dot (zh)}] (40) at (13, 0.75) {};
		\node [style={Z dot (zh)}] (41) at (13, 5.25) {};
		\node [style={Z dot (zh)}] (42) at (7, 3.75) {};
		\node [style={Z dot (zh)}] (43) at (7, 2.25) {};
		\node [style={Z dot (zh)}] (44) at (7, 0.75) {};
		\node [style={Z dot (zh)}] (45) at (7, 5.25) {};
		\node [style=labelled hbox] (46) at (8, 4.5) {$0$};
		\node [style=labelled hbox] (47) at (8, 1.5) {$0$};
		\node [style=none] (48) at (13, 4.75) {$\vdots$};
		\node [style=none] (49) at (13, 1.75) {$\vdots$};
		\node [style=none] (50) at (7, 1.75) {$\vdots$};
		\node [style=none] (51) at (7, 4.75) {$\vdots$};
		\node [style=none] (52) at (13, 3.25) {$\vdots$};
		\node [style=none] (53) at (7, 3.25) {$\vdots$};
		\node [style=none] (54) at (6.5, 3.75) {};
		\node [style=none] (55) at (6.5, 2.25) {};
		\node [style=none] (56) at (6.5, 0.75) {};
		\node [style=none] (57) at (6.5, 5.25) {};
		\node [style=none] (58) at (13.5, 3.75) {};
		\node [style=none] (59) at (13.5, 2.25) {};
		\node [style=none] (60) at (13.5, 0.75) {};
		\node [style=none] (61) at (13.5, 5.25) {};
		\node [style=none] (62) at (5.5, 3.25) {$\stackrel{\eqref{satembed_zerostate}}{=}$};
		\node [style={X dot (zh)}] (65) at (19, 4.5) {};
		\node [style={X dot (zh)}] (66) at (19, 1.5) {};
		\node [style=labelled hbox] (67) at (20, 4.5) {$0$};
		\node [style=labelled hbox] (68) at (20, 1.5) {$0$};
		\node [style={Z dot (zh)}] (69) at (21, 3.75) {};
		\node [style={Z dot (zh)}] (70) at (21, 2.25) {};
		\node [style={Z dot (zh)}] (71) at (21, 0.75) {};
		\node [style={Z dot (zh)}] (72) at (21, 5.25) {};
		\node [style={Z dot (zh)}] (73) at (16, 3.75) {};
		\node [style={Z dot (zh)}] (74) at (16, 2.25) {};
		\node [style={Z dot (zh)}] (75) at (16, 0.75) {};
		\node [style={Z dot (zh)}] (76) at (16, 5.25) {};
		\node [style=labelled hbox] (77) at (17, 4.5) {$0$};
		\node [style=labelled hbox] (78) at (17, 1.5) {$0$};
		\node [style=none] (79) at (21, 4.75) {$\vdots$};
		\node [style=none] (80) at (21, 1.75) {$\vdots$};
		\node [style=none] (81) at (16, 1.75) {$\vdots$};
		\node [style=none] (82) at (16, 4.75) {$\vdots$};
		\node [style=none] (83) at (21, 3.25) {$\vdots$};
		\node [style=none] (84) at (16, 3.25) {$\vdots$};
		\node [style=none] (85) at (15.5, 3.75) {};
		\node [style=none] (86) at (15.5, 2.25) {};
		\node [style=none] (87) at (15.5, 0.75) {};
		\node [style=none] (88) at (15.5, 5.25) {};
		\node [style=none] (89) at (21.5, 3.75) {};
		\node [style=none] (90) at (21.5, 2.25) {};
		\node [style=none] (91) at (21.5, 0.75) {};
		\node [style=none] (92) at (21.5, 5.25) {};
		\node [style={X dot (zh)}] (95) at (18, 4.5) {$\pi$};
		\node [style={X dot (zh)}] (96) at (18, 1.5) {$\pi$};
		\node [style=none] (126) at (22.5, 3.25) {$\stackrel{\eqref{dpll_statepush}}{=}$};
		\node [style={Z dot (zh)}] (131) at (26.5, 3.75) {};
		\node [style={Z dot (zh)}] (132) at (26.5, 2.25) {};
		\node [style={Z dot (zh)}] (133) at (26.5, 0.75) {};
		\node [style={Z dot (zh)}] (134) at (26.5, 5.25) {};
		\node [style={Z dot (zh)}] (135) at (24, 3.75) {};
		\node [style={Z dot (zh)}] (136) at (24, 2.25) {};
		\node [style={Z dot (zh)}] (137) at (24, 0.75) {};
		\node [style={Z dot (zh)}] (138) at (24, 5.25) {};
		\node [style=labelled hbox] (139) at (25, 4.5) {$0$};
		\node [style=labelled hbox] (140) at (25, 1.5) {$0$};
		\node [style=none] (141) at (26.5, 4.75) {$\vdots$};
		\node [style=none] (142) at (26.5, 1.75) {$\vdots$};
		\node [style=none] (143) at (24, 1.75) {$\vdots$};
		\node [style=none] (144) at (24, 4.75) {$\vdots$};
		\node [style=none] (145) at (26.5, 3.25) {$\vdots$};
		\node [style=none] (146) at (24, 3.25) {$\vdots$};
		\node [style=none] (147) at (23.5, 3.75) {};
		\node [style=none] (148) at (23.5, 2.25) {};
		\node [style=none] (149) at (23.5, 0.75) {};
		\node [style=none] (150) at (23.5, 5.25) {};
		\node [style=none] (151) at (27, 3.75) {};
		\node [style=none] (152) at (27, 2.25) {};
		\node [style=none] (153) at (27, 0.75) {};
		\node [style=none] (154) at (27, 5.25) {};
		\node [style=none] (215) at (14.5, 3.25) {$\stackrel{\eqref{satembed_copy}}{=}$};
	\end{pgfonlayer}
	\begin{pgfonlayer}{edgelayer}
		\draw (1) to (2);
		\draw (2) to (0);
		\draw [bend left] (0) to (3);
		\draw [bend right] (0) to (4);
		\draw [bend left=15] (5) to (10);
		\draw [bend right=15] (5) to (7);
		\draw (3) to (5);
		\draw (4) to (6);
		\draw [bend left=15] (6) to (8);
		\draw [bend right=15] (6) to (9);
		\draw [bend right] (0) to (15);
		\draw [bend left] (0) to (16);
		\draw [bend right=15] (16) to (12);
		\draw [bend left=15] (16) to (13);
		\draw [bend right=15] (11) to (15);
		\draw [bend right=15] (15) to (14);
		\draw (29.center) to (9);
		\draw (8) to (28.center);
		\draw (27.center) to (7);
		\draw (10) to (30.center);
		\draw (13) to (25.center);
		\draw (12) to (24.center);
		\draw (11) to (23.center);
		\draw (14) to (26.center);
		\draw (33) to (31);
		\draw [bend left] (31) to (34);
		\draw [bend right] (31) to (35);
		\draw [bend left=15] (36) to (41);
		\draw [bend right=15] (36) to (38);
		\draw (34) to (36);
		\draw (35) to (37);
		\draw [bend left=15] (37) to (39);
		\draw [bend right=15] (37) to (40);
		\draw [bend right] (31) to (46);
		\draw [bend left] (31) to (47);
		\draw [bend right=15] (47) to (43);
		\draw [bend left=15] (47) to (44);
		\draw [bend right=15] (42) to (46);
		\draw [bend right=15] (46) to (45);
		\draw (60.center) to (40);
		\draw (39) to (59.center);
		\draw (58.center) to (38);
		\draw (41) to (61.center);
		\draw (44) to (56.center);
		\draw (43) to (55.center);
		\draw (42) to (54.center);
		\draw (45) to (57.center);
		\draw [bend left=15] (67) to (72);
		\draw [bend right=15] (67) to (69);
		\draw (65) to (67);
		\draw (66) to (68);
		\draw [bend left=15] (68) to (70);
		\draw [bend right=15] (68) to (71);
		\draw [bend right=15] (78) to (74);
		\draw [bend left=15] (78) to (75);
		\draw [bend right=15] (73) to (77);
		\draw [bend right=15] (77) to (76);
		\draw (91.center) to (71);
		\draw (70) to (90.center);
		\draw (89.center) to (69);
		\draw (72) to (92.center);
		\draw (75) to (87.center);
		\draw (74) to (86.center);
		\draw (73) to (85.center);
		\draw (76) to (88.center);
		\draw (95) to (77);
		\draw (96) to (78);
		\draw [bend right=15] (140) to (136);
		\draw [bend left=15] (140) to (137);
		\draw [bend right=15] (135) to (139);
		\draw [bend right=15] (139) to (138);
		\draw (153.center) to (133);
		\draw (132) to (152.center);
		\draw (151.center) to (131);
		\draw (134) to (154.center);
		\draw (137) to (149.center);
		\draw (136) to (148.center);
		\draw (135) to (147.center);
		\draw (138) to (150.center);
	\end{pgfonlayer}
\end{tikzpicture}
    \end{equation}
    which completes the proof.
\end{proof}

\begin{lemmarep}
    \label{dpll_branch_statement}
    The following diagrammatic equivalent to the variable branching rule holds:
    
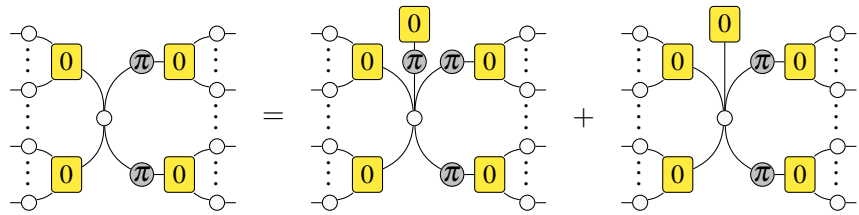
\begin{equation*}
        \begin{tikzpicture}
	\begin{pgfonlayer}{nodelayer}
		\node [style={Z dot (zh)}] (0) at (-1, 3) {};
		\node [style={X dot (zh)}] (3) at (0, 4.5) {$\pi$};
		\node [style={X dot (zh)}] (4) at (0, 1.5) {$\pi$};
		\node [style=labelled hbox] (5) at (1, 4.5) {$0$};
		\node [style=labelled hbox] (6) at (1, 1.5) {$0$};
		\node [style={Z dot (zh)}] (7) at (2, 3.75) {};
		\node [style={Z dot (zh)}] (8) at (2, 2.25) {};
		\node [style={Z dot (zh)}] (9) at (2, 0.75) {};
		\node [style={Z dot (zh)}] (10) at (2, 5.25) {};
		\node [style={Z dot (zh)}] (11) at (-3, 3.75) {};
		\node [style={Z dot (zh)}] (12) at (-3, 2.25) {};
		\node [style={Z dot (zh)}] (13) at (-3, 0.75) {};
		\node [style={Z dot (zh)}] (14) at (-3, 5.25) {};
		\node [style=labelled hbox] (15) at (-2, 4.5) {$0$};
		\node [style=labelled hbox] (16) at (-2, 1.5) {$0$};
		\node [style=none] (17) at (2, 4.75) {$\vdots$};
		\node [style=none] (18) at (2, 1.75) {$\vdots$};
		\node [style=none] (19) at (-3, 1.75) {$\vdots$};
		\node [style=none] (20) at (-3, 4.75) {$\vdots$};
		\node [style=none] (21) at (2, 3.25) {$\vdots$};
		\node [style=none] (22) at (-3, 3.25) {$\vdots$};
		\node [style=none] (23) at (-3.5, 3.75) {};
		\node [style=none] (24) at (-3.5, 2.25) {};
		\node [style=none] (25) at (-3.5, 0.75) {};
		\node [style=none] (26) at (-3.5, 5.25) {};
		\node [style=none] (27) at (2.5, 3.75) {};
		\node [style=none] (28) at (2.5, 2.25) {};
		\node [style=none] (29) at (2.5, 0.75) {};
		\node [style=none] (30) at (2.5, 5.25) {};
		\node [style={Z dot (zh)}] (31) at (7.25, 3) {};
		\node [style=labelled hbox] (32) at (7.25, 5.5) {$0$};
		\node [style={X dot (zh)}] (33) at (7.25, 4.5) {$\pi$};
		\node [style={X dot (zh)}] (34) at (8.25, 4.5) {$\pi$};
		\node [style={X dot (zh)}] (35) at (8.25, 1.5) {$\pi$};
		\node [style=labelled hbox] (36) at (9.25, 4.5) {$0$};
		\node [style=labelled hbox] (37) at (9.25, 1.5) {$0$};
		\node [style={Z dot (zh)}] (38) at (10.25, 3.75) {};
		\node [style={Z dot (zh)}] (39) at (10.25, 2.25) {};
		\node [style={Z dot (zh)}] (40) at (10.25, 0.75) {};
		\node [style={Z dot (zh)}] (41) at (10.25, 5.25) {};
		\node [style={Z dot (zh)}] (42) at (5, 3.75) {};
		\node [style={Z dot (zh)}] (43) at (5, 2.25) {};
		\node [style={Z dot (zh)}] (44) at (5, 0.75) {};
		\node [style={Z dot (zh)}] (45) at (5, 5.25) {};
		\node [style=labelled hbox] (46) at (6, 4.5) {$0$};
		\node [style=labelled hbox] (47) at (6, 1.5) {$0$};
		\node [style=none] (48) at (10.25, 4.75) {$\vdots$};
		\node [style=none] (49) at (10.25, 1.75) {$\vdots$};
		\node [style=none] (50) at (5, 1.75) {$\vdots$};
		\node [style=none] (51) at (5, 4.75) {$\vdots$};
		\node [style=none] (52) at (10.25, 3.25) {$\vdots$};
		\node [style=none] (53) at (5, 3.25) {$\vdots$};
		\node [style=none] (54) at (4.5, 3.75) {};
		\node [style=none] (55) at (4.5, 2.25) {};
		\node [style=none] (56) at (4.5, 0.75) {};
		\node [style=none] (57) at (4.5, 5.25) {};
		\node [style=none] (58) at (10.75, 3.75) {};
		\node [style=none] (59) at (10.75, 2.25) {};
		\node [style=none] (60) at (10.75, 0.75) {};
		\node [style=none] (61) at (10.75, 5.25) {};
		\node [style={Z dot (zh)}] (62) at (15.5, 3) {};
		\node [style=labelled hbox] (63) at (15.5, 5.5) {$0$};
		\node [style={X dot (zh)}] (65) at (16.5, 4.5) {$\pi$};
		\node [style={X dot (zh)}] (66) at (16.5, 1.5) {$\pi$};
		\node [style=labelled hbox] (67) at (17.5, 4.5) {$0$};
		\node [style=labelled hbox] (68) at (17.5, 1.5) {$0$};
		\node [style={Z dot (zh)}] (69) at (18.5, 3.75) {};
		\node [style={Z dot (zh)}] (70) at (18.5, 2.25) {};
		\node [style={Z dot (zh)}] (71) at (18.5, 0.75) {};
		\node [style={Z dot (zh)}] (72) at (18.5, 5.25) {};
		\node [style={Z dot (zh)}] (73) at (13.25, 3.75) {};
		\node [style={Z dot (zh)}] (74) at (13.25, 2.25) {};
		\node [style={Z dot (zh)}] (75) at (13.25, 0.75) {};
		\node [style={Z dot (zh)}] (76) at (13.25, 5.25) {};
		\node [style=labelled hbox] (77) at (14.25, 4.5) {$0$};
		\node [style=labelled hbox] (78) at (14.25, 1.5) {$0$};
		\node [style=none] (79) at (18.5, 4.75) {$\vdots$};
		\node [style=none] (80) at (18.5, 1.75) {$\vdots$};
		\node [style=none] (81) at (13.25, 1.75) {$\vdots$};
		\node [style=none] (82) at (13.25, 4.75) {$\vdots$};
		\node [style=none] (83) at (18.5, 3.25) {$\vdots$};
		\node [style=none] (84) at (13.25, 3.25) {$\vdots$};
		\node [style=none] (85) at (12.75, 3.75) {};
		\node [style=none] (86) at (12.75, 2.25) {};
		\node [style=none] (87) at (12.75, 0.75) {};
		\node [style=none] (88) at (12.75, 5.25) {};
		\node [style=none] (89) at (19, 3.75) {};
		\node [style=none] (90) at (19, 2.25) {};
		\node [style=none] (91) at (19, 0.75) {};
		\node [style=none] (92) at (19, 5.25) {};
		\node [style=none] (93) at (3.5, 3) {$=$};
		\node [style=none] (94) at (11.75, 3) {$+$};
	\end{pgfonlayer}
	\begin{pgfonlayer}{edgelayer}
		\draw [bend left] (0) to (3);
		\draw [bend right] (0) to (4);
		\draw [bend left=15] (5) to (10);
		\draw [bend right=15] (5) to (7);
		\draw (3) to (5);
		\draw (4) to (6);
		\draw [bend left=15] (6) to (8);
		\draw [bend right=15] (6) to (9);
		\draw [bend right] (0) to (15);
		\draw [bend left] (0) to (16);
		\draw [bend right=15] (16) to (12);
		\draw [bend left=15] (16) to (13);
		\draw [bend right=15] (11) to (15);
		\draw [bend right=15] (15) to (14);
		\draw (29.center) to (9);
		\draw (8) to (28.center);
		\draw (27.center) to (7);
		\draw (10) to (30.center);
		\draw (13) to (25.center);
		\draw (12) to (24.center);
		\draw (11) to (23.center);
		\draw (14) to (26.center);
		\draw (32) to (33);
		\draw (33) to (31);
		\draw [bend left] (31) to (34);
		\draw [bend right] (31) to (35);
		\draw [bend left=15] (36) to (41);
		\draw [bend right=15] (36) to (38);
		\draw (34) to (36);
		\draw (35) to (37);
		\draw [bend left=15] (37) to (39);
		\draw [bend right=15] (37) to (40);
		\draw [bend right] (31) to (46);
		\draw [bend left] (31) to (47);
		\draw [bend right=15] (47) to (43);
		\draw [bend left=15] (47) to (44);
		\draw [bend right=15] (42) to (46);
		\draw [bend right=15] (46) to (45);
		\draw (60.center) to (40);
		\draw (39) to (59.center);
		\draw (58.center) to (38);
		\draw (41) to (61.center);
		\draw (44) to (56.center);
		\draw (43) to (55.center);
		\draw (42) to (54.center);
		\draw (45) to (57.center);
		\draw [bend left] (62) to (65);
		\draw [bend right] (62) to (66);
		\draw [bend left=15] (67) to (72);
		\draw [bend right=15] (67) to (69);
		\draw (65) to (67);
		\draw (66) to (68);
		\draw [bend left=15] (68) to (70);
		\draw [bend right=15] (68) to (71);
		\draw [bend right] (62) to (77);
		\draw [bend left] (62) to (78);
		\draw [bend right=15] (78) to (74);
		\draw [bend left=15] (78) to (75);
		\draw [bend right=15] (73) to (77);
		\draw [bend right=15] (77) to (76);
		\draw (91.center) to (71);
		\draw (70) to (90.center);
		\draw (89.center) to (69);
		\draw (72) to (92.center);
		\draw (75) to (87.center);
		\draw (74) to (86.center);
		\draw (73) to (85.center);
		\draw (76) to (88.center);
		\draw (63) to (62);
	\end{pgfonlayer}
\end{tikzpicture}
    \end{equation*}
    On the left-hand side, we have a variable connected to arbitrary clauses, whereas on the right-hand side we have two terms, each with a clause of one literal introduced onto the variable.
\end{lemmarep}
\begin{proof}
    This follows from writing the Z-spider as a sum
    \begin{equation}
        \label{dpll_ksat_statepush}
        \begin{tikzpicture}
	\begin{pgfonlayer}{nodelayer}
		\node [style={Z dot (zh)}] (0) at (0, 0.75) {};
		\node [style=none] (1) at (1, 0.75) {};
		\node [style={X dot (zh)}] (2) at (3, 0.75) {};
		\node [style=none] (3) at (4, 0.75) {};
		\node [style={X dot (zh)}] (4) at (6, 0.75) {$\pi$};
		\node [style=none] (5) at (7, 0.75) {};
		\node [style={X dot (zh)}] (6) at (9, 0.75) {};
		\node [style=none] (7) at (10, 0.75) {};
		\node [style={X dot (zh)}] (8) at (13, 0.75) {$\pi$};
		\node [style=none] (9) at (14, 0.75) {};
		\node [style={X dot (zh)}] (10) at (12, 0.75) {};
		\node [style=none] (11) at (11, 0.75) {$+$};
		\node [style=none] (12) at (5, 0.75) {$+$};
		\node [style=none] (13) at (8, 1) {$\stackrel{SF}{=}$};
		\node [style=none] (14) at (2, 1) {$\stackrel{\eqref{zspidersum}}{=}$};
		\node [style=labelled hbox] (15) at (3.5, -0.75) {$0$};
		\node [style=none] (16) at (4.5, -0.75) {};
		\node [style={X dot (zh)}] (17) at (8, -0.75) {$\pi$};
		\node [style=none] (18) at (9, -0.75) {};
		\node [style=labelled hbox] (19) at (7, -0.75) {$0$};
		\node [style=none] (20) at (5.5, -0.75) {$+$};
		\node [style=none] (21) at (2, -0.5) {$\stackrel{\eqref{satembed_zerostate}}{=}$};
	\end{pgfonlayer}
	\begin{pgfonlayer}{edgelayer}
		\draw (0) to (1.center);
		\draw (2) to (3.center);
		\draw (4) to (5.center);
		\draw (6) to (7.center);
		\draw (8) to (9.center);
		\draw (10) to (8);
		\draw (15) to (16.center);
		\draw (17) to (18.center);
		\draw (19) to (17);
	\end{pgfonlayer}
\end{tikzpicture}
    \end{equation}
    and applying the $SF_Z$ rule to the central spider:
    \begin{equation}
        \begin{tikzpicture}
	\begin{pgfonlayer}{nodelayer}
		\node [style={Z dot (zh)}] (0) at (-4.5, -0.25) {};
		\node [style={X dot (zh)}] (1) at (-3.5, 1.25) {$\pi$};
		\node [style={X dot (zh)}] (2) at (-3.5, -1.75) {$\pi$};
		\node [style=labelled hbox] (3) at (-2.5, 1.25) {$0$};
		\node [style=labelled hbox] (4) at (-2.5, -1.75) {$0$};
		\node [style={Z dot (zh)}] (5) at (-1.5, 0.5) {};
		\node [style={Z dot (zh)}] (6) at (-1.5, -1) {};
		\node [style={Z dot (zh)}] (7) at (-1.5, -2.5) {};
		\node [style={Z dot (zh)}] (8) at (-1.5, 2) {};
		\node [style={Z dot (zh)}] (9) at (-6.5, 0.5) {};
		\node [style={Z dot (zh)}] (10) at (-6.5, -1) {};
		\node [style={Z dot (zh)}] (11) at (-6.5, -2.5) {};
		\node [style={Z dot (zh)}] (12) at (-6.5, 2) {};
		\node [style=labelled hbox] (13) at (-5.5, 1.25) {$0$};
		\node [style=labelled hbox] (14) at (-5.5, -1.75) {$0$};
		\node [style=none] (15) at (-1.5, 1.5) {$\vdots$};
		\node [style=none] (16) at (-1.5, -1.5) {$\vdots$};
		\node [style=none] (17) at (-6.5, -1.5) {$\vdots$};
		\node [style=none] (18) at (-6.5, 1.5) {$\vdots$};
		\node [style=none] (19) at (-1.5, 0) {$\vdots$};
		\node [style=none] (20) at (-6.5, 0) {$\vdots$};
		\node [style=none] (21) at (-7, 0.5) {};
		\node [style=none] (22) at (-7, -1) {};
		\node [style=none] (23) at (-7, -2.5) {};
		\node [style=none] (24) at (-7, 2) {};
		\node [style=none] (25) at (-1, 0.5) {};
		\node [style=none] (26) at (-1, -1) {};
		\node [style=none] (27) at (-1, -2.5) {};
		\node [style=none] (28) at (-1, 2) {};
		\node [style={Z dot (zh)}] (29) at (11.25, -0.25) {};
		\node [style=labelled hbox] (30) at (11.25, 2.25) {$0$};
		\node [style={X dot (zh)}] (31) at (11.25, 1.25) {$\pi$};
		\node [style={X dot (zh)}] (32) at (12.25, 1.25) {$\pi$};
		\node [style={X dot (zh)}] (33) at (12.25, -1.75) {$\pi$};
		\node [style=labelled hbox] (34) at (13.25, 1.25) {$0$};
		\node [style=labelled hbox] (35) at (13.25, -1.75) {$0$};
		\node [style={Z dot (zh)}] (36) at (14.25, 0.5) {};
		\node [style={Z dot (zh)}] (37) at (14.25, -1) {};
		\node [style={Z dot (zh)}] (38) at (14.25, -2.5) {};
		\node [style={Z dot (zh)}] (39) at (14.25, 2) {};
		\node [style={Z dot (zh)}] (40) at (9.25, 0.5) {};
		\node [style={Z dot (zh)}] (41) at (9.25, -1) {};
		\node [style={Z dot (zh)}] (42) at (9.25, -2.5) {};
		\node [style={Z dot (zh)}] (43) at (9.25, 2) {};
		\node [style=labelled hbox] (44) at (10.25, 1.25) {$0$};
		\node [style=labelled hbox] (45) at (10.25, -1.75) {$0$};
		\node [style=none] (46) at (14.25, 1.5) {$\vdots$};
		\node [style=none] (47) at (14.25, -1.5) {$\vdots$};
		\node [style=none] (48) at (9.25, -1.5) {$\vdots$};
		\node [style=none] (49) at (9.25, 1.5) {$\vdots$};
		\node [style=none] (50) at (14.25, 0) {$\vdots$};
		\node [style=none] (51) at (9.25, 0) {$\vdots$};
		\node [style=none] (52) at (8.75, 0.5) {};
		\node [style=none] (53) at (8.75, -1) {};
		\node [style=none] (54) at (8.75, -2.5) {};
		\node [style=none] (55) at (8.75, 2) {};
		\node [style=none] (56) at (14.75, 0.5) {};
		\node [style=none] (57) at (14.75, -1) {};
		\node [style=none] (58) at (14.75, -2.5) {};
		\node [style=none] (59) at (14.75, 2) {};
		\node [style={Z dot (zh)}] (60) at (18.75, -0.25) {};
		\node [style=labelled hbox] (61) at (18.75, 2.25) {$0$};
		\node [style={X dot (zh)}] (62) at (19.75, 1.25) {$\pi$};
		\node [style={X dot (zh)}] (63) at (19.75, -1.75) {$\pi$};
		\node [style=labelled hbox] (64) at (20.75, 1.25) {$0$};
		\node [style=labelled hbox] (65) at (20.75, -1.75) {$0$};
		\node [style={Z dot (zh)}] (66) at (21.75, 0.5) {};
		\node [style={Z dot (zh)}] (67) at (21.75, -1) {};
		\node [style={Z dot (zh)}] (68) at (21.75, -2.5) {};
		\node [style={Z dot (zh)}] (69) at (21.75, 2) {};
		\node [style={Z dot (zh)}] (70) at (16.75, 0.5) {};
		\node [style={Z dot (zh)}] (71) at (16.75, -1) {};
		\node [style={Z dot (zh)}] (72) at (16.75, -2.5) {};
		\node [style={Z dot (zh)}] (73) at (16.75, 2) {};
		\node [style=labelled hbox] (74) at (17.75, 1.25) {$0$};
		\node [style=labelled hbox] (75) at (17.75, -1.75) {$0$};
		\node [style=none] (76) at (21.75, 1.5) {$\vdots$};
		\node [style=none] (77) at (21.75, -1.5) {$\vdots$};
		\node [style=none] (78) at (16.75, -1.5) {$\vdots$};
		\node [style=none] (79) at (16.75, 1.5) {$\vdots$};
		\node [style=none] (80) at (21.75, 0) {$\vdots$};
		\node [style=none] (81) at (16.75, 0) {$\vdots$};
		\node [style=none] (82) at (16.25, 0.5) {};
		\node [style=none] (83) at (16.25, -1) {};
		\node [style=none] (84) at (16.25, -2.5) {};
		\node [style=none] (85) at (16.25, 2) {};
		\node [style=none] (86) at (22.25, 0.5) {};
		\node [style=none] (87) at (22.25, -1) {};
		\node [style=none] (88) at (22.25, -2.5) {};
		\node [style=none] (89) at (22.25, 2) {};
		\node [style=none] (90) at (8, -0.25) {$=$};
		\node [style=none] (91) at (15.5, -0.25) {$+$};
		\node [style={Z dot (zh)}] (92) at (3.75, -0.25) {};
		\node [style={X dot (zh)}] (93) at (4.75, 1.25) {$\pi$};
		\node [style={X dot (zh)}] (94) at (4.75, -1.75) {$\pi$};
		\node [style=labelled hbox] (95) at (5.75, 1.25) {$0$};
		\node [style=labelled hbox] (96) at (5.75, -1.75) {$0$};
		\node [style={Z dot (zh)}] (97) at (6.75, 0.5) {};
		\node [style={Z dot (zh)}] (98) at (6.75, -1) {};
		\node [style={Z dot (zh)}] (99) at (6.75, -2.5) {};
		\node [style={Z dot (zh)}] (100) at (6.75, 2) {};
		\node [style={Z dot (zh)}] (101) at (1.5, 0.5) {};
		\node [style={Z dot (zh)}] (102) at (1.5, -1) {};
		\node [style={Z dot (zh)}] (103) at (1.5, -2.5) {};
		\node [style={Z dot (zh)}] (104) at (1.5, 2) {};
		\node [style=labelled hbox] (105) at (2.5, 1.25) {$0$};
		\node [style=labelled hbox] (106) at (2.5, -1.75) {$0$};
		\node [style=none] (107) at (6.75, 1.5) {$\vdots$};
		\node [style=none] (108) at (6.75, -1.5) {$\vdots$};
		\node [style=none] (109) at (1.5, -1.5) {$\vdots$};
		\node [style=none] (110) at (1.5, 1.5) {$\vdots$};
		\node [style=none] (111) at (6.75, 0) {$\vdots$};
		\node [style=none] (112) at (1.5, 0) {$\vdots$};
		\node [style=none] (113) at (1, 0.5) {};
		\node [style=none] (114) at (1, -1) {};
		\node [style=none] (115) at (1, -2.5) {};
		\node [style=none] (116) at (1, 2) {};
		\node [style=none] (117) at (7.25, 0.5) {};
		\node [style=none] (118) at (7.25, -1) {};
		\node [style=none] (119) at (7.25, -2.5) {};
		\node [style=none] (120) at (7.25, 2) {};
		\node [style={Z dot (zh)}] (121) at (3.75, 1.25) {};
		\node [style=none] (122) at (0, 0) {$\stackrel{SF_Z}{=}$};
	\end{pgfonlayer}
	\begin{pgfonlayer}{edgelayer}
		\draw [bend left] (0) to (1);
		\draw [bend right] (0) to (2);
		\draw [bend left=15] (3) to (8);
		\draw [bend right=15] (3) to (5);
		\draw (1) to (3);
		\draw (2) to (4);
		\draw [bend left=15] (4) to (6);
		\draw [bend right=15] (4) to (7);
		\draw [bend right] (0) to (13);
		\draw [bend left] (0) to (14);
		\draw [bend right=15] (14) to (10);
		\draw [bend left=15] (14) to (11);
		\draw [bend right=15] (9) to (13);
		\draw [bend right=15] (13) to (12);
		\draw (27.center) to (7);
		\draw (6) to (26.center);
		\draw (25.center) to (5);
		\draw (8) to (28.center);
		\draw (11) to (23.center);
		\draw (10) to (22.center);
		\draw (9) to (21.center);
		\draw (12) to (24.center);
		\draw (30) to (31);
		\draw (31) to (29);
		\draw [bend left] (29) to (32);
		\draw [bend right] (29) to (33);
		\draw [bend left=15] (34) to (39);
		\draw [bend right=15] (34) to (36);
		\draw (32) to (34);
		\draw (33) to (35);
		\draw [bend left=15] (35) to (37);
		\draw [bend right=15] (35) to (38);
		\draw [bend right] (29) to (44);
		\draw [bend left] (29) to (45);
		\draw [bend right=15] (45) to (41);
		\draw [bend left=15] (45) to (42);
		\draw [bend right=15] (40) to (44);
		\draw [bend right=15] (44) to (43);
		\draw (58.center) to (38);
		\draw (37) to (57.center);
		\draw (56.center) to (36);
		\draw (39) to (59.center);
		\draw (42) to (54.center);
		\draw (41) to (53.center);
		\draw (40) to (52.center);
		\draw (43) to (55.center);
		\draw [bend left] (60) to (62);
		\draw [bend right] (60) to (63);
		\draw [bend left=15] (64) to (69);
		\draw [bend right=15] (64) to (66);
		\draw (62) to (64);
		\draw (63) to (65);
		\draw [bend left=15] (65) to (67);
		\draw [bend right=15] (65) to (68);
		\draw [bend right] (60) to (74);
		\draw [bend left] (60) to (75);
		\draw [bend right=15] (75) to (71);
		\draw [bend left=15] (75) to (72);
		\draw [bend right=15] (70) to (74);
		\draw [bend right=15] (74) to (73);
		\draw (88.center) to (68);
		\draw (67) to (87.center);
		\draw (86.center) to (66);
		\draw (69) to (89.center);
		\draw (72) to (84.center);
		\draw (71) to (83.center);
		\draw (70) to (82.center);
		\draw (73) to (85.center);
		\draw (61) to (60);
		\draw [bend left] (92) to (93);
		\draw [bend right] (92) to (94);
		\draw [bend left=15] (95) to (100);
		\draw [bend right=15] (95) to (97);
		\draw (93) to (95);
		\draw (94) to (96);
		\draw [bend left=15] (96) to (98);
		\draw [bend right=15] (96) to (99);
		\draw [bend right] (92) to (105);
		\draw [bend left] (92) to (106);
		\draw [bend right=15] (106) to (102);
		\draw [bend left=15] (106) to (103);
		\draw [bend right=15] (101) to (105);
		\draw [bend right=15] (105) to (104);
		\draw (119.center) to (99);
		\draw (98) to (118.center);
		\draw (117.center) to (97);
		\draw (100) to (120.center);
		\draw (103) to (115.center);
		\draw (102) to (114.center);
		\draw (101) to (113.center);
		\draw (104) to (116.center);
		\draw (121) to (92);
	\end{pgfonlayer}
\end{tikzpicture}
    \end{equation}
\end{proof}

Modifications of the CDP algorithm are used extensively in practice, with some of the best solvers like sharpSAT \cite{thurley_sharpsat_2006} and Cachet \cite{sang_combining_2004} making use of this technique. All the best-known theoretical upper bounds on runtime are based on careful analysis of this algorithm. Note that we left the choice of variable to branch on unspecified - choosing this wisely is crucial to obtaining good runtimes, as we will see in the next section.

\subsection{Reduction from \sSAT to $\stwoSAT_\pm$}
\label{subsec:reduction}

Valiant \cite{valiant_complexity_1979} showed that $\stwoSAT$ is \sP-complete, which implies that there is a polynomial-time Turing reduction from $\sSAT$ to $\stwoSAT$. Since good upper bounds are known for $\stwoSAT$, reducing $\sSAT$ to $\stwoSAT$ is one strategy to obtain better than brute-force bounds that are independent of $k$. However, the polynomial-time reduction guaranteed by \sP-completeness maps instances with $n$ variables and $m$ clauses to instances with $O(nm)$ variables, which destroys any advantage we could have gained from faster $\stwoSAT$ algorithms. Instead we present the following linear-time reduction to a weighted variant of \sSAT:

\begin{lemmarep}
    \label{dpll_neg_kclause}
    The following diagrammatic equivalence holds:
    \begin{equation}
        \begin{tikzpicture}
	\begin{pgfonlayer}{nodelayer}
		\node [style=labelled hbox] (0) at (-1, 0) {$0$};
		\node [style=none] (1) at (-2, 0.75) {};
		\node [style=none] (2) at (-2, -0.75) {};
		\node [style=none] (3) at (0.5, 0) {$=$};
		\node [style=none] (4) at (2, 0.75) {};
		\node [style=none] (5) at (2, -0.75) {};
		\node [style=labelled hbox] (6) at (4, 0.75) {$0$};
		\node [style=labelled hbox] (7) at (4, -0.75) {$0$};
		\node [style={X dot (zh)}] (8) at (3, 0.75) {$\pi$};
		\node [style={X dot (zh)}] (9) at (3, -0.75) {$\pi$};
		\node [style={Z dot (zh)}] (10) at (5, 0) {$\pi$};
		\node [style=none] (11) at (-2, 0.25) {$\vdots$};
		\node [style=none] (12) at (2, 0.25) {$\vdots$};
	\end{pgfonlayer}
	\begin{pgfonlayer}{edgelayer}
		\draw [bend left] (1.center) to (0);
		\draw [bend right] (2.center) to (0);
		\draw (4.center) to (8);
		\draw (8) to (6);
		\draw [bend left] (6) to (10);
		\draw [bend right] (7) to (10);
		\draw (7) to (9);
		\draw (9) to (5.center);
	\end{pgfonlayer}
\end{tikzpicture}
    \end{equation}
    This directly generalizes the $\mathrm{BW}$ axiom of the $\Delta\mathrm{ZX}$-calculus \cite{vilmart_zx-calculus_2019}.
\end{lemmarep}
\begin{proof}
    We can see the following
        \ctikzfig{dpll-neg-kclause-proof-1}
        \ctikzfig{dpll-neg-kclause-proof-2}
    which completes the proof.
\end{proof}

This lemma allows us to remove clauses with degree greater than two, at the cost of introducing extra Z-spiders. If we applied this to every clause in a $\sSAT$ diagram with $n$ variables and $m$ clauses, we would have a diagram with $n + m$ spiders and furthermore, with the exception of $\pi$-phases on $m$ of these spiders, this diagram would represent a $\stwoSAT$ instance. Therefore, we will relax our definition of $\sSAT$ to permit $\pi$-phases appearing on the Z-spiders (i.e.~on the variables). Note that we have the following
\begin{equation}
    \begin{tikzpicture}
	\begin{pgfonlayer}{nodelayer}
		\node [style={Z dot (zh)}] (0) at (-2, 0) {$\pi$};
		\node [style=none] (1) at (-2.75, 1.25) {};
		\node [style=none] (2) at (-1.25, 1.25) {};
		\node [style=none] (3) at (-2.75, -1.25) {};
		\node [style=none] (4) at (-1.25, -1.25) {};
		\node [style=none] (5) at (-2, 1.25) {$\dots$};
		\node [style=none] (6) at (-2, -1.25) {$\dots$};
		\node [style={X phase dot (zh)}] (7) at (-3, 0) {$k\pi$};
		\node [style={Z dot (zh)}] (8) at (3.25, 0) {};
		\node [style=none] (9) at (2.5, 1.25) {};
		\node [style=none] (10) at (4, 1.25) {};
		\node [style=none] (11) at (2.5, -1.25) {};
		\node [style=none] (12) at (4, -1.25) {};
		\node [style=none] (13) at (3.25, 1.25) {$\dots$};
		\node [style=none] (14) at (3.25, -1.25) {$\dots$};
		\node [style={X phase dot (zh)}] (15) at (1.25, 0) {$k\pi$};
		\node [style={Z dot (zh)}] (16) at (1.75, -0.75) {$\pi$};
		\node [style={Z dot (zh)}] (17) at (7.25, 0) {};
		\node [style=none] (18) at (6.5, 1.25) {};
		\node [style=none] (19) at (8, 1.25) {};
		\node [style=none] (20) at (6.5, -1.25) {};
		\node [style=none] (21) at (8, -1.25) {};
		\node [style=none] (22) at (7.25, 1.25) {$\dots$};
		\node [style=none] (23) at (7.25, -1.25) {$\dots$};
		\node [style={X phase dot (zh)}] (24) at (6.25, 0) {$k\pi$};
		\node [style={Z dot (zh)}] (25) at (8, 0) {$\pi$};
		\node [style={X phase dot (zh)}] (26) at (9, 0) {$k\pi$};
		\node [style=none] (27) at (10.5, 0) {$=$};
		\node [style=none] (28) at (5, 0.25) {$\stackrel{\eqref{satembed_copy}}{=}$};
		\node [style=none] (29) at (-0.25, 0.25) {$\stackrel{SF_Z}{=}$};
		\node [style=none] (30) at (15, 1) {};
		\node [style=none] (31) at (15, -1) {};
		\node [style={Z dot (zh)}] (32) at (13, 0) {};
		\node [style=none] (33) at (12.25, 1.25) {};
		\node [style=none] (34) at (13.75, 1.25) {};
		\node [style=none] (35) at (12.25, -1.25) {};
		\node [style=none] (36) at (13.75, -1.25) {};
		\node [style=none] (37) at (13, 1.25) {$\dots$};
		\node [style=none] (38) at (13, -1.25) {$\dots$};
		\node [style={X phase dot (zh)}] (39) at (12, 0) {$k\pi$};
		\node [style=none] (40) at (14.25, 0) {$\cdot$};
		\node [style=none] (41) at (16.25, 0.75) {$1$};
		\node [style=none] (42) at (16, -0.5) {$-1$};
		\node [style=none] (43) at (18.25, 0.75) {if $k = 0$};
		\node [style=none] (44) at (18.25, -0.5) {if $k = 1$};
		\node [style={Z dot (zh)}] (45) at (2.25, 0) {};
	\end{pgfonlayer}
	\begin{pgfonlayer}{edgelayer}
		\draw [bend right=15] (4.center) to (0);
		\draw [bend right=15] (0) to (3.center);
		\draw [bend right=15] (0) to (2.center);
		\draw [bend left=15] (0) to (1.center);
		\draw (7) to (0);
		\draw [bend right=15] (12.center) to (8);
		\draw [bend right=15] (8) to (11.center);
		\draw [bend right=15] (8) to (10.center);
		\draw [bend left=15] (8) to (9.center);
		\draw (15) to (8);
		\draw [bend right=15] (21.center) to (17);
		\draw [bend right=15] (17) to (20.center);
		\draw [bend right=15] (17) to (19.center);
		\draw [bend left=15] (17) to (18.center);
		\draw (25) to (26);
		\draw (24) to (17);
		\draw [style=brace edge] (31.center) to (30.center);
		\draw [bend right=15] (36.center) to (32);
		\draw [bend right=15] (32) to (35.center);
		\draw [bend right=15] (32) to (34.center);
		\draw [bend left=15] (32) to (33.center);
		\draw (39) to (32);
		\draw (16) to (45);
	\end{pgfonlayer}
\end{tikzpicture}
\end{equation}
which implies that these diagrams are the same as $\sSAT$ instances, but the sign is flipped whenever a variable corresponding to a Z-spider with a $\pi$-phase is assigned to be true. The overall sign of the diagram for a particular assignment of variables is given by the parity of the assignment of such variables. We can extend the $\sSAT$ problem as follows to handle this natively.

\begin{definition}
    The problem $\sSAT_\pm$ is defined as follows. Given a CNF formula $f(x_1, \dots, x_n)$ with $n$ variables and $m$ clauses, and a set $N \subseteq \{1, \dots, n\}$, compute the quantity
    \begin{equation}
        \sSAT_\pm(f, N) = \sum_{\substack{\vec{x} \in \mathbb{B}^n \\ \text{even }N\text{-parity}}}f(\vec{x}) ~~- \sum_{\substack{\vec{x} \in \mathbb{B}^n \\ \text{odd }N\text{-parity}}}f(\vec{x})
    \end{equation}
    where a vector $\vec{x} \in \mathbb{B}^n$ has even or odd $N$-parity if $\bigoplus_{i \in N} x_i = 0$ or $1$, respectively.
\end{definition}

\begin{theoremrep}
    $\sSAT_\pm$ is GapP-complete
\end{theoremrep}
\begin{proof}
    Suppose there is a problem $M$ in \textbf{GapP}, then the aim is to compute $\mathrm{gap}_M = \#M - \#\overline{M}$ where $\#M$ ($\#\overline{M}$) is the number of accepting (rejecting) paths of a non-deterministic Turing machine $M$. Clearly $\#M$ and $\#\overline{M}$ are both in \#P. Because the Cook-Levin reduction from NP to \textbf{SAT} is parsimonious, $\sSAT$ is \#P-complete, and there exist CNF formulae $f_M$ and $f_{\overline{M}}$ such that $\sSAT(f_M) = \#M$ and $\sSAT(f_{\overline{M}}) = \#\overline{M}$. Define the following formula:
    \begin{align*}
        f_{M - \overline{M}} &= \left(\bigwedge_{i = 1}^{m_M} C_M^i \lor z\right) \land \left(\bigwedge_{j = 1}^{n_M} \lnot x_M^j \lor \lnot z\right) \land \left(\bigwedge_{i = 1}^{m_{\overline{M}}} C_{\overline{M}}^i \lor \lnot z\right) \land \left(\bigwedge_{j = 1}^{n_{\overline{M}}} \lnot x_{\overline{M}}^j \lor z\right) \\ 
    \end{align*}
    where $C_M^i$ ($C_{\overline{M}}^i$) and $x_M^j$ ($x_{\overline{M}}^j$) are the clauses and variables of $f_M$ ($f_{\overline{M}}$) respectively, and $z$ is a fresh variable. When $z$ is false, $f_{M - \overline{M}}$ reduces to
    \begin{equation*}
        f_{M - \overline{M}} \land \lnot z = \left(\bigwedge_{i = 1}^{m_M} C_M^i\right) \land \left(\bigwedge_{j = 1}^{n_{\overline{M}}} \lnot x_{\overline{M}}^j\right) = f_M \land \left(\bigwedge_{j = 1}^{n_{\overline{M}}} \lnot x_{\overline{M}}^j\right)
    \end{equation*}
    which has exactly $\sSAT(f_M) = \#M$ satisfying solutions, and likewise with $z$ true, we have
    \begin{equation*}
        f_{M - \overline{M}} \land z = f_{\overline{M}} \land \left(\bigwedge_{j = 1}^{n_{M}} \lnot x_{M}^j\right)
    \end{equation*}
    which has exactly $\sSAT(f_{\overline{M}}) = \#\overline{M}$ satisfying solutions. Finally let $N = \{z\}$, then an assignment of $f_{M - \overline{M}}$ has odd or even N-parity exactly when $z$ is true or false respectively. Therefore,
    \begin{align*}
        \sSAT_\pm(f_{M - \overline{M}}, N) &= 
        \sum_{\Substack{\vec{x} \in \mathbb{B}^n \\ \text{even }N\text{-parity}}}f_{M - \overline{M}}(\vec{x}) - \sum_{\Substack{\vec{x} \in \mathbb{B}^n \\ \text{odd }N\text{-parity}}}f_{M - \overline{M}}(\vec{x}) = \sSAT(f_{M - \overline{M}} \land \lnot z) - \sSAT(f_{M - \overline{M}} \land z) \\
        &= \sSAT(f_M) - \sSAT(f_{\overline{M}}) = \#M - \#\overline{M} = \mathrm{gap}_M
    \end{align*}
    so there is a polynomial-time counting reduction from \textbf{GapP} to $\sSAT_\pm$, and it must be \textbf{GapP}-hard. Furthermore, let $f$, $N$ be an instance of $\sSAT_\pm$ with variables $x_i$, then clearly we have that
    \begin{align*}
        \sSAT_\pm(f, N) &= 
        \sum_{\Substack{\vec{x} \in \mathbb{B}^n \\ \text{even }N\text{-parity}}}f(\vec{x}) - \sum_{\Substack{\vec{x} \in \mathbb{B}^n \\ \text{odd }N\text{-parity}}}f(\vec{x}) = \sSAT\left(f \land \lnot \left(\bigoplus_{i \in N} x_i\right)\right) - \sSAT\left(f \land \left(\bigoplus_{i \in N} x_i\right)\right)
    \end{align*}
    since $g = 1 \iff g$ and $g = 0 \iff \lnot g$, so $\lnot (\bigoplus_{i \in N} x_i)$ is the same as asserting even N-parity (and likewise odd N-parity). But \textbf{GapP} is the closure of \#P under subtraction, so $\sSAT_\pm(f, N)$ is in \textbf{GapP}. Thus it follows that $\sSAT_\pm$ is \textbf{GapP}-complete.
\end{proof}

Since $\sSAT_\pm$ is in GapP, which is strictly harder than \#P (for instance, GapP is the closure of \#P under subtraction \cite{fenner_gap-definable_1994}), it may seem that an upper bound for this problem is guaranteed to be worse. However, this is not the case, as the following diagrammatic arguments show that the DPLL algorithm can be easily adapted to handle the sign change as Algorithm \ref{dpllpm}.

\begin{lemmarep}
    \label{dpll-neg-branch}
    The following diagrammatic equivalent to the variable branching rule holds for variables with $\pi$-phases:
    \begin{equation}
        \begin{tikzpicture}
	\begin{pgfonlayer}{nodelayer}
		\node [style={Z dot (zh)}] (0) at (-0.75, 0) {$\pi$};
		\node [style={X dot (zh)}] (3) at (0.25, 1.5) {$\pi$};
		\node [style={X dot (zh)}] (4) at (0.25, -1.5) {$\pi$};
		\node [style=labelled hbox] (5) at (1.25, 1.5) {$0$};
		\node [style=labelled hbox] (6) at (1.25, -1.5) {$0$};
		\node [style={Z dot (zh)}] (7) at (2.25, 0.75) {};
		\node [style={Z dot (zh)}] (8) at (2.25, -0.75) {};
		\node [style={Z dot (zh)}] (9) at (2.25, -2.25) {};
		\node [style={Z dot (zh)}] (10) at (2.25, 2.25) {};
		\node [style={Z dot (zh)}] (11) at (-3, 0.75) {};
		\node [style={Z dot (zh)}] (12) at (-3, -0.75) {};
		\node [style={Z dot (zh)}] (13) at (-3, -2.25) {};
		\node [style={Z dot (zh)}] (14) at (-3, 2.25) {};
		\node [style=labelled hbox] (15) at (-2, 1.5) {$0$};
		\node [style=labelled hbox] (16) at (-2, -1.5) {$0$};
		\node [style=none] (17) at (2.25, 1.75) {$\vdots$};
		\node [style=none] (18) at (2.25, -1.25) {$\vdots$};
		\node [style=none] (19) at (-3, -1.25) {$\vdots$};
		\node [style=none] (20) at (-3, 1.75) {$\vdots$};
		\node [style=none] (21) at (2.25, 0.25) {$\vdots$};
		\node [style=none] (22) at (-3, 0.25) {$\vdots$};
		\node [style=none] (23) at (-3.5, 0.75) {};
		\node [style=none] (24) at (-3.5, -0.75) {};
		\node [style=none] (25) at (-3.5, -2.25) {};
		\node [style=none] (26) at (-3.5, 2.25) {};
		\node [style=none] (27) at (2.75, 0.75) {};
		\node [style=none] (28) at (2.75, -0.75) {};
		\node [style=none] (29) at (2.75, -2.25) {};
		\node [style=none] (30) at (2.75, 2.25) {};
		\node [style={Z dot (zh)}] (31) at (7.5, 0) {};
		\node [style=labelled hbox] (32) at (7.5, 2.5) {$0$};
		\node [style={X dot (zh)}] (33) at (7.5, 1.5) {$\pi$};
		\node [style={X dot (zh)}] (34) at (8.5, 1.5) {$\pi$};
		\node [style={X dot (zh)}] (35) at (8.5, -1.5) {$\pi$};
		\node [style=labelled hbox] (36) at (9.5, 1.5) {$0$};
		\node [style=labelled hbox] (37) at (9.5, -1.5) {$0$};
		\node [style={Z dot (zh)}] (38) at (10.5, 0.75) {};
		\node [style={Z dot (zh)}] (39) at (10.5, -0.75) {};
		\node [style={Z dot (zh)}] (40) at (10.5, -2.25) {};
		\node [style={Z dot (zh)}] (41) at (10.5, 2.25) {};
		\node [style={Z dot (zh)}] (42) at (5.25, 0.75) {};
		\node [style={Z dot (zh)}] (43) at (5.25, -0.75) {};
		\node [style={Z dot (zh)}] (44) at (5.25, -2.25) {};
		\node [style={Z dot (zh)}] (45) at (5.25, 2.25) {};
		\node [style=labelled hbox] (46) at (6.25, 1.5) {$0$};
		\node [style=labelled hbox] (47) at (6.25, -1.5) {$0$};
		\node [style=none] (48) at (10.5, 1.75) {$\vdots$};
		\node [style=none] (49) at (10.5, -1.25) {$\vdots$};
		\node [style=none] (50) at (5.25, -1.25) {$\vdots$};
		\node [style=none] (51) at (5.25, 1.75) {$\vdots$};
		\node [style=none] (52) at (10.5, 0.25) {$\vdots$};
		\node [style=none] (53) at (5.25, 0.25) {$\vdots$};
		\node [style=none] (54) at (4.75, 0.75) {};
		\node [style=none] (55) at (4.75, -0.75) {};
		\node [style=none] (56) at (4.75, -2.25) {};
		\node [style=none] (57) at (4.75, 2.25) {};
		\node [style=none] (58) at (11, 0.75) {};
		\node [style=none] (59) at (11, -0.75) {};
		\node [style=none] (60) at (11, -2.25) {};
		\node [style=none] (61) at (11, 2.25) {};
		\node [style={Z dot (zh)}] (62) at (15.75, 0) {};
		\node [style=labelled hbox] (63) at (15.75, 2.5) {$0$};
		\node [style={X dot (zh)}] (65) at (16.75, 1.5) {$\pi$};
		\node [style={X dot (zh)}] (66) at (16.75, -1.5) {$\pi$};
		\node [style=labelled hbox] (67) at (17.75, 1.5) {$0$};
		\node [style=labelled hbox] (68) at (17.75, -1.5) {$0$};
		\node [style={Z dot (zh)}] (69) at (18.75, 0.75) {};
		\node [style={Z dot (zh)}] (70) at (18.75, -0.75) {};
		\node [style={Z dot (zh)}] (71) at (18.75, -2.25) {};
		\node [style={Z dot (zh)}] (72) at (18.75, 2.25) {};
		\node [style={Z dot (zh)}] (73) at (13.5, 0.75) {};
		\node [style={Z dot (zh)}] (74) at (13.5, -0.75) {};
		\node [style={Z dot (zh)}] (75) at (13.5, -2.25) {};
		\node [style={Z dot (zh)}] (76) at (13.5, 2.25) {};
		\node [style=labelled hbox] (77) at (14.5, 1.5) {$0$};
		\node [style=labelled hbox] (78) at (14.5, -1.5) {$0$};
		\node [style=none] (79) at (18.75, 1.75) {$\vdots$};
		\node [style=none] (80) at (18.75, -1.25) {$\vdots$};
		\node [style=none] (81) at (13.5, -1.25) {$\vdots$};
		\node [style=none] (82) at (13.5, 1.75) {$\vdots$};
		\node [style=none] (83) at (18.75, 0.25) {$\vdots$};
		\node [style=none] (84) at (13.5, 0.25) {$\vdots$};
		\node [style=none] (85) at (13, 0.75) {};
		\node [style=none] (86) at (13, -0.75) {};
		\node [style=none] (87) at (13, -2.25) {};
		\node [style=none] (88) at (13, 2.25) {};
		\node [style=none] (89) at (19.25, 0.75) {};
		\node [style=none] (90) at (19.25, -0.75) {};
		\node [style=none] (91) at (19.25, -2.25) {};
		\node [style=none] (92) at (19.25, 2.25) {};
		\node [style=none] (93) at (3.75, 0) {$=$};
		\node [style=none] (94) at (12, 0) {$-$};
	\end{pgfonlayer}
	\begin{pgfonlayer}{edgelayer}
		\draw [bend left] (0) to (3);
		\draw [bend right] (0) to (4);
		\draw [bend left=15] (5) to (10);
		\draw [bend right=15] (5) to (7);
		\draw (3) to (5);
		\draw (4) to (6);
		\draw [bend left=15] (6) to (8);
		\draw [bend right=15] (6) to (9);
		\draw [bend right] (0) to (15);
		\draw [bend left] (0) to (16);
		\draw [bend right=15] (16) to (12);
		\draw [bend left=15] (16) to (13);
		\draw [bend right=15] (11) to (15);
		\draw [bend right=15] (15) to (14);
		\draw (29.center) to (9);
		\draw (8) to (28.center);
		\draw (27.center) to (7);
		\draw (10) to (30.center);
		\draw (13) to (25.center);
		\draw (12) to (24.center);
		\draw (11) to (23.center);
		\draw (14) to (26.center);
		\draw (32) to (33);
		\draw (33) to (31);
		\draw [bend left] (31) to (34);
		\draw [bend right] (31) to (35);
		\draw [bend left=15] (36) to (41);
		\draw [bend right=15] (36) to (38);
		\draw (34) to (36);
		\draw (35) to (37);
		\draw [bend left=15] (37) to (39);
		\draw [bend right=15] (37) to (40);
		\draw [bend right] (31) to (46);
		\draw [bend left] (31) to (47);
		\draw [bend right=15] (47) to (43);
		\draw [bend left=15] (47) to (44);
		\draw [bend right=15] (42) to (46);
		\draw [bend right=15] (46) to (45);
		\draw (60.center) to (40);
		\draw (39) to (59.center);
		\draw (58.center) to (38);
		\draw (41) to (61.center);
		\draw (44) to (56.center);
		\draw (43) to (55.center);
		\draw (42) to (54.center);
		\draw (45) to (57.center);
		\draw [bend left] (62) to (65);
		\draw [bend right] (62) to (66);
		\draw [bend left=15] (67) to (72);
		\draw [bend right=15] (67) to (69);
		\draw (65) to (67);
		\draw (66) to (68);
		\draw [bend left=15] (68) to (70);
		\draw [bend right=15] (68) to (71);
		\draw [bend right] (62) to (77);
		\draw [bend left] (62) to (78);
		\draw [bend right=15] (78) to (74);
		\draw [bend left=15] (78) to (75);
		\draw [bend right=15] (73) to (77);
		\draw [bend right=15] (77) to (76);
		\draw (91.center) to (71);
		\draw (70) to (90.center);
		\draw (89.center) to (69);
		\draw (72) to (92.center);
		\draw (75) to (87.center);
		\draw (74) to (86.center);
		\draw (73) to (85.center);
		\draw (76) to (88.center);
		\draw (63) to (62);
	\end{pgfonlayer}
\end{tikzpicture}
    \end{equation}
\end{lemmarep}
\begin{proof}
    This follows from the sum
    \begin{equation}
        \begin{tikzpicture}
	\begin{pgfonlayer}{nodelayer}
		\node [style={Z dot (zh)}] (0) at (0, 0.5) {$\pi$};
		\node [style=none] (4) at (0, -0.5) {};
		\node [style=none] (8) at (3, -0.5) {};
		\node [style={X dot (zh)}] (18) at (3, 0.5) {};
		\node [style=none] (21) at (5.5, -0.5) {};
		\node [style={X dot (zh)}] (26) at (5.5, 0.5) {$\pi$};
		\node [style=none] (27) at (4.25, 0) {$-$};
		\node [style=none] (28) at (1.5, 0) {$=$};
	\end{pgfonlayer}
	\begin{pgfonlayer}{edgelayer}
		\draw (4.center) to (0);
		\draw (18) to (8.center);
		\draw (26) to (21.center);
	\end{pgfonlayer}
\end{tikzpicture}
    \end{equation}
    together with the proof of Lemma \ref{dpll_branch_statement}.
\end{proof}

\begin{lemmarep}
    \label{dpll-neg-unitprop}
    The following diagrammatic equivalent to the unit propagation rule holds for variables with $\pi$-phases:
    \begin{equation}
        \begin{tikzpicture}
	\begin{pgfonlayer}{nodelayer}
		\node [style={Z dot (zh)}] (0) at (-0.75, 3) {$\pi$};
		\node [style=labelled hbox] (1) at (-0.75, 5.5) {$0$};
		\node [style={X dot (zh)}] (2) at (-0.75, 4.5) {$\pi$};
		\node [style={X dot (zh)}] (3) at (0.5, 4.5) {$\pi$};
		\node [style={X dot (zh)}] (4) at (0.5, 1.5) {$\pi$};
		\node [style=labelled hbox] (5) at (1.5, 4.5) {$0$};
		\node [style=labelled hbox] (6) at (1.5, 1.5) {$0$};
		\node [style={Z dot (zh)}] (8) at (2.5, 3.75) {};
		\node [style={Z dot (zh)}] (9) at (2.5, 2.25) {};
		\node [style={Z dot (zh)}] (10) at (2.5, 0.75) {};
		\node [style={Z dot (zh)}] (11) at (2.5, 5.25) {};
		\node [style={Z dot (zh)}] (16) at (-3.25, 3.75) {};
		\node [style={Z dot (zh)}] (17) at (-3.25, 2.25) {};
		\node [style={Z dot (zh)}] (18) at (-3.25, 0.75) {};
		\node [style={Z dot (zh)}] (19) at (-3.25, 5.25) {};
		\node [style=labelled hbox] (20) at (-2.25, 4.5) {$0$};
		\node [style=labelled hbox] (21) at (-2.25, 1.5) {$0$};
		\node [style=none] (22) at (2.5, 4.75) {$\vdots$};
		\node [style=none] (23) at (2.5, 1.75) {$\vdots$};
		\node [style=none] (24) at (-3.25, 1.75) {$\vdots$};
		\node [style=none] (25) at (-3.25, 4.75) {$\vdots$};
		\node [style=none] (26) at (2.5, 3.25) {$\vdots$};
		\node [style=none] (27) at (-3.25, 3.25) {$\vdots$};
		\node [style=none] (51) at (4, 3) {$=$};
		\node [style=none] (52) at (-3.75, 3.75) {};
		\node [style=none] (53) at (-3.75, 2.25) {};
		\node [style=none] (54) at (-3.75, 0.75) {};
		\node [style=none] (55) at (-3.75, 5.25) {};
		\node [style=none] (56) at (3, 3.75) {};
		\node [style=none] (57) at (3, 2.25) {};
		\node [style=none] (58) at (3, 0.75) {};
		\node [style=none] (59) at (3, 5.25) {};
		\node [style={Z dot (zh)}] (67) at (9.5, 3.75) {};
		\node [style={Z dot (zh)}] (68) at (9.5, 2.25) {};
		\node [style={Z dot (zh)}] (69) at (9.5, 0.75) {};
		\node [style={Z dot (zh)}] (70) at (9.5, 5.25) {};
		\node [style={Z dot (zh)}] (71) at (6.5, 3.75) {};
		\node [style={Z dot (zh)}] (72) at (6.5, 2.25) {};
		\node [style={Z dot (zh)}] (73) at (6.5, 0.75) {};
		\node [style={Z dot (zh)}] (74) at (6.5, 5.25) {};
		\node [style=labelled hbox] (75) at (7.5, 4.5) {$0$};
		\node [style=labelled hbox] (76) at (7.5, 1.5) {$0$};
		\node [style=none] (77) at (9.5, 4.75) {$\vdots$};
		\node [style=none] (78) at (9.5, 1.75) {$\vdots$};
		\node [style=none] (79) at (6.5, 1.75) {$\vdots$};
		\node [style=none] (80) at (6.5, 4.75) {$\vdots$};
		\node [style=none] (81) at (9.5, 3.25) {$\vdots$};
		\node [style=none] (82) at (6.5, 3.25) {$\vdots$};
		\node [style=none] (83) at (6, 3.75) {};
		\node [style=none] (84) at (6, 2.25) {};
		\node [style=none] (85) at (6, 0.75) {};
		\node [style=none] (86) at (6, 5.25) {};
		\node [style=none] (87) at (10, 3.75) {};
		\node [style=none] (88) at (10, 2.25) {};
		\node [style=none] (89) at (10, 0.75) {};
		\node [style=none] (90) at (10, 5.25) {};
		\node [style={Z dot (zh)}] (91) at (15, 3) {$\pi$};
		\node [style=labelled hbox] (92) at (15, 4.75) {$0$};
		\node [style={X dot (zh)}] (94) at (16.25, 4.5) {$\pi$};
		\node [style={X dot (zh)}] (95) at (16.25, 1.5) {$\pi$};
		\node [style=labelled hbox] (96) at (17.25, 4.5) {$0$};
		\node [style=labelled hbox] (97) at (17.25, 1.5) {$0$};
		\node [style={Z dot (zh)}] (98) at (18.25, 3.75) {};
		\node [style={Z dot (zh)}] (99) at (18.25, 2.25) {};
		\node [style={Z dot (zh)}] (100) at (18.25, 0.75) {};
		\node [style={Z dot (zh)}] (101) at (18.25, 5.25) {};
		\node [style={Z dot (zh)}] (102) at (12.5, 3.75) {};
		\node [style={Z dot (zh)}] (103) at (12.5, 2.25) {};
		\node [style={Z dot (zh)}] (104) at (12.5, 0.75) {};
		\node [style={Z dot (zh)}] (105) at (12.5, 5.25) {};
		\node [style=labelled hbox] (106) at (13.5, 4.5) {$0$};
		\node [style=labelled hbox] (107) at (13.5, 1.5) {$0$};
		\node [style=none] (108) at (18.25, 4.75) {$\vdots$};
		\node [style=none] (109) at (18.25, 1.75) {$\vdots$};
		\node [style=none] (110) at (12.5, 1.75) {$\vdots$};
		\node [style=none] (111) at (12.5, 4.75) {$\vdots$};
		\node [style=none] (112) at (18.25, 3.25) {$\vdots$};
		\node [style=none] (113) at (12.5, 3.25) {$\vdots$};
		\node [style=none] (114) at (19.75, 3) {$=$};
		\node [style=none] (115) at (12, 3.75) {};
		\node [style=none] (116) at (12, 2.25) {};
		\node [style=none] (117) at (12, 0.75) {};
		\node [style=none] (118) at (12, 5.25) {};
		\node [style=none] (119) at (18.75, 3.75) {};
		\node [style=none] (120) at (18.75, 2.25) {};
		\node [style=none] (121) at (18.75, 0.75) {};
		\node [style=none] (122) at (18.75, 5.25) {};
		\node [style={Z dot (zh)}] (123) at (25.25, 3.75) {};
		\node [style={Z dot (zh)}] (124) at (25.25, 2.25) {};
		\node [style={Z dot (zh)}] (125) at (25.25, 0.75) {};
		\node [style={Z dot (zh)}] (126) at (25.25, 5.25) {};
		\node [style={Z dot (zh)}] (127) at (22.25, 3.75) {};
		\node [style={Z dot (zh)}] (128) at (22.25, 2.25) {};
		\node [style={Z dot (zh)}] (129) at (22.25, 0.75) {};
		\node [style={Z dot (zh)}] (130) at (22.25, 5.25) {};
		\node [style=labelled hbox] (131) at (24.25, 4.5) {$0$};
		\node [style=labelled hbox] (132) at (24.25, 1.5) {$0$};
		\node [style=none] (133) at (25.25, 4.75) {$\vdots$};
		\node [style=none] (134) at (25.25, 1.75) {$\vdots$};
		\node [style=none] (135) at (22.25, 1.75) {$\vdots$};
		\node [style=none] (136) at (22.25, 4.75) {$\vdots$};
		\node [style=none] (137) at (25.25, 3.25) {$\vdots$};
		\node [style=none] (138) at (22.25, 3.25) {$\vdots$};
		\node [style=none] (139) at (21.75, 3.75) {};
		\node [style=none] (140) at (21.75, 2.25) {};
		\node [style=none] (141) at (21.75, 0.75) {};
		\node [style=none] (142) at (21.75, 5.25) {};
		\node [style=none] (143) at (25.75, 3.75) {};
		\node [style=none] (144) at (25.75, 2.25) {};
		\node [style=none] (145) at (25.75, 0.75) {};
		\node [style=none] (146) at (25.75, 5.25) {};
		\node [style=none] (147) at (5.25, 3) {$-$};
		\node [style=none] (148) at (21, 3) {$+$};
	\end{pgfonlayer}
	\begin{pgfonlayer}{edgelayer}
		\draw (1) to (2);
		\draw (2) to (0);
		\draw [bend left] (0) to (3);
		\draw [bend right] (0) to (4);
		\draw [bend left=15] (5) to (11);
		\draw [bend right=15] (5) to (8);
		\draw (3) to (5);
		\draw (4) to (6);
		\draw [bend left=15] (6) to (9);
		\draw [bend right=15] (6) to (10);
		\draw [bend right] (0) to (20);
		\draw [bend left] (0) to (21);
		\draw [bend right=15] (21) to (17);
		\draw [bend left=15] (21) to (18);
		\draw [bend right=15] (16) to (20);
		\draw [bend right=15] (20) to (19);
		\draw (58.center) to (10);
		\draw (9) to (57.center);
		\draw (56.center) to (8);
		\draw (11) to (59.center);
		\draw (18) to (54.center);
		\draw (17) to (53.center);
		\draw (16) to (52.center);
		\draw (19) to (55.center);
		\draw [bend right=15] (76) to (72);
		\draw [bend left=15] (76) to (73);
		\draw [bend right=15] (71) to (75);
		\draw [bend right=15] (75) to (74);
		\draw (89.center) to (69);
		\draw (68) to (88.center);
		\draw (87.center) to (67);
		\draw (70) to (90.center);
		\draw (73) to (85.center);
		\draw (72) to (84.center);
		\draw (71) to (83.center);
		\draw (74) to (86.center);
		\draw [bend left] (91) to (94);
		\draw [bend right] (91) to (95);
		\draw [bend left=15] (96) to (101);
		\draw [bend right=15] (96) to (98);
		\draw (94) to (96);
		\draw (95) to (97);
		\draw [bend left=15] (97) to (99);
		\draw [bend right=15] (97) to (100);
		\draw [bend right] (91) to (106);
		\draw [bend left] (91) to (107);
		\draw [bend right=15] (107) to (103);
		\draw [bend left=15] (107) to (104);
		\draw [bend right=15] (102) to (106);
		\draw [bend right=15] (106) to (105);
		\draw (121.center) to (100);
		\draw (99) to (120.center);
		\draw (119.center) to (98);
		\draw (101) to (122.center);
		\draw (104) to (117.center);
		\draw (103) to (116.center);
		\draw (102) to (115.center);
		\draw (105) to (118.center);
		\draw (145.center) to (125);
		\draw (124) to (144.center);
		\draw (143.center) to (123);
		\draw (126) to (146.center);
		\draw (129) to (141.center);
		\draw (128) to (140.center);
		\draw (127) to (139.center);
		\draw (130) to (142.center);
		\draw (92) to (91);
		\draw [bend left=15] (131) to (126);
		\draw [bend right=15] (131) to (123);
		\draw [bend left=15] (132) to (124);
		\draw [bend left=15] (125) to (132);
	\end{pgfonlayer}
\end{tikzpicture}
    \end{equation}
\end{lemmarep}
\begin{proof}
    This follows from the identities
    \begin{equation}
        \begin{tikzpicture}
	\begin{pgfonlayer}{nodelayer}
		\node [style={X dot (zh)}] (0) at (-0.25, 1.5) {};
		\node [style={Z dot (zh)}] (1) at (0.75, 1.5) {$\pi$};
		\node [style=none] (2) at (1.75, 2.25) {};
		\node [style=none] (3) at (1.75, 0.75) {};
		\node [style=none] (4) at (1.75, 1.75) {$\vdots$};
		\node [style={X dot (zh)}] (5) at (4.25, 1.5) {};
		\node [style={Z dot (zh)}] (6) at (5.25, 1.5) {};
		\node [style=none] (7) at (6.25, 2.25) {};
		\node [style=none] (8) at (6.25, 0.75) {};
		\node [style=none] (9) at (6.25, 1.75) {$\vdots$};
		\node [style=none] (10) at (3, 1.5) {$=$};
		\node [style=hbox] (11) at (5.25, 2.5) {};
		\node [style={X dot (zh)}] (12) at (9.75, 2.25) {};
		\node [style=none] (14) at (10.75, 2.25) {};
		\node [style=none] (15) at (10.75, 0.75) {};
		\node [style=none] (16) at (10.75, 1.75) {$\vdots$};
		\node [style=none] (17) at (7.5, 1.75) {$\stackrel{BA_Z}{=}$};
		\node [style=hbox] (18) at (8.75, 2.5) {};
		\node [style={X dot (zh)}] (19) at (8.75, 1.5) {};
		\node [style={X dot (zh)}] (20) at (9.75, 0.75) {};
		\node [style={X dot (zh)}] (21) at (13.75, 2.25) {};
		\node [style=none] (22) at (14.75, 2.25) {};
		\node [style=none] (23) at (14.75, 0.75) {};
		\node [style=none] (24) at (14.75, 1.75) {$\vdots$};
		\node [style=none] (25) at (12, 1.75) {$\stackrel{BA_H}{=}$};
		\node [style={X dot (zh)}] (28) at (13.75, 0.75) {};
		\node [style={X dot (zh)}] (29) at (-0.25, -1.5) {$\pi$};
		\node [style={Z dot (zh)}] (30) at (0.75, -1.5) {$\pi$};
		\node [style=none] (31) at (1.75, -0.75) {};
		\node [style=none] (32) at (1.75, -2.25) {};
		\node [style=none] (33) at (1.75, -1.25) {$\vdots$};
		\node [style={X dot (zh)}] (34) at (4.25, -1.5) {$\pi$};
		\node [style={Z dot (zh)}] (35) at (5.25, -1.5) {};
		\node [style=none] (36) at (6.25, -0.75) {};
		\node [style=none] (37) at (6.25, -2.25) {};
		\node [style=none] (38) at (6.25, -1.25) {$\vdots$};
		\node [style=none] (39) at (3, -1.5) {$=$};
		\node [style=hbox] (40) at (5.25, -0.5) {};
		\node [style={X dot (zh)}] (41) at (9.75, -0.75) {$\pi$};
		\node [style=none] (42) at (10.75, -0.75) {};
		\node [style=none] (43) at (10.75, -2.25) {};
		\node [style=none] (44) at (10.75, -1.25) {$\vdots$};
		\node [style=none] (45) at (7.5, -1.25) {$\stackrel{BA_Z}{=}$};
		\node [style=hbox] (46) at (8.75, -0.5) {};
		\node [style={X dot (zh)}] (47) at (8.75, -1.5) {$\pi$};
		\node [style={X dot (zh)}] (48) at (9.75, -2.25) {$\pi$};
		\node [style={X dot (zh)}] (67) at (13.75, -0.75) {};
		\node [style=none] (68) at (14.75, -0.75) {};
		\node [style=none] (69) at (14.75, -2.25) {};
		\node [style=none] (70) at (14.75, -1.25) {$\vdots$};
		\node [style=none] (71) at (12, -1.5) {$=$};
		\node [style={X dot (zh)}] (72) at (13.75, -2.25) {};
		\node [style=none] (73) at (13, -1.5) {$-$};
	\end{pgfonlayer}
	\begin{pgfonlayer}{edgelayer}
		\draw [bend right=15] (2.center) to (1);
		\draw [bend right=15] (1) to (3.center);
		\draw (1) to (0);
		\draw [bend right=15] (7.center) to (6);
		\draw [bend right=15] (6) to (8.center);
		\draw (6) to (5);
		\draw (11) to (6);
		\draw (19) to (18);
		\draw (12) to (14.center);
		\draw (20) to (15.center);
		\draw (21) to (22.center);
		\draw (28) to (23.center);
		\draw [bend right=15] (31.center) to (30);
		\draw [bend right=15] (30) to (32.center);
		\draw (30) to (29);
		\draw [bend right=15] (36.center) to (35);
		\draw [bend right=15] (35) to (37.center);
		\draw (35) to (34);
		\draw (40) to (35);
		\draw (47) to (46);
		\draw (41) to (42.center);
		\draw (48) to (43.center);
		\draw (67) to (68.center);
		\draw (72) to (69.center);
	\end{pgfonlayer}
\end{tikzpicture}
    \end{equation}
    together with the proof of Lemma \ref{dpll_unitprop_overview}.
\end{proof}

\begin{algorithm}[t]
    \DontPrintSemicolon
    \caption{The $\mathrm{CDP}_\pm$ algorithm solving $\sSAT_\pm$. \label{dpllpm}}
    \KwIn{A CNF formula $f$ with $n$ variables and $m$ clauses, and a set of variables $N \subseteq \{1, \dots, n\}$.}
    \KwOut{The value of $\sSAT_\pm(f, N)$.}

    \uIf{$f$ contains an empty clause}{
        \Return{$0$}
    }\uElseIf{$f$ has no clauses}{
        \Return{$2^n$}
    }\Else{
        Pick $i \in \{1, \dots, n\}$ according to some strategy.\;
        $f_1 \gets \operatorname{Unit-Propagate}(f \land \lnot x_i)$\;
        $f_2 \gets \operatorname{Unit-Propagate}(f \land x_i)$\;
        \eIf{$i \in N$}{
            \Return{$\mathrm{CDP}_\pm(f_1, N) - \mathrm{CDP}_\pm(f_2, N)$}
        }{
            \Return{$\mathrm{CDP}_\pm(f_1, N) + \mathrm{CDP}_\pm(f_2, N)$}
        }
    }
\end{algorithm}

With a smart choice of the variables to branch on, the worst-case runtime of Algorithm \ref{dpllpm} can be bounded in exactly the same way as the regular $\mathrm{CDP}$ algorithm and its variants. This is because the bounds we consider here (e.g \cite{kutzkov_new_2007, furer_algorithms_2007, wahlstrom_tighter_2008, wang_worst_2013}) are obtained from two principles: 
\begin{itemize}
    \item Exploiting the fact that the conjunction of unrelated \sSAT instances combine multiplicatively. Diagrammatically, this corresponds to the parallel composition of scalar diagrams being defined as scalar multiplication. This remains true for the case when phases are present on Z-spiders.
    \item By analyzing which classes of sub-formulas can occur in any instance, and how they are affected by the unit propagation and branching rules. Diagrammatically, this corresponds to a case analysis on the possible subgraphs of the diagram, ignoring scalar factors. Since unit propagation and branching are unaffected (except for scalar factors) by the presence of phases on the Z-spiders (see Lemmas \ref{dpll-neg-branch} and \ref{dpll-neg-unitprop}), this also applies directly in this case.
\end{itemize}

In particular, these bounds do not depend on the specific scalar factors of each diagram, or the way that separate diagrams generated by branching are recombined. This means that the $O^*(1.2377^{\mathrm{variables}})$ bound of Wahlstr\"om \cite{wahlstrom_tighter_2008} can be adapted directly. By applying Lemma \ref{dpll_neg_kclause} to any $\sSAT$ diagram and then applying $\mathrm{CDP}_\pm$ to the resulting diagram directly, we can evaluate $\sSAT$ instances in time $O^*(1.2377^{n + m})$ $= O^*(2^{0.3068n + 0.3068m})$, which is certainly better than the bound given by decomposing into a sum of diagrams. We will refer to this method, given in Algorithm \ref{cdp2pm}, as $\mathrm{CDP}_\pm^{\to 2}$.

\begin{algorithm}[t]
    \DontPrintSemicolon
    \caption{The $\mathrm{CDP}^{\to 2}_\pm$ algorithm for $\sSAT$. \label{cdp2pm}}
    \KwIn{A CNF formula $f$ with $n$ variables and $m$ clauses.}
    \KwOut{The value of $\sSAT(f)$.}

    Generate $f'$ by applying Lemma \ref{dpll_neg_kclause} to every clause in $f$ of width at least three.\;
    Set $N$ to be all the variables labeled with $\pi$-phases.\;
    \Return{$\mathrm{CDP}_\pm(f', N)$}\;
\end{algorithm}

\begin{theorem}[Restatement of Theorem \ref{ssatmbound}]
    Given a CNF formula $f:\B^n\to \B$, we can count the number of satisfying assignments $\#\{\vec x\in \B^n\,|\, f(\vec x) = 1\}$ in time $O^*(1.2377^{n + m_{\geq 3}})$, where $m_{\geq 3}$ is the number of clauses of width at least three.
\end{theorem}
\begin{proof}
    Apply the algorithm $\mathrm{CDP}_\pm^{\to 2}$ to $f$ using Wahlstr\"om's \cite{wahlstrom_tighter_2008} $O^*(1.2377^{\mathrm{variables}})$ algorithm for solving $\stwoSAT$. Then $m_{\geq 3}$ new (negative) variables will be created by applying Lemma \ref{dpll_neg_kclause}, so the overall runtime is given by $O^*(1.2377^{n + m_{\geq 3}})$.
\end{proof}

It remains to ask, when is $\mathrm{CDP}_\pm^{\to 2}$ actually useful? First, note that if only positive variables are picked for branching, the action of $\mathrm{CDP}_\pm$ on a translated $\sSAT$ diagram is \emph{exactly} the same as the action of regular $\mathrm{CDP}$ on the original diagram. Therefore, we would only expect gains when decomposing some of the negative variables (i.e clauses), and thus it is natural to suspect that this bound will only be useful for instances with few clauses.

\subsection{Complexity analysis}
For instances with a fixed maximum density $\delta_{max}$, and assuming the worst-case of $m_{\geq 3} = m$, we have that $m \leq n\delta_{max}$, and so the runtime of $\mathrm{CDP}_\pm^{\to 2}$ is bounded by $O^*(2^{0.3068(1 + \delta_{max})n})$. Firstly, we can see that this is better than the naive $O^*(2^n)$ whenever $\delta_{max} < 2.2503$. Since this is independent of $k$, it means that for any $\delta_{max} < 2.2503$ and sufficiently large $k$, this beats both the worst-case bound of Dubois \cite{dubois_counting_1991} and the average-case bound of Williams \cite{williams_computing_2004}. 

Concretely, $\mathrm{CDP}_\pm^{\to 2}$ is better than the average-case bounds of Williams \cite{williams_computing_2004} whenever $\delta_{max} < 1.217$ and $k \geq 6$, and better than the worst-case bounds of Dubois \cite{dubois_counting_1991} whenever $k \geq 3$ and $\delta_{max} < 1.858$. The exact bounds for each $k$ are given in Table \ref{density_comparison}. Clearly, $\mathrm{CDP}_\pm^{\to 2}$ offers no improvement on $\stwoSAT$, but it is also not directly applicable to $\#\textbf{3SAT}$ - when $\delta < 1.6$, the $O^*(1.4142^{m})$ bound of Zhou et al \cite{zhou_new_2010} is sharper, and when $\delta \geq 1.6$ the $O^*(1.6423^n)$ bound of Kutzkov \cite{kutzkov_new_2007} is sharper.  

For $\SAT$, it has been shown that there is a phase transition in the satisfiability of a random formula as $\delta$ passes some threshold. Instances with densities near this threshold are known to be hard to solve, and it is known that this threshold scales exponentially with $k$ \cite{achlioptas_rigorous_2005}. However, no bound is known for the equivalent `hardest' density in $\sSAT$. Assuming SETH, our result indicates that the `hardest' density of $\sSAT$ must be some $\delta > 2.2503$, since otherwise $\sSAT$ (and hence $\SAT$) could be solved in time better than $O^*(2^n)$.

\subsection{A non-diagrammatic argument for $\mathrm{CDP}_\pm^{\to 2}$}

While we originally found this reduction and algorithm using the ZH-calculus and prefer its diagrammatic presentation, in Appendix \ref{app:nodiagrams} we present a self-contained argument for the reduction from \sSAT to $\stwoSAT_\pm$ without any diagrams that may be more intuitive to readers unfamiliar with graphical calculi.

\section{Variations on the main result}\label{sec:variations}

\subsection{Bounding $\sSAT$ in terms of literals}\label{sec:literal-bound}

Note that $\mathrm{CDP}_\pm^{\to 2}$ maps a $\sSAT$ instance of $m$ clauses to a $\stwoSAT$ instance of at most $L$ clauses (with negative variables), where $L$ is the number of literals in the original instance. Therefore, applying the upper bound of $O^*(1.1740^{m})$ on $\stwoSAT$ found by Wang and Gu \cite{wang_worst_2013}, we can bound the runtime of $\mathrm{CDP}_\pm^{\to 2}$ in terms of literals as $O^*(1.1740^{L})$. This implies a better than $O^*(2^n)$ runtime whenever the average degree $\hat{d} = \frac{L}{n}$ of variables in an instance satisfies $\hat{d} < 4.3209$, or the maximal number of clauses $d$ a variable participates in is $d \leq 4$. As far as we are aware this is the first bound of this type for $\sSAT$ with unrestricted clause width, but similar bounds are known for $\SAT$ - e.g the $O^*(1.0646^L)$ of Peng and Xiao \cite{peng_further_2021}. If $k$ is small, then better bounds for $\sSAT$ follow trivially from bounds in terms of $m$ - e.g for $k = 2$, $m \leq \frac{L}{2}$, so \cite{wang_worst_2013} implies $O^*(1.0835^L)$.

\subsection{An algorithm for low-density $\#\textbf{3SAT}$ instances}\label{low3sat}

In the previous section, we noted that $\mathrm{CDP}^{\to 2}_\pm$ can't beat existing bounds for $\#\textbf{3SAT}$ on its own. This is because we have to assume that $m = m_{\geq 3}$ in the worst-case. However, we can use a technique introduced by Kutzkov \cite{kutzkov_new_2007} to take advantage of the extra structure afforded by $\#\textbf{3SAT}$ and introduce extra branching steps which allow us to assume that $m_{\geq 3} < m$.

Assume we have some $\#\textbf{3SAT}$ instance $f$. Every time we branch on a variable $x$ that occurs in $d$ 3-clauses in $f$, in both branches at least $d$ 3-clauses are eliminated (since each clause will be totally removed in one branch and become a 2-clause in the other). If we know that the density of 3-clauses in $f$ is $\delta_3$, then the average number of 3-clauses a variable is connected to (its average 3-degree) is $3\delta_3$. Therefore, whenever $d - 1 < 3\delta_3 \leq d$, there must exist a variable with 3-degree at least $d$, and so by branching on the variable with the highest 3-degree, we remove at least $d$ 3-clauses. 

Following Kutzkov \cite{kutzkov_new_2007}, suppose $x$ is the number of variables needed to reduce $\delta_3$ to at most $\frac{d - 1}{3}$ by repeatedly branching on the variable with the highest 3-degree. Then we need
\begin{equation}
    \frac{d}{3}n - d x \leq \frac{d - 1}{3}(n - x)
\end{equation}
thus $x \leq \frac{n}{2d + 1}$, so in the limit of large $n$, we have $x = \frac{n}{2d + 1}$ in the worst case, and $n - \frac{n}{2d + 1}$ variables remain unassigned. Therefore, the number of variables we need to branch on to reduce $\delta_3$ to at most $\frac{2}{3}$ is:
\begin{equation}
    \label{n23def}
    n_{2/3} = n - n \prod_{i = 3}^d\left(1 - \frac{1}{2i + 1}\right)
\end{equation}
Since $m_{\geq 3} = \delta_3n$, after performing this branching on $n_{2/3}$ variables, we will have $2^{n_{2/3}}$ instances, each with $m_{\geq 3} \leq \frac{2}{3}(n - n_{2/3})$, so these can be evaluated with $\mathrm{CDP}_\pm^{\to 2}$. This strategy, formalized as Algorithm \ref{cdplow3}, therefore has an overall time bound of:
\begin{equation}
    \begin{aligned}
        O^*(2^{n_{2/3}}) &O^*(1.2377^{(n - n_{2/3}) + m_{\geq 3}}) \leq O^*(2^{n_{2/3}}1.2377^{(n - n_{2/3})(1 + \frac{2}{3})}) = O^*(2^{0.5128n + 0.4872n_{2/3}})
    \end{aligned}
\end{equation}
In order to calculate the running-time bound for a given maximum $\delta_3$, we can plug $d = \lceil3\delta_3\rceil$ into Equation \eqref{n23def} to calculate $n_{2/3}$ as a fraction of $n$. For example, if $\delta_3 < \frac{5}{3}$ then $d = 5$ and $n_{2/3} = 0.3074n$, yielding a time of 
$$O^*(2^{(0.5128 + 0.4872\cdot 0.3074)n}) = O^*(1.5829^n).$$
Suppose then that $\delta = \delta_3$ (i.e the worst-case), then this bound is better than the bound of Zhou \cite{zhou_new_2010} whenever $\delta > 1.2577$ (i.e $d \geq 4$ but not $d = 3$, comparing the $O^*(1.5463^n)$ complexity for $d = 4$ to Zhou's $O^*(1.4142^m)$ to find the exact cutoff point) and the $O^*(1.6423^n)$ bound of Kutzkov \cite{kutzkov_new_2007} whenever $\delta \leq \frac{7}{3}$ (i.e for $d \leq 7$), yielding complexities of $O^*(1.5463^n)$ ($d = 4$) to $O^*(1.6350^n)$ ($d = 7$) respectively. It is possible this bound could be extended to a yield an improved bound on general $\#\textbf{3SAT}$ using a case analysis similar to Kutzkov's, but this is quite complicated so we postpone exploring this to future work.

\begin{algorithm}[t]
    \DontPrintSemicolon
    \caption{The algorithm $\mathrm{CDP}_\pm^{3 \to 2}$ for solving $\#\textbf{3SAT}$. \label{cdplow3}}
    \KwIn{A CNF formula $f$ with $n$ variables and $m$ clauses.}
    \KwOut{The value of $\sSAT(f)$.}
    \eIf{$\delta_3 > \frac{2}{3}$}{
        Pick $x$ in $f$ with maximal 3-degree.\;
        $f_1 \gets \operatorname{Unit-Propagate}(f \land \lnot x_i)$\;
        $f_2 \gets \operatorname{Unit-Propagate}(f \land x_i)$\;
        \Return{$\mathrm{CDP}_\pm^{3 \to 2}(f_1) + \mathrm{CDP}_\pm^{3 \to 2}(f_2)$}
    }{
        \Return{$\mathrm{CDP}^{\to 2}_\pm(f)$}\;
    }
\end{algorithm}

\subsection{Solving \sSAT with arbitrary phases}\label{sec:SAT-phases}

While allowing $\pi$-phases on variables has allowed us to find a simple reduction from $\sSAT$ to $\stwoSAT_\pm$ which can be solved with $\mathrm{CDP}_\pm$, the CDP algorithm easily extends to arbitrary phases in the same way. Indeed let us define a generalization of the $\sSAT$ problem, $\sSAT_\mathcal{A}$ - this is exactly the weighted model counting problem where the weights are restricted to $\mathcal{A}$.

\begin{definition}
    The problem $\sSAT_\mathcal{A}$ for $\mathcal{A} \subseteq \mathbb{C} \setminus \{0\}$ is defined as follows. Given a CNF formula $f(x_1, \dots)$ with $n$ variables and $m$ clauses and a vector $A \in \mathcal{A}^n$, compute the quantity:
    \begin{equation}
        \sSAT_\mathcal{A}(f, A) = \sum_{\vec{x} \in \mathbb{B}^n} \left(\prod_{i = 1}^n A_i^{x_i}\right) f(\vec{x})
    \end{equation}
\end{definition}

It is easy to see then that $\sSAT = \sSAT_{\{1\}}$, and $\sSAT_\pm = \sSAT_{\{1, -1\}}$: let $A_j = -1$ if $j \in N$ and $A_j = 1$ otherwise. 
A straightforward adaptation of the CDP algorithm can solve $\sSAT_\mathcal{A}$; see Algorithm~\ref{cdpphase}.

The complex numbers in $\mathcal{A}$ on the variables of an instance can be easily represented in the ZH-calculus, by generalizing Eq.~\eqref{eq:zspiderphase}.
The variable branching and unit propagation rules then generalize from Eq.~\eqref{zspidersum}. The advantage of working with $\sSAT_\mathcal{A}$ is that by expanding $\mathcal{A}$, we are afforded additional rewriting rules on the corresponding ZH-diagrams. Moving from $\{1\}$ to $\{1, -1\}$ allowed us to rewrite arbitrary arity zero-labeled H-boxes into arity two H-boxes. Further expanding this to $\{\frac{k}{2} \mid k \in \mathbb{Z}\}$ allows us the following rule, removing (up to a scalar) any variables that only occur once in a formula:
\begin{equation}
    \label{extension_absorb}
    \begin{tikzpicture}
	\begin{pgfonlayer}{nodelayer}
		\node [style={Z phase dot (zh)}] (0) at (0, 1.5) {$\alpha$};
		\node [style=labelled hbox] (1) at (1.25, 1.5) {$0$};
		\node [style={Z dot (zh)}] (2) at (2.25, 1.5) {};
		\node [style=none] (3) at (-1, 2.25) {};
		\node [style=none] (4) at (-1, 0.75) {};
		\node [style=none] (5) at (-1, 1.75) {$\vdots$};
		\node [style=none] (13) at (3.75, 1.5) {$=$};
		\node [style={Z phase dot (zh)}] (14) at (5.75, 1.5) {$\alpha$};
		\node [style=labelled hbox] (15) at (7, 1.5) {$0$};
		\node [style={Z dot (zh)}] (16) at (8.75, 1.5) {};
		\node [style=none] (17) at (4.75, 2.25) {};
		\node [style=none] (18) at (4.75, 0.75) {};
		\node [style=none] (19) at (4.75, 1.75) {$\vdots$};
		\node [style={Z phase dot (zh)}] (20) at (13.25, 1.5) {$\alpha + i\ln(2)$};
		\node [style=none] (21) at (11.5, 2.25) {};
		\node [style=none] (22) at (11.5, 0.75) {};
		\node [style=none] (23) at (11.5, 1.75) {$\vdots$};
		\node [style=none] (24) at (11, 1.5) {$2$};
		\node [style=none] (25) at (10.25, 1.5) {$=$};
		\node [style={X dot (zh)}] (26) at (8, 1.5) {$\pi$};
		\node [style={Z phase dot (zh)}] (27) at (0, -1.5) {$\alpha$};
		\node [style=labelled hbox] (28) at (2, -1.5) {$0$};
		\node [style={Z dot (zh)}] (29) at (3, -1.5) {};
		\node [style=none] (30) at (-1, -0.75) {};
		\node [style=none] (31) at (-1, -2.25) {};
		\node [style=none] (32) at (-1, -1.25) {$\vdots$};
		\node [style=none] (33) at (3.75, -1.5) {$=$};
		\node [style={Z phase dot (zh)}] (34) at (5.75, -1.5) {$\alpha$};
		\node [style=labelled hbox] (35) at (7.75, -1.5) {$0$};
		\node [style={Z dot (zh)}] (36) at (9.5, -1.5) {};
		\node [style=none] (37) at (4.75, -0.75) {};
		\node [style=none] (38) at (4.75, -2.25) {};
		\node [style=none] (39) at (4.75, -1.25) {$\vdots$};
		\node [style={Z phase dot (zh)}] (40) at (13.25, -1.5) {$\alpha - i\ln(2)$};
		\node [style=none] (41) at (11.5, -0.75) {};
		\node [style=none] (42) at (11.5, -2.25) {};
		\node [style=none] (43) at (11.5, -1.25) {$\vdots$};
		\node [style=none] (45) at (10.25, -1.5) {$=$};
		\node [style={X dot (zh)}] (46) at (8.75, -1.5) {$\pi$};
		\node [style={X dot (zh)}] (47) at (1, -1.5) {$\pi$};
		\node [style={X dot (zh)}] (48) at (6.75, -1.5) {$\pi$};
	\end{pgfonlayer}
	\begin{pgfonlayer}{edgelayer}
		\draw (2) to (1);
		\draw (1) to (0);
		\draw [bend right=15] (4.center) to (0);
		\draw [bend right=15] (0) to (3.center);
		\draw (15) to (14);
		\draw [bend right=15] (18.center) to (14);
		\draw [bend right=15] (14) to (17.center);
		\draw [bend right=15] (22.center) to (20);
		\draw [bend right=15] (20) to (21.center);
		\draw (16) to (26);
		\draw (26) to (15);
		\draw (29) to (28);
		\draw [bend right=15] (31.center) to (27);
		\draw [bend right=15] (27) to (30.center);
		\draw [bend right=15] (38.center) to (34);
		\draw [bend right=15] (34) to (37.center);
		\draw [bend right=15] (42.center) to (40);
		\draw [bend right=15] (40) to (41.center);
		\draw (36) to (46);
		\draw (46) to (35);
		\draw (28) to (47);
		\draw (47) to (27);
		\draw (48) to (35);
		\draw (48) to (34);
	\end{pgfonlayer}
\end{tikzpicture}
\end{equation}
This is a limited form of pure-literal elimination, a rewrite rule that is usually only valid in $\SAT$ and not $\sSAT$. Applying this simplification to the formula recursively (after unit propagation in the CDP algorithm, i.e between lines 7 and 8 of Algorithm \ref{cdp_algo}), we may assume that every variable has degree at least two. This would allow improvement on early bounds such as \cite{dubois_counting_1991}, while being much simpler. Therefore, an interesting avenue for further research would be investigating how this approach of weighting variables could be used to simplify $\sSAT$ instances or find upper bounds on runtime. Another example is given in Appendix~\ref{sec:circuit-sim} where we show how this leads to an algorithm for simulating quantum circuits with runtime in terms of the total number of gates.

\setlength{\textfloatsep}{1.5mm}
\begin{algorithm}[t]
    \DontPrintSemicolon
    \caption{The $\mathrm{CDP}_\mathcal{A}$ algorithm for solving $\sSAT_\mathcal{A}$. \label{cdpphase}}
    \KwIn{A CNF formula $f$ with $n$ variables and $m$ clauses, and $A \in \mathcal{A}^n$.}
    \KwOut{The value of $\sSAT_\mathcal{A}(f, A)$.}

    \uIf{$f$ contains an empty clause}{
        \Return{$0$}
    }\uElseIf{$f$ has no clauses}{
        \Return{$2^n$}
    }\Else{
        Pick $i \in \{1, \dots, n\}$ according to some strategy.\;
        $f_1 \gets \operatorname{Unit-Propagate}(f \land \lnot x_i)$\;
        $f_2 \gets \operatorname{Unit-Propagate}(f \land x_i)$\;
        \Return{$\mathrm{CDP}_\mathcal{A}(f_1, A) + A_i\mathrm{CDP}_\mathcal{A}(f_2, A)$}
    }
\end{algorithm}
\vspace{-1mm}

\section{Conclusion}\label{sec:conclusion}

In this paper, we used the ZH-calculus to study the $\sSAT$ problem and produced an upper bound which does not depend on the clause width $k$. We believe bounds of this kind were previously only known for the decision variant $\SAT$ \cite{yamamoto_improved_2005-1}. The bound is less than $O^*(2^n)$ whenever the clause density $\delta=\frac{n}{m}$ is smaller than $2.2503$, suggesting that the `hardest' density of $\sSAT$ problems must be some $\delta > 2.2503$, assuming the strong exponential time hypothesis. We found these bounds by rephrasing the $\sSAT$ problem in terms of ZH-diagrams, and generalizing known rewrite rules to give a reduction from $\sSAT$ to $\stwoSAT_\pm$, a weighted variant of $\sSAT$ that can be solved with Wahlstr\"om's \cite{wahlstrom_tighter_2008} variant of the CDP algorithm \cite{birnbaum_good_1999}. 
Using a more involved analysis and algorithm we also improved on the upper bound for $\#\textbf{3SAT}$ for $1.2577 < \delta < \frac{7}{3}$.
In addition, using a result of Wang and Gu~\cite{wang_worst_2013}, we produced an explicit bound of $O^*(1.1740^L)$ for $\sSAT$ in terms of the number of literals $L$, to our knowledge the first such non-trivial bound for $\sSAT$. A summary of all the bounds obtained in this paper is presented in Table~\ref{boundtable}. We suggest extending this technique of reducing $\sSAT$ to weighted $\sSAT$ as an avenue of future research.

Our results show that graphical calculi can lead to concrete algorithmic improvements in areas where significant research has already been done, even when originally intended for a different domain like quantum computing. An interesting question then is in which other domains we can make improvements by framing the problem using graphical reasoning.

\section*{Acknowledgments} We thank Matty Hoban and the anonymous QPL reviewers for helpful feedback. The majority of this work was done while TL was a student at the University of Oxford, and the main result is also presented in an MSc thesis~\cite{laakkonen_graphical_2022}. JvdW acknowledges funding from the European Union's Horizon 2020 research and innovation programme under the Marie Sk{\l}odowska-Curie grant agreement No 101018390.


%
%
%
%

\bibliographystyle{eptcs}
\bibliography{bibliography}

\end{document}